\RequirePackage[l2tabu,orthodox]{nag}

\documentclass[headsepline,footsepline,footinclude=false,fontsize=11pt,paper=a4,listof=totoc,bibliography=totoc,BCOR=12mm,DIV=12]{scrbook} 

\PassOptionsToPackage{table,svgnames,dvipsnames}{xcolor}

\usepackage[utf8]{inputenc}
\usepackage[T1]{fontenc}
\usepackage[sc]{mathpazo}
\usepackage[ngerman, american]{babel}
\usepackage[autostyle]{csquotes}
\usepackage[%
  backend=biber,
  url=false,
  sorting=anyt,
  style=alphabetic,
  maxnames=4,
  minnames=2,
  maxbibnames=99,
  giveninits,
  uniquename=init]{biblatex}
\usepackage{graphicx}
\usepackage{scrhack} 
\usepackage{listings}
\usepackage{lstautogobble}
\usepackage{tikz}
\usetikzlibrary{patterns}
\usetikzlibrary{intersections}
\usepackage{pgfplots}
\usepgfplotslibrary{fillbetween}
\usepgfplotslibrary{dateplot}
\usepackage{pgfplotstable}
\usepackage{booktabs}
\usepackage[final]{microtype}
\usepackage{caption}
\usepackage{amsmath}
\usepackage{mathtools}
\usepackage{amsthm,thmtools}
\usepackage[nottoc]{tocbibind}
\usepackage[ruled]{algorithm2e}
\usepackage{enumerate}
\usepackage{tabularx}
\usepackage{imakeidx}
\usepackage[italic]{esdiff}
\usepackage{subcaption}
\usepackage{ltablex}
\usepackage{pdflscape}
\usepackage[hidelinks]{hyperref} 
\usepackage[nameinlink]{cleveref}

\crefname{chapter}{Chapter}{Chapters}
\crefname{section}{Section}{Sections}
\crefname{subsection}{Subsection}{Subsections}
\crefname{subsubsection}{Subsubsection}{Subsubsections}
\crefname{algorithm}{Algorithm}{Algorithms}
\crefname{figure}{Figure}{Figures}
\crefname{table}{Table}{Tables}
\crefname{equation}{equation}{equations}

\keepXColumns

\usepackage[stable]{footmisc}

\DeclareMathOperator*{\argmin}{arg\,min}

\indexsetup{headers={\indexname}{\indexname}}
\makeindex[intoc]

\DeclarePairedDelimiter{\norm}{\lVert}{\rVert}

\newcounter{common}
\declaretheorem[name=Theorem,sibling=common]{theorem}
\declaretheorem[name=Lemma,sibling=theorem]{lemma}
\declaretheorem[name=Definition,sibling=theorem]{definition}
\declaretheorem[name=Problem,sibling=theorem]{problem}
\makeatletter
\def\ll@problem{
  \protect\numberline{\theproblem}\thmt@shortoptarg
}
\makeatother
\makeatletter
\let\c@algocf\relax 
\makeatother
\usepackage{aliascnt}
\newaliascnt{algocf}{common} 

\bibliography{sources}

\setkomafont{disposition}{\normalfont\bfseries} 
\BeforeTOCHead[toc]{{\cleardoublepage\pdfbookmark[0]{\contentsname}{toc}}}

\definecolor{TUMBlue}{HTML}{0065BD}
\definecolor{TUMSecondaryBlue}{HTML}{005293}
\definecolor{TUMSecondaryBlue2}{HTML}{003359}
\definecolor{TUMBlack}{HTML}{000000}
\definecolor{TUMWhite}{HTML}{FFFFFF}
\definecolor{TUMDarkGray}{HTML}{333333}
\definecolor{TUMGray}{HTML}{808080}
\definecolor{TUMLightGray}{HTML}{CCCCC6}
\definecolor{TUMAccentGray}{HTML}{DAD7CB}
\definecolor{TUMAccentOrange}{HTML}{E37222}
\definecolor{TUMAccentGreen}{HTML}{A2AD00}
\definecolor{TUMAccentLightBlue}{HTML}{98C6EA}
\definecolor{TUMAccentBlue}{HTML}{64A0C8}

\pgfplotsset{compat=newest}
\pgfplotsset{
  cycle list={TUMBlue\\TUMAccentOrange\\TUMAccentGreen\\TUMSecondaryBlue2\\TUMDarkGray\\},
}

\def\b{\textcolor{blue}}

\SetKw{KwBreak}{break}
\SetKw{KwOr}{or}

\makeatletter
\long\def\algocf@caption@proc#1[#2]#3{%
  \ifthenelse{\boolean{algocf@nokwfunc}}{\relax}{%
    \SetKwFunction{\algocf@captname#3@}{\algocf@captname#3@}%
  }%
  \ifthenelse{\boolean{algocf@func}}{\def\@proc@func{algocffunc}}{\def\@proc@func{algocfproc}}%
  \@ifundefined{hyper@refstepcounter}{\relax}{
    \ifthenelse{\boolean{algocf@procnumbered}}{%
      \expandafter\def\csname theH\@proc@func\endcsname{\algocf@captname#3@}
    }{%
      \expandafter\def\csname theH\@proc@func\endcsname{\algocf@captname#3@}
    }%
    \hyper@refstepcounter{\@proc@func}%
  }%
  \ifthenelse{\boolean{algocf@procnumbered}}{\relax}{%
    \addtocounter{algocf}{-1}
    \gdef\@currentlabel{\algocf@captname#3@}
  }%
  \renewcommand{\addcontentsline}[3]{}
  \ifthenelse{\equal{\algocf@captparam#2@}{\arg@e}}{
    \algocf@latexcaption{#1}[\algocf@captname#2@]{#3}%
  }{
    \algocf@latexcaption{#1}[#2]{#3}%
  }%
}%
\makeatother

\makeatletter
\let\original@algocf@latexcaption\algocf@latexcaption
\long\def\algocf@latexcaption#1[#2]{%
  \@ifundefined{NR@gettitle}{%
    \def\@currentlabelname{#2}%
  }{%
    \NR@gettitle{#2}%
  }%
  \original@algocf@latexcaption{#1}[{#2}]%
}
\makeatother

\addto\captionsamerican{%
  }

\usetikzlibrary{external}

\newcommand*{\getUniversity}{Technical University of Munich}
\newcommand*{\getFaculty}{Department of Informatics}
\newcommand*{\getTitle}{Implementation of Algorithms for Right-Sizing Data Centers}
\newcommand*{\getTitleGer}{Implementierung von Algorithmen für die Größenanpassung von Rechenzentren}
\newcommand*{\getAuthor}{Jonas Hübotter}
\newcommand*{\getDoctype}{Bachelor's Thesis in Informatics}
\newcommand*{\getSupervisor}{Prof. Dr. Susanne Albers}
\newcommand*{\getAdvisor}{Jens Quedenfeld}
\newcommand*{\getSubmissionDate}{August 13, 2021}
\newcommand*{\getSubmissionLocation}{Munich}

\begin{document}

\pagenumbering{alph}
\begin{titlepage}
  \oddsidemargin=\evensidemargin\relax
  \textwidth=\dimexpr\paperwidth-2\evensidemargin-2in\relax
  \hsize=\textwidth\relax

  \centering
  \null
  \vspace{18mm}
  \IfFileExists{thesis/logos/tum.pdf}{%
    \includegraphics[height=20mm]{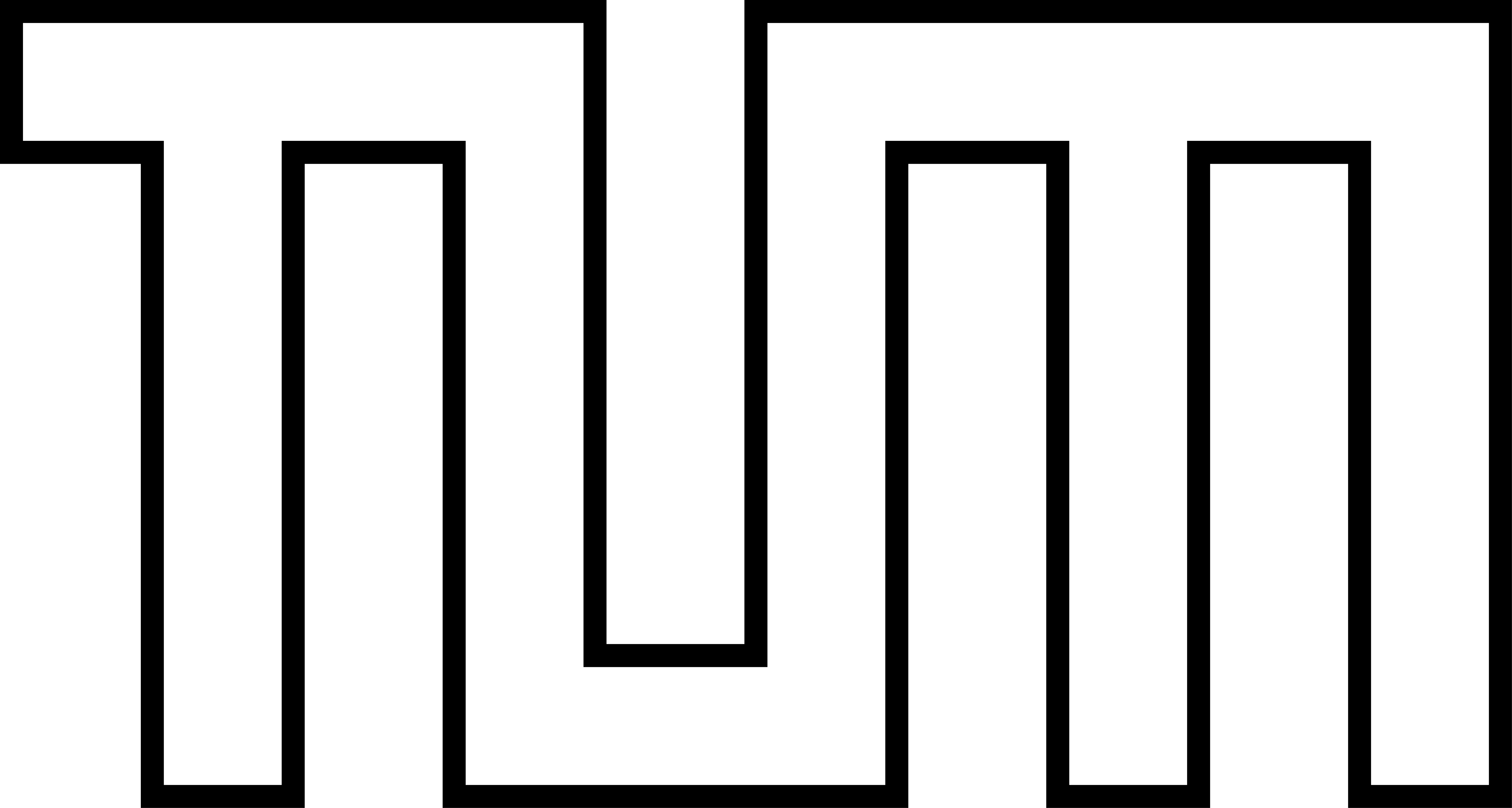}
  }{%
    \vspace*{20mm}
  }

  \vspace{10mm}
  {\huge\MakeUppercase{\getFaculty{}}}\\

  \vspace{5mm}
  {\large\MakeUppercase{\getUniversity{}}}\\

  \vspace{20mm}
  {\Large \getDoctype{}}

  \vspace{15mm}
  {\huge\bfseries \getTitle{}}

  \vspace{15mm}
  {\LARGE \getAuthor{}}

  \IfFileExists{thesis/logos/faculty.pdf}{%
    \vfill{}
    \includegraphics[height=20mm]{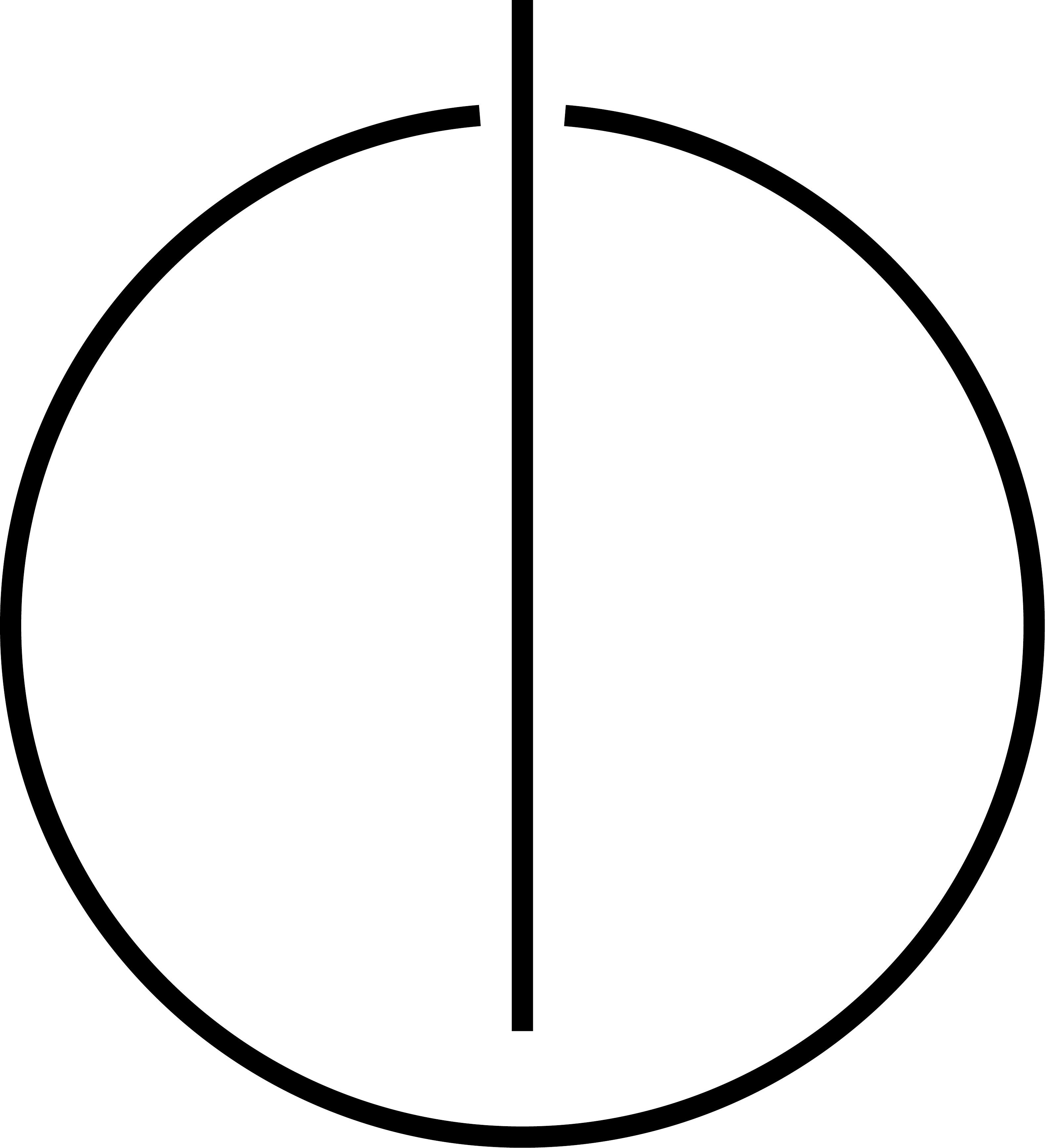}
  }{}
\end{titlepage}

\frontmatter{}

\begin{titlepage}
  \centering

  \IfFileExists{thesis/logos/tum.pdf}{%
    \includegraphics[height=20mm]{thesis/logos/tum.pdf}
  }{%
    \vspace*{20mm}
  }

  \vspace{10mm}
  {\huge\MakeUppercase{\getFaculty{}}}\\

  \vspace{5mm}
  {\large\MakeUppercase{\getUniversity{}}}\\

  \vspace{20mm}
  {\Large \getDoctype{}}

  \vspace{15mm}
  {\huge\bfseries \getTitle{}}

  \vspace{10mm}
  {\huge\bfseries \foreignlanguage{ngerman}{\getTitleGer{}}}

  \vspace{15mm}
  \begin{tabular}{l l}
    Author:          & \getAuthor{} \\
    Supervisor:      & \getSupervisor{} \\
    Advisor:         & \getAdvisor{} \\
    Submission Date: & \getSubmissionDate{} \\
  \end{tabular}
  
  \vspace{6mm}
  \IfFileExists{thesis/logos/faculty.pdf}{%
    \vfill{}
    \includegraphics[height=20mm]{thesis/logos/faculty.pdf}
  }{}
\end{titlepage}
\cleardoublepage{}
\thispagestyle{empty}
\vspace*{0.8\textheight}
\noindent
I confirm that this bachelor's thesis is my own work and I have documented all sources and material used.

\vspace{15mm}
\noindent
\getSubmissionLocation{}, \getSubmissionDate{} \hspace{50mm} \getAuthor{}

\cleardoublepage{}

\addcontentsline{toc}{chapter}{Acknowledgments}
\thispagestyle{empty}

\vspace*{20mm}

\begin{center}
{\usekomafont{section} Acknowledgments}
\end{center}

\vspace{10mm}

I would like to thank my supervisor Prof. Dr. Susanne Albers, for allowing me to work on this exciting topic. I would like to thank my advisor Jens Quedenfeld for his patient work on various drafts, checking my proofs, and his help whenever I felt stuck. I would like to thank the maintainer of pyo3, David Hewitt, for his great help with debugging the Python bindings. Especially, I would like to thank my friends and my family for their support and stimulating discussions on this and other topics.

\cleardoublepage{}

\chapter{\abstractname}

The energy consumption of data centers assumes a significant fraction of the world's overall energy consumption. Most data centers are statically provisioned, leading to a very low average utilization of servers. In this work, we survey uni-dimensional and high-dimensional approaches for dynamically powering up and powering down servers to reduce the energy footprint of data centers while ensuring that incoming jobs are processed in time. We implement algorithms for smoothed online convex optimization and variations thereof where, in each round, the agent receives a convex cost function. The agent seeks to balance minimizing this cost and a movement cost associated with changing decisions in-between rounds. We implement the algorithms in their most general form, inviting future research on their performance in other application areas. We evaluate the algorithms for the application of right-sizing data centers using traces from Facebook, Microsoft, Alibaba, and Los Alamos National Lab. Our experiments show that the online algorithms perform close to the dynamic offline optimum in practice and promise a significant cost reduction compared to a static provisioning of servers. We discuss how features of the data center model and trace impact the performance. Finally, we investigate the practical use of predictions to achieve further cost reductions.
\microtypesetup{protrusion=false}
\tableofcontents{}
\microtypesetup{protrusion=true}

\mainmatter{}


\chapter{Introduction}\label{chapter:introduction}

\section{Motivation}

The energy use of data centers is estimated to account for between 1\% and 3\% of the global electricity supply and up to 1\% of global greenhouse gas emissions~\cite{Shehabi2016, Jones2018, Bashroush2020, Masanet2020}. It is also estimated that energy use will dramatically increase over the next decades unless there are significant advances in energy efficiency~\cite{Jones2018}. However, energy conservation in data centers is not just important for ecological reasons. As the energy cost of data centers represents a substantial fraction of its overall expenses, increasing a data center's energy efficiency can also significantly reduce costs~\cite{Barroso2007, Brill2007, Hamilton2008}.

In practice, the energy consumption of many data centers is largely decoupled from their load. Often, \emph{peak-provisioning}\index{peak-provisioning} is used in the design of data centers to prevent a lack of resources during periods of high loads~\cite{Whitney2014}. The Natural Resources Defense Council (NRDC) found that due to this over-provisioning and the absence of dynamic right-sizing, servers typically operate at utilization levels between 12\% and 18\%~\cite{Whitney2014}. Other studies still find an average utilization between 20\% and 40\%~\cite{Barroso2007, Armbrust2010}.

It was shown that servers reach their peak efficiency at full utilization~\cite{Barroso2007}. At the typical utilization levels, a server's energy efficiency is between 20\% and 30\%~\cite{Barroso2007}. Even when idling, servers consume half of their peak power~\cite{Barroso2007}. Therefore, an important research goal has been to couple the energy consumption of a data center with its workload. In other words, to achieve \emph{power-proportionality}\index{power-proportionality} of a data center. We investigate the approach to dynamically power-up and power-down servers as the load of a data center or a network of data centers changes. Solutions to this approach are also known as \emph{power-down mechanisms}\index{power-down mechanisms}~\cite{Jin2016}.

It is natural to model this problem as an optimization minimizing some cost metric reflecting a data center's energy consumption. Commonly, two variations of such optimization problems are considered. The \emph{online}\index{online algorithm} variant receives a sequence of requests and performs an immediate action in response to each request. In contrast, the \emph{offline}\index{offline algorithm} variant receives all requests upfront and only responds once~\cite{Karp1992}. It is easy to see that any solution to an online problem is at best as good as a solution to its offline variant.

\paragraph{Smoothed Online Convex Optimization} The models that we investigate are generalizations of online convex optimization. In \emph{online convex optimization}\index{online convex optimization}, an agent interacts with their environment in a sequence of rounds. In each round, the agent is presented with a convex cost function (in the following called \emph{hitting cost}\index{hitting cost}). Based on this cost function, the agent chooses where to move in a given decision space. The agent's goal is to minimize cost over a long time horizon~\cite{Hazan2019}. In the context of right-sizing data centers, it is apparent that we can use the hitting costs to model the energy consumption of active servers and possible delay incurred by having an inadequate number of active servers for some incoming load.

However, in some cases, it is cheaper to keep a server in an idle state for a short period as powering up and powering down servers increases energy consumption and incurs wear-and-tear costs~\cite{Lin2011}. We, therefore, generalize online convex optimization by introducing a known \emph{movement cost}\index{movement cost} that penalizes the agent for movement in the decision space. This generalization is commonly called \emph{smoothed online convex optimization}. However, the movement cost makes the problem more challenging as the optimal choice in each round depends on future cost functions~\cite{Chen2015}. The movement cost sets up a trade-off similar to the \emph{exploration-exploitation trade-off}\index{exploration-exploitation trade-off} where decisions need to be made in the face of uncertainty. However, unlike the exploration-exploitation trade-off, decisions do not lead to the discovery of new information. We also examine variants of smoothed online convex optimization that further restrict the class of allowed convex cost functions, the movement cost, or the decision space based on common modeling choices for right-sizing data centers.

\paragraph{Variants of Smoothed Online Convex Optimization} There are three main differentiators between the algorithms we consider in this work. First, we examine both offline and online algorithms. While only the latter algorithms are of immediate practical use, we use the former as a benchmark. Second, we consider both fractional and integral versions of this problem. In principle, only integral versions solve the right-sizing problem of data centers. However, typically the number of servers in a data center is large enough to warrant the use of fractional versions~\cite{Bansal2015}. Third, we differentiate between uni-dimensional and multi-dimensional versions of the problem. Initially, the right-sizing problem of data centers was proposed for homogeneous data centers~\cite{Lin2011, Bansal2015, Albers2018}. In a homogeneous data center, all servers share the same performance metrics, i.e., are of the same type. In more recent works, the problem was extended to heterogeneous data centers that may employ many different types of servers~\cite{Lin2012, Chen2018, Goel2019, Albers2021, Albers2021_2}.

\paragraph{Competitive Ratio and Regret} The two main performance metrics used to analyze algorithms for smoothed online convex optimization are the competitive ratio and regret. These different performance metrics arise from two separate communities in which online convex optimization problems often appear. In the context of online learning, algorithms aim to minimize regret (or often to achieve sublinear regret). \emph{Regret} is the difference between the loss achieved by the algorithm and the loss of the best fixed point in hindsight~\cite{Chen2018}. In contrast, in the context of online algorithms, the aim is to minimize the competitive ratio (or often to achieve a constant competitive ratio). The \emph{competitive ratio} is defined analogously to the approximation ratio as the ratio of the loss achieved by the online algorithm and the loss achieved by an optimal offline algorithm~\cite{Chen2018}. Note that in contrast to the definition of regret, the optimal solution may move in the decision space. We use both notions in our analysis and formally define the competitive ratio, regret, and alterations in \cref{section:theory:performance_metrics}. Due to the use of two distinct metrics, some recent research is focused on finding unified frameworks to construct algorithms that perform well with regard to both notions of performance~\cite{Chen2018, Goel2019}.

\paragraph{Convex Body Chasing} In the literature, smoothed online convex optimization is also known as \emph{convex function chasing}\index{convex function chasing}. A related problem to convex function chasing is the problem of \emph{convex body chasing}\index{convex body chasing}. In convex body chasing, requests consist of a convex set instead of a convex function. The agent is then required to choose points in a sequence of convex sets so as to minimize their movement~\cite{Antoniadis2016}. Convex body chasing can be reduced to convex function chasing by considering for any convex set $K$ the function $f^K$ which is 0 on $K$ and $\infty$ off of $K$~\cite{Sellke2019}. In contrast,  \citeauthor{Sellke2019}~\cite{Sellke2019} showed that convex function chasing in $\mathbb{R}^d$ can be reduced to convex body chasing in $\mathbb{R}^{d+1}$ by choosing the convex set as the epigraph $\{(x,y) \in \mathbb{R}^d \times \mathbb{R} \mid y \geq f_{\tau}(x)\}$ of the hitting cost $f_{\tau}$. To ensure that points on the lower boundary of the set are chosen (i.e., the line of $f_{\tau}$), requests are alternated with the hyperplane $\mathbb{R}^d \times \{0\}$~\cite{Sellke2019}. $\mathcal{O}(\sqrt{d})$ was shown to be a lower bound to the competitiveness of convex function chasing (and convex body chasing) in high dimensions \cite{Chen2018}. A large body of recent work seeks to find algorithms with near-optimal bounds~\cite{Bansal2017, Bubeck2018, Bubeck2018_2, Argue2019_2, Argue2019, Argue2020, Bubeck2020}. Just recently, \citeauthor{Argue2019}~\cite{Argue2019} found a $\mathcal{O}(d)$-competitive algorithm for convex body chasing.

\paragraph{Predictions} It was shown that to achieve a constant competitive ratio for smoothed online convex optimization, one must either restrict the class of allowed convex cost functions or movement costs~\cite{Chen2018}. Moreover, no algorithm can achieve sublinear regret and a constant competitive ratio even for linear cost functions~\cite{Andrew2015}. These results motivated another body of research that makes use of predictions to bypass this fundamental limitation. It is only natural to attempt using predictions for right-sizing data centers as, usually, there is an extensive collection of server traces to base predictions on. There has also been much recent work improving the accuracy of time-series predictions~\cite{Taylor2017, Benidis2020, Chen2020, Hosseini2021}.

Interestingly, the performance of such predictions in the context of right-sizing data centers has not been studied in much detail. Most of the previous work was focused on finding algorithms based on \emph{receding horizon control} methods that use perfect predictions in some finite \emph{prediction window}~\cite{Lin2012, Chen2015, Badiei2015, Chen2016, Li2018, Lin2019}. We investigate how these algorithms perform with actual predictions as compared to perfect predictions.

\section{Outline}

In this work, we investigate algorithms for smoothed convex optimization and specifications thereof. We then examine the performance of these algorithms in the application of right-sizing data centers based on real server traces. Over the past decade, many algorithms following a variety of different approaches were introduced. One of our main goals in this work is thus to determine which algorithms perform best in which scenarios. We also determine the overall cost- and energy-saving potential of dynamically right-sizing data centers.

We begin in \cref{chapter:application} by modeling the cost associated with operating data centers. In \cref{chapter:theory}, we first introduce commonly used performance metrics in algorithm design. We then formally introduce smoothed convex optimization and its variants that we consider in this work. We also show that smoothed convex optimization and a more restricted variant are NP-hard when the dimension is allowed to vary. In Chapters~\ref{chapter:offline_algorithms} and~\ref{chapter:online_algorithms}, we introduce offline and online algorithms, respectively. In each chapter, we begin with algorithms solving the uni-dimensional problem and then generalize some as well as introduce new approaches for the multi-dimensional setting. We end \cref{chapter:online_algorithms} by examining algorithms that use predictions to make more informed decisions. A central focus of our work is to simplify the use of the algorithms for smoothed optimization in a wide variety of applications. In \cref{chapter:implementation}, we discuss the architecture of our implementation and how the algorithms can be used with other models. In \cref{chapter:case_studies}, we evaluate the performance of the discussed algorithms in the application of right-sizing data centers. To test algorithms solving the homogeneous and the heterogeneous right-sizing problem, we use server traces from Facebook, Microsoft, the Los Alamos National Lab, and Alibaba originating from varying applications. In \cref{chapter:future_work}, we mention important questions for future research and in \cref{chapter:conclusion} we summarize our conclusions. \nameref{chapter:notation} gives an overview of the used notation.

\section{Related Work}\label{section:introduction:related_work}

In this section, we reference two related bodies of research. First, we discuss alternatives to using power-down mechanisms, which we use in this work to increase the energy efficiency of data centers. Second, we examine other areas of application that have been shown to benefit from using techniques from smoothed online convex optimization. The implementation of the algorithms discussed in this work are generic and publicly available~\cite{Huebotter2021}. Therefore, it is an interesting open question how these algorithms perform in practice in the applications referenced here.

\paragraph{Alternatives to Power-Down Mechanisms} As was mentioned previously, power-down mechanisms are just one approach to increasing the energy efficiency of data centers. Another widely studied approach is to lower the frequency/speed of devices to save energy. This approach is known as \emph{dynamic speed scaling}\index{dynamic speed scaling} or dynamic voltage and frequency scaling. Dynamic speed scaling determines the optimal processing speeds and job assignments to minimize energy usage and meet specified performance constraints (like job deadlines)~\cite{Albers2007, Albers2011, Jin2016}. There is also a large body of research considering hybrid approaches that scale device speeds as well as power-down devices~\cite{Jin2016}. The hybrid problem is also known as \emph{speed scaling with a sleep state}\index{speed scaling with a sleep state}~\cite{Albers2014}. Other work has focused on so-called \emph{machine activation}\index{machine activation} problems where a subset of servers is selected according to a fixed cost budget so as to minimize the makespan of a set of jobs~\cite{Khuller2009, Li2011}.

\paragraph{Other Applications of Smoothed Online Convex Optimization} Integral smoothed online convex optimization is subsumed by the class of problems known as metrical task systems. A \emph{metrical task system}\index{metrical task system} consists of a finite metric decision space, a movement cost, and a sequence of hitting costs~\cite{Bubeck2018_3}. Metrical task systems are more general than integral smoothed online convex optimization as the decision space is not required to be embedded in $\mathbb{Z}$ and the hitting costs do not need to be convex. A prominent instance of metrical task systems is the \emph{$k$-server problem}\index{$k$-server problem}~\cite{Bubeck2017}.

Any online convex optimization problem where, in reality, some cost is associated with taking action can be interpreted as a smoothed convex optimization. This is the case for many typical online convex optimization problems, as movement costs are often disregarded to simplify the algorithm design. Such problems are

\begin{itemize}
    \item \emph{video streaming} where the encoding quality varies based on available bandwidth, but frequent changes in encoding quality should be avoided~\cite{Lin2012},
    \item \emph{portfolio management} in which expert advice indicates that certain actions maximize profit but taking action incurs some cost~\cite{Calafiore2008, Das2014, Ballu2019},
    \item \emph{power generation} with dynamic demand as the cheapest generators tend to have high costs associated with toggling them on and off~\cite{Lin2012, Badiei2015},
    \item \emph{contextual sequence prediction} --- used, for example, in object tracking, natural language processing, and sequence alignment --- where the prediction of the next element of the sequence has to be contextualized (i.e., smoothed) based on its predecessors~\cite{Kim2015},
    \item \emph{vertical container scaling} where the per-container resource allocation should depend on load but scaling leads to a period of unavailability~\cite{Rossi2019},
    \item \emph{multi-timescale control} where the linear control constraints act as regularizers~\cite{Goel2017},
    \item \emph{smoothed online regression, ridge regression, logistic regression, online maximum likelihood estimation}, and \emph{linear quadratic regulator control} are direct instances of smoothed online convex optimization~\cite{Goel2018, Goel2019_2, Shi2020},
    \item \emph{thermal management} where the operating changes to achieve temperature constraints should be smooth rather than abrupt~\cite{Zanini2009},
    \item \emph{electric vehicle charging} where prices can be used to prevent load variations, but prices should not change too quickly~\cite{Kim2014}, and
    \item \emph{routing in networks} (e.g., automatically switched optical networks) where there is a cost for establishing a connection~\cite{Lin2012}.
\end{itemize}

\chapter{Application: Right-Sizing Data Centers}\label{chapter:application}

In this chapter, we seek to derive models for right-sizing data centers based on smoothed convex optimization. We then use these models to motivate the variations of smoothed convex optimization we discuss in \cref{chapter:theory}. We begin by discussing the general features of server infrastructures. Then, we examine the modeling of \emph{operating costs}\index{operating cost} (i.e., hitting costs) and \emph{switching costs}\index{switching cost} (i.e., movement costs) in detail.

Data centers are large-scale, complex systems. Therefore, any model we examine in this chapter is an approximation. Our goal is to find models that generalize well across many data centers. When examining a specific data center, the models we discuss can be extended to yield better approximations.

\section{Architectures}\label{section:application:architectures}

We begin by revisiting the characteristics by which the design of data centers varies.

\paragraph{Speed-Scalability} For our data center model, it is natural to assume that the servers are speed scalable. That is, their utilization can vary from idling at 0\% utilization to full load at 100\% utilization. Furthermore, we assume that server utilization scales linearly with its load as this is required to maintain a steady quality of service~\cite{Bansal2015}.

\paragraph{Energy} There has been much recent work on modeling the energy consumption of a data center as a function of the speed of individual servers. We discuss these approaches in detail in \cref{section:application:operating_cost:energy}. However, in many cases, the ecological and economic cost of energy does not remain constant but fluctuates over time. More importantly, at no point in time is the energy cost a linear function of energy consumption because many data centers have quotas for each energy source~\cite{Miller2021}. For example, there may only be a limited supply of renewable energy. Once the energy consumption of our data center exceeds this supply, it has to resort to different sources of energy with different costs. Recently, a trend has also been for data centers to produce their own renewable energy~\cite{Lin2012}. In this case, this renewable energy is significantly cheaper than any energy purchased after energy consumption exceeds energy production. If, in contrast, more energy is produced than is consumed, some energy may even be sold depending on the energy grid's state.

\paragraph{Homogeneity and Heterogeneity} When the data center right-sizing problem was first introduced, approaches were focused on homogeneous data centers. In a \emph{homogeneous}\index{homogeneous data center} data center, all servers are of the same type. We say that a server is of a different type than another server if their operating costs or switching costs differ significantly. In any such scenario where we have multiple server types, the data center is \emph{heterogeneous}\index{heterogeneous data center}. It is easy to see that each server type resembles one dimension in our smoothed optimization. While a homogeneous data center model is much simpler, most data centers are heterogeneous in practice, a trend that is projected to intensify~\cite{Jin2016}. Heterogeneity may arise naturally as defect servers of a homogeneous data center are replaced by newer servers. However, the primary advantage of heterogeneous architectures is that certain tasks can be delegated to specialized servers~\cite{Jin2016}. For example, CPUs and GPUs may be used within the same data center, but GPUs should only perform massive parallel computations as CPUs are faster when tasks cannot be parallelized easily~\cite{Shan2006}. The differing power-performance relationships of multiple server types are another benefit of heterogeneous architectures~\cite{Jin2016}.

\paragraph{Size} In principle, the data center right-sizing problem only admits integral solutions as, at any time, we can only run an integral number of servers of each type. However, if the size of our data center is large in each dimension, it is reasonable to use fractional solutions as an approximation. Typically, data centers fulfill this requirement. Lately, the surge in hyperscale facilities indicates a trend towards larger data centers~\cite{Jones2018}. For this reason, we discuss integral as well as fractional solutions in our analysis.

\paragraph{Reliability and Availability} Many services must satisfy specific requirements regarding reliability and availability, which are often key components of \emph{service level agreements}\index{service level agreement} (SLAs)~\cite{Lin2011}. Such requirements can be enforced as hard constraints using the decision space by requiring a minimum number of active servers per server type in our model. As our model only chooses the number of servers of some type, the algorithm can freely choose the active servers between all servers according to some guidelines to meet the requirements. However, choosing a decision space that is too tight may reduce the cost-saving potential. Therefore, an alternative approach is to use operating costs (discussed in \cref{section:application:operating_cost}) as softer constraints, for example, by enforcing a maximum utilization on servers of type $k$ that is less than some $\theta_k < 100\%$. Availability requirements for specific jobs can also be encoded into the revenue loss as a function of average job delay.

\paragraph{Networks} Most of our analysis is focused on the case of a single data center. However, the problem of deciding where to route incoming loads within a network of data centers so as to minimize the overall cost is acutely relevant~\cite{Miller2021}. For example, if data centers produce their own renewable energy, the cost of energy at each individual data center is likely to vary drastically over time as weather conditions shift~\cite{Lin2012}. Therefore, previous studies focusing on individual data centers found that wind and solar can only be used with large-scale storage due to their intermittency~\cite{Gmach2010, Gmach2010_2}. Nevertheless, \citeauthor{Lin2012}~\cite{Lin2012} showed that by running a network of data centers in separate locations, the adverse effects of renewable energy production could largely be avoided as the availability of solar and wind can be aggregated across locations. In \cref{section:application:dynamic_routing}, we show how our cost model can be extended to support geographical load balancing across a network of data centers.

\section{Dispatching}\label{section:application:dispatching}

The modeling of load is the core of our data center model. We say that a \emph{load profile}\index{load profile} of our system during time slot $t$ is a vector $\lambda_t \in \mathbb{N}_0^e$ where $e$ is the number of job types. In \cref{section:application:dispatching:multiple_load_types}, we give more concrete examples of varying job types, but generally, jobs are of different types when their cost model is different. For now, we assume that the processing time of all jobs on any server type takes exactly one time slot. In \cref{section:application:dynamic_duration}, we discuss how this approach can be extended to jobs with a dynamic duration (per server type).

We denote by $m_k \in \mathbb{N}$ the maximum number of available servers of type $k$ and by $d$ the number of server types. Our decision space is therefore given as $\mathcal{X} := \mathbb{R}_{\geq 0, \leq m_0} \times \dots \times \mathbb{R}_{\geq 0, \leq m_d}$. Further, we denote by $l_k^{\text{max}} \in \mathbb{N}$ the maximum number of jobs a server of type $k$ can process in a single time slot.

We call a load profile $\lambda_t$ \emph{feasible}\index{feasible load profile} if \begin{align}
    \sum_{i=1}^e \lambda_{t,i} \leq \sum_{k=1}^d l_k^{\text{max}} m_k =: \lambda^{\text{max}}.
\label{eq:feasible_load_profiles}
\end{align} In words, a load profile is feasible if we can process all incoming jobs by using all available servers. In our subsequent analysis, we assume that all load profiles are feasible. This is not a restriction as, in practice, if a load profile is not feasible, one would have to delay a large enough selection of jobs to an upcoming time slot until the feasibility of the current load profile is achieved.

\subsection{Optimal Load Balancing}\label{section:application:dispatching:optimal_load_balancing}

To begin with, we consider a data center with homogeneous loads, so by \cref{eq:feasible_load_profiles}, we have $\lambda_t \in [\lambda^{\text{max}}]_0 = \{0, 1, \dots, \lambda^{\text{max}}\}$. As we discussed in \cref{section:application:architectures}, energy consumption depends on server speed which in turn depends on the number of available servers $x_t \in \mathcal{X}$ (determined by the algorithm) and the load profile $\lambda_t$ (determined by an adversary). Therefore, it is natural to ask how we can optimally distribute the load across the available servers. Let $g_{t,k} : [0,l_k^{\text{max}}] \to \mathbb{R}_{\geq 0}$ be a convex increasing non-negative function representing the cost incurred by operating a server of type $k$ during time slot $t$ with a load of $l \in [0,l_k^{\text{max}}]$. We set $g_{t,k}(l) = \infty$ for $l > l_k^{\text{max}}$. The utilization (or speed) of a server of type $k$ with a load of $l$ is then given as $s_k(l) := l / l_k^{\text{max}} \in [0,1]$, assuming that the speed of a server is linearly proportional to its load.

We first consider the homogeneous setting. In the homogeneous case, we write $g_t := g_{t,1}$. We discuss concrete functions in \cref{section:application:operating_cost}, but for this section we assume that $g_t$ can be any non-negative convex function. Disregarding switching costs, we obtain the following optimization for the homogeneous setting: \begin{align*}
    \min_{x_t \in \mathcal{X}} \quad &\sum_{t=1}^T \sum_{i=1}^{x_t} g_t(l_{t,i}) \\
    \text{subject to}        \quad &\sum_{i=1}^{x_t} l_{t,i} = \lambda_t
\end{align*} where $l_{t,i} \in \mathbb{R}_{\geq 0}$ denotes the number of jobs processed by server $i$ during time slot $t$. \citeauthor{Lin2011}~\cite{Lin2011} showed that for fixed $x_t$ the remaining dispatching problem is convex. The optimal dispatching rule $l_{t,i}^*$ is $\lambda_t / x_t$ for all $i \in [x_t]$, implying that given $x_t$, it is optimal to balance load evenly across all servers. \citeauthor{Albers2021_2}~\cite{Albers2021_2} prove this fact using Jensen's inequality. We, therefore, define our overall operating cost as \begin{align}\label{eq:homogeneous_load_balancing}
    f_t(x) := x g_t\left(\frac{\lambda_t}{x}\right).
\end{align} It is easy to see that this definition of $f_t$ is jointly convex in $\lambda_t$ and $x$. Crucially, this is an approximation as we did not impose the restriction that job arrival rates must be integral, which is required in practice.

In the heterogeneous setting, it is easy to see that there is no single optimal dispatching rule. However, our analysis implies that the optimal dispatching strategy is to load balance within one server type, even in the heterogeneous setting~\cite{Albers2021_2}. We define the operating cost for servers of type $k$ during time slot $t$ as \begin{align}\label{eq:heterogeneous_load_balancing_unit}
    h_{t,k}(x,z) := \begin{cases} 
        x g_{t,k}\left(\frac{l_{t,k}}{x}\right) & x > 0 \\
        \infty                                  & x = 0 \land l_{t,k} > 0 \\
        0                                       & x = 0 \land l_{t,k} = 0
    \end{cases}
\end{align} where $l_{t,k} = \lambda_t z$, $x$ is the number of active servers of type $k$, and $z \in [0,1]$ is the fraction of the job volume $\lambda_t$ that is assigned to server type $k$~\cite{Albers2021_2}. As $g_{t,k}$ is convex and increasing it follows that $h_{t,k}$ too is jointly convex in $\lambda_t$ and $x$. Given the set of all possible job assignments to a collection of $d$ different server types $\mathcal{Z} := \{z \in [0,1]^d \mid \sum_{k=1}^d z_k = 1\}$, the overall operating cost can be defined as the convex optimization \begin{align}\label{eq:heterogeneous_load_balancing}
    f_t(x) := \min_{z \in \mathcal{Z}} \sum_{k=1}^d h_{t,k}(x_k,z_k).
\end{align} It is easy to see that \cref{eq:homogeneous_load_balancing} is equivalent to \cref{eq:heterogeneous_load_balancing} for $d = 1$.

\Cref{eq:heterogeneous_load_balancing} can be computed using the convex program \begin{alignat}{2}\label{eq:heterogeneous_load_balancing_constrained_convex_program}
    &\min_{z \in [0,1]^d} &\qquad&\sum_{k=1}^d h_{t,k}(x_k,z_k) \\
    &\textrm{subject to}  &      &\sum_{k=1}^d z_k = 1.
\end{alignat} A linear-equality-constrained problem can be solved via a change of variable~\cite{Singer2016}. \Cref{eq:heterogeneous_load_balancing_constrained_convex_program} using an equality constraint is equivalent to a $(d-1)$-dimensional convex program as the last dimension $d$ is completely determined by the dimensions $1$ through $d-1$: \begin{alignat}{2}\label{eq:heterogeneous_load_balancing_convex_program}
    &\min_{z \in [0,1]^{d-1}} &\qquad&\sum_{k=1}^{d-1} h_{t,k}(x_k,z_k) + h_{t,d}(x_d,z_d) \\
    &\textrm{subject to}  &      &\sum_{k=1}^{d-1} z_k \leq 1 \nonumber
\end{alignat} where $z_d = 1 - \sum_{k=1}^{d-1} z_k$. \Cref{eq:heterogeneous_load_balancing_convex_program} can easily be computed numerically using a standard solver for convex programs.

\subsection{Multiple Job Types}\label{section:application:dispatching:multiple_load_types}

The problem becomes harder when we consider heterogeneous loads. However, as we have seen in \cref{section:application:architectures}, multiple \emph{job types}\index{job type} are often required in practice, for example, to distinguish different processing speeds of CPUs and GPUs for different tasks. Instead of determining the optimal assignment of fractions of the total load to server types, we now need to determine the optimal assignment of fractions of individual job types to server types. To that end, we define the set of such assignments as \begin{align*}
    \mathcal{Z}_t := \left\{z_t \in [0,1]^{e^d} \mid \forall i \in [e].\ \lambda_t > 0 \implies \sum_{k=1}^d z_{t,k,i} = \frac{\lambda_{t,i}}{\lambda_t}\right\}
\end{align*} where $\lambda_t := \sum_{i=1}^e \lambda_{t,i}$. Here, $z_{t,k,i}$ is the fraction of jobs of type $i$ assigned to servers of type $k$ during time slot $t$. $\lambda_{t,i} / \lambda_t$ is the fraction of jobs of type $i$ at time $t$.

We continue to use optimal load balancing similarly to \cref{eq:heterogeneous_load_balancing_unit} to distribute load evenly across all servers of the same type. However, we introduce an additional cost that is paid per job of type $i$ that is processed on a server of type $k$. Our new operating cost for servers of type $k$ during time slot $t$ thus becomes \begin{align}\label{eq:multiple_load_types_load_balancing_unit}
    h_{t,k}(x,z) := \begin{cases} 
        x g_{t,k}\left(\frac{l_{t,k}}{x}\right) + \sum_{i=1}^e l_{t,k,i} q_{t,k,i}\left(\frac{l_{t,k}}{x}\right) & x > 0 \\
        \infty                                                                                                   & x = 0 \land l_{t,k} > 0 \\
        0                                                                                                        & x = 0 \land l_{t,k} = 0
    \end{cases}
\end{align} where $l_{t,k,i} = \lambda_{t,i} z_i$, $l_{t,k} = \sum_{i=1}^e l_{t,k,i}$, $x$ is the number of active servers of type $k$, and $z \in [0,1]^e$ are the fractions of the job volumes $\lambda_{t,i}$ that are assigned to server type $k$. Here, $q_{t,k,i}(l)$ is the convex increasing non-negative cost incurred by processing a job of type $i$ on a server of type $k$ during time slot $t$ when a total of $l$ jobs are processed on this server. We discuss this cost in greater detail in \cref{section:application:operating_cost}. $g_{t,k}(l)$ remains the convex increasing non-negative operating cost of a server of type $k$ during time slot $t$ under total load $l$. It is easy to see that this definition of $h$ remains jointly convex in $\lambda_t$ and $x$. We still set $g_{t,k}(l) = \infty$ if $l > l_k^{max}$.

We can now define the overall operating cost analogously to \cref{eq:heterogeneous_load_balancing} as the convex optimization \begin{align}\label{eq:multiple_load_types_load_balancing}
    f_t(x) := \min_{z \in \mathcal{Z}_t} \sum_{k=1}^d h_{t,k}(x_k,z_k).
\end{align}

Again, we observe that for $e = 1$ \cref{eq:multiple_load_types_load_balancing} is equivalent to \cref{eq:heterogeneous_load_balancing} by setting $q_{t,k,1} \equiv 0$. Henceforth, we restrict our analysis to the model from \cref{eq:multiple_load_types_load_balancing}.

We can use a similar approach to \cref{eq:heterogeneous_load_balancing_convex_program} to simplify \cref{eq:multiple_load_types_load_balancing} to a $(d-1)$-dimensional convex optimization: \begin{alignat}{2}\label{eq:multiple_load_types_load_balancing_convex_program}
    &\min_{z \in [0,1]^{d-1}} &\qquad&\sum_{k=1}^{d-1} h_{t,k}(x_k,z_k) + h_{t,d}(x_d,z_d) \\
    &\textrm{subject to}  &      &\sum_{k=1}^{d-1} z_{k,i} \leq \lambda_{t,i} / \lambda_t \qquad\forall i \in [e]\quad\text{if }\lambda_t > 0 \nonumber
\end{alignat} where $z_{d,i} = \lambda_{t,i} / \lambda_t - \sum_{k=1}^{d-1} z_{k,i}$ if $\lambda_t > 0$ and $z_{d,i} = 0$ otherwise. Again, \cref{eq:multiple_load_types_load_balancing_convex_program} can be computed using a standard solver for convex programs.

\section{Operating Cost}\label{section:application:operating_cost}

Our next goal is to model the operating cost of servers in a data center and the cost of powering up and powering down servers. Note that reducing the energy consumption of a data center also reduces cooling and power distribution costs~\cite{Lin2011, Clark2005}. It is, therefore, reasonable to focus on server-specific costs.

Given our analysis from \cref{section:application:dispatching}, we seek to determine the server-dependent cost $g_{t,k}(l)$ and the job-dependent cost $q_{t,k,i}(l)$ given $l$ jobs are processed on servers of type $k$. As introduced in \cref{chapter:introduction}, we interpret the server-dependent cost as \emph{energy cost}\index{energy cost} and the job-dependent costs as \emph{revenue loss}\index{revenue loss}.

\paragraph{Energy Cost} Energy cost is a function of energy consumption which in turn is a function of server utilization. We have seen in \cref{section:application:dispatching:optimal_load_balancing} that the server utilization of a server of type $k$ given a server load of $l$ is $l / l_k^{\text{max}}$ where we defined $l_k^{\text{max}}$ as the maximum number of jobs a server of type $k$ can process in a single time slot. Let $e_{t,k}(s)$ be the energy cost of operating a server of type $k$ during time slot $t$ with utilization $s$. Then, \begin{align*}
    g_{t,k}(l) := e_{t,k}\left(\frac{l}{l_k^{\text{max}}}\right).
\end{align*} Some authors only consider energy costs as they assume the largest fraction of operating costs~\cite{Bansal2015}.

\paragraph{Revenue Loss} Revenue loss measures the lost revenue based on our distribution of incoming job types to server types. We model revenue loss as a convex increasing non-negative function $r_{t,i}(d)$ that describes the domain-specific revenue loss of jobs of type $i$ during time slot $t$ given an average delay of $d$. We model the average delay $d$ of jobs processed on a server of type $k$ where the total load on the server is $l$ as the convex increasing non-negative function $d_{k}(l)$. We also introduce an additional delay $\delta_{t,k,i}$ which models the constant delay incurred by processing jobs of type $i$ on servers of type $k$ during time slot $t$. Hence, we obtain \begin{align*}
    q_{t,k,i}(l) := r_{t,i}(d_{k}(l) + \delta_{t,k,i}).
\end{align*}

Typically, revenue loss is scaled linearly with delay~\cite{Lin2011}. We thus set $r_{t,i}(d) = \gamma_i \cdot d$ for factors $\gamma_i \in \mathbb{R}_{\geq 0}$.

\paragraph{} In the subsequent sections, we consider models for energy cost and delay, respectively.

\subsection{Energy}\label{section:application:operating_cost:energy}

Our goal is to model the energy cost $e_{t,k}(s)$ of a server of type $k$ during time slot $t$ based on its utilization $s$. To this end, we consider two functions. First, let $\phi_k(s)$ denote the energy consumption of a server of type $k$ with utilization $s$. Second, let $\nu_{t,k}(p)$ be the energy cost of a server of type $k$ during time slot $t$ associated with an energy consumption of $p$. We then set $e_{t,k}(s) := \nu_{t,k}(\phi_k(s))$. If the utilization $s$ exceeds the maximum allowed utilization $\theta_k \in [0,1]$ of server $k$ as described in \cref{section:application:architectures}, i.e. $\theta_k < s$, we set $e_{t,k}(s) = \infty$.

\paragraph{Energy Cost} We begin by modeling the energy cost associated with a consumption of $p$ units of energy during time slot $t$. The simplest model assigns each unit of energy the average cost during time slot $t$, which we call $c_t$. We then obtain $\nu_{t,k}(p) := c_t p$ for all $k \in [d]$ which is convex, increasing, and non-negative. If we simply want to achieve power-proportionality of the data center, it suffices to set $c \equiv 1$. We seek a more complex model in many practical applications, which we introduce at the end of this subsection and describe formally in \cref{section:application:energy_quotas}.

\paragraph{Energy Consumption} There are a variety of models for energy consumption in a data center. We give an overview of the most common models and reference additional models. To rephrase our objective, we seek to model the energy consumption of a single server based on the utilization (also referred to as speed or frequency) $s$. The energy consumption $\phi_k(s)$ can be calculated as $\delta \Phi_k(s)$ where $\delta$ is the length of a time slot and $\Phi_k(s)$ is the power consumption of a server of type $k$ with utilization $s$~\cite{Dayarathna2016}. Models of power consumption can be categorized into linear and non-linear models~\cite{Ismail2020}. In this work, we mostly restrict our analysis to linear models of which the performance is highly dependent on the chosen parameters~\cite{Ismail2020}. \citeauthor{Ismail2020}~\cite{Ismail2020} also present more accurate non-linear and machine learning models.

An intuitive model of power consumption is discussed by \citeauthor{Dayarathna2016}~\cite{Dayarathna2016} and \citeauthor{Ismail2020}~\cite{Ismail2020} is \begin{align}\label{eq:energy_model:1}
    \Phi_k(s) = (\Phi_k^{\text{max}} - \Phi_k^{\text{min}})s + \Phi_k^{\text{min}}
\end{align} where $\Phi_k^{\text{max}}$ and $\Phi_k^{\text{min}}$ are the power a server of type $k$ consumes at full load and when idling, respectively. In linear power models we distinguish between the \emph{dynamic power}\index{dynamic power}, here the first term $(\Phi_k^{\text{max}} - \Phi_k^{\text{min}})s$, and the \emph{static power}\index{static power} (or leakage power), $\Phi_k^{\text{min}}$. Following the findings of \citeauthor{Barroso2007}~\cite{Barroso2007} and to simplify the above model, we can use that generally servers consume half of their peak power when idling~\cite{Ismail2020}: \begin{align}\label{eq:energy_model:2}
    \Phi_k(s) = \frac{1}{2} \Phi_k^{\text{max}} (1 + s).
\end{align}

The above models of power consumption are linear. Another frequently used model is non-linear and defined as \begin{align}\label{eq:energy_model:3}
    \Phi_k(s) = \frac{s^{\alpha_k}}{\beta_k} + \Phi_k^{\text{min}}
\end{align} where $\alpha_k > 1$ and $\beta_k > 0$ are constants~\cite{Dayarathna2016}. Here, $s^{\alpha_k}/\beta_k$ is the dynamic power and $\Phi_k^{\text{min}}$ is the static power. A variant of this model is used by \citeauthor{Bansal2015}~\cite{Bansal2015}.

We observe that all of the above models are convex, increasing, and non-negative.

\paragraph{Energy Quotas} We mentioned the prevalence of energy quotas in many practical applications in \cref{section:application:architectures}. Up until now, we have only considered a fixed energy price per unit of energy. When we consider energy quotas, this setup changes. For example, we may produce a changing amount of renewable energy at a data center which is much cheaper than regular energy~\cite{Lin2012}. It is easy to see that computing the energy cost per server is insufficient in such a scenario. Instead, we must adjust our model from \cref{eq:multiple_load_types_load_balancing} to simultaneously calculate the energy cost across all servers (of all server types). However, we can easily model more complex energy prices within our existing framework once this adjustment is made. For example, \citeauthor{Lin2012}~\cite{Lin2012} use a simplified model that does not take into account individual server utilization but assumes a quota of renewable energy that is assumed to be free of charge: \begin{align*}
    e_t(x) = c_{t}(x - p_t)^+
\end{align*} where $c_t$ is the average price per unit of energy during time slot $t$ and $p_t$ is the quota of free renewable energy during time slot $t$, considering only a single server type. We define $(\cdot)^+ := \max \{0, \cdot\}$. In this model, each active server consumes approximately one unit of energy. Note that here $e_t$ depends on the number of active servers $x$ allowing the consideration of quotas, whereas, in our original model, it solely depends on $l$.

This simple model could be extended by considering quotas for multiple sources of energy, considering the gain of selling unused renewable energy back to the grid, or by computing energy consumption based on the actual utilization of active servers. We present one such model in \cref{section:application:energy_quotas}.

\paragraph{Economic and Ecological Cost} Our energy quota model can be extended to model a variety of incentives where the incentives are provided by our choice of the average energy costs of energy source $i$ per unit of energy during time slot $t$ which are denoted by $c_{t,i}$. Contrary to initial intuition, the nature of these costs does not have to be purely economical. While it is reasonable to consider the cost of energy, we can extend our incentives by considering the emission of $CO_2$ equivalents per unit of energy of each energy source. Such a policy can guide towards a carbon-free makeup of energy sources across all data center locations. Such a model is of interest in many current data center networks~\cite{Hoelzle2020, Miller2021}.

\subsection{Delay}

We use queueing theory to model the queueing delay of jobs in the system. We are interested in the average delay $d_{k}(l)$ of jobs when they are processed by a server of type $k$ with a total of $l$ serviced jobs.

In our model, we consider a single service channel, our server of type $k$. We further assume that the queue's capacity is unlimited, as is the potential number of job arrivals. The latter assumption is an approximation as, in principle, we can expect a total of $l$ arrivals during time slot $t$. It is natural to assume that the arrival of jobs is Markovian, i.e., Poisson-distributed. So the interarrival times of jobs follow the exponential distribution. To remain as general as possible, the only restriction we impose on service times is that we assume they are independent. Hence, our assumptions naturally lead us to model delay based on an M/GI/1 queue.

A good model of the queueing discipline of a server is the \emph{Round Robin}\index{round robin scheduling} (RR) scheduling algorithm, where jobs are processed in turn, but when the processing of a job exceeds some time quantum, it is moved to the back of the circular queue. In many cases, however, the idealized \emph{Processor Sharing}\index{processor sharing queue} (PS) discipline is used as an approximation of RR~\cite{Lin2011, Lin2012}. In a PS queue, each job in the system is processed simultaneously at a rate inversely proportional to the current number of jobs. Therefore, the service rate is given as $C / n$ where $C$ is the server's capacity and $n$ is the current number of jobs. PS is an approximation of RR as, in general, the capacity of a server cannot be divided into real-valued parts~\cite{Virtamo2007}. However, note that this approximation echoes our simplification in \cref{section:application:dispatching:optimal_load_balancing}, leading to optimal load balancing on a per-server level.

A single server modeled by the PS discipline and operating according to Poisson-distributed arrivals has the valuable property that its queue length distribution is geometric irrespectively of the service time distribution~\cite{Aalto2007}. In other words, the average delay in a PS queue is insensitive to the service time distribution.

Let $X \sim \text{Po}(\lambda)$ be the number of arriving jobs per time unit with rate $\lambda$. The expected delay $E T$ of a PS queue is then given as \begin{align}\label{eq:avg_delay}
    E T = \frac{1/\mu}{1-\rho} = \frac{1}{\mu - \lambda}
\end{align} where $\mu = C / E X = C / \lambda$ is the service rate of the server and $\rho = \lambda / \mu$ is the parameter to the geometric queue length distribution~\cite{Virtamo2007}. The PS discipline can be considered utmost egalitarian as the average delay of any job in our system is directly proportional to the total number of jobs but does not depend on the type of the job~\cite{Virtamo2007}.

Using the model from \cref{eq:avg_delay}, the average delay of jobs processed by servers of type $k$ is given by \begin{align}\label{eq:delay}
    d_{k}(l) := \frac{1}{\mu_k - l}
\end{align} where $\mu_k$ is the service rate of a server of type $k$ and $l$ is the total number of jobs processed by the server. Assuming that a server can only process a single job during a time slot, i.e. $C = 1$, the service rate is given as $\mu_k = 1$, a model previously used by \citeauthor{Lin2011}~\cite{Lin2011, Lin2012}. In \cref{section:application:dynamic_duration}, we discuss how the delay can be obtained in a more general setting where the duration of jobs is heterogeneous.

Given average delay $d$, we use a natural model of the revenue loss similar to the proposal of \citeauthor{Lin2011}~\cite{Lin2011} which is given by $r_{t,i}(d) := (d - \delta_i)^+$ where $\delta_i$ denotes the minimal detectable delay of jobs of type $i$. For $\delta_i = 0$ for all $i \in [e]$, this model is the same as the model of~\cite{Lin2012}. Note that our definitions of $d_k$ and $r_{t,i}$ are convex, increasing, and non-negative.

\section{Switching Cost}\label{section:application:switching_cost}

The switching cost can be understood as the cost associated with transitioning a server from a sleep state to the active state and vice versa. This switching cost is independent of time but may depend on the type of server that is transitioned. Hence, we naturally arrive at the restriction which we later impose on simplified smoothed convex optimization (\cref{problem:simplified_smoothed_convex_optimization}) where we introduce dimension-dependent transition costs $\beta_k$ and define the movement cost as a generalization of the $\ell_1$-norm \begin{align*}
    \norm{x} = \sum_{k=1}^d \beta_k (x)^+
\end{align*} where we assume $X_0 = X_{T+1} = \mathbf{0}$. Note that in this model, we only pay the transition cost $\beta_k$ when a server is powered up. As all servers have to arrive in the sleep state eventually, we can fold the cost of powering down a server into $\beta_k$, i.e., $\beta_k$ represents the cost associated with powering up and powering down a server. Additionally, we assume that the operating cost associated with a sleeping server is $0$. This restriction is reasonable when we interpret the sleep state as a server that is fully powered down.

\paragraph{Model} \citeauthor{Lin2011}~\cite{Lin2011} identify four costs contributing to the transition cost: First, (1) the additional energy consumed by toggling a server on and off $\epsilon_k$, (2) the delay in migrating connections or data $\delta_k$, for example, when using virtual machines, (3) wear-and-tear costs of toggling a server $\tau_k$, and (4) the perceived risk $\rho_k$ associated with toggling a server of type $k$. We thus model the transition cost as \begin{align*}
    \beta_k := c_k(\epsilon_k + \delta_k \Phi_k^{\text{max}}) + \tau_k + \rho_k
\end{align*} where $c_k$ is the average cost of energy for servers of type $k$.

When only (1) and (2) are considered, $\beta_k$ is on the order of migrating network state~\cite{Chen2008}, storage state~\cite{Thereska2009}, or a large virtual machine~\cite{Clark2005} which roughly translates to the cost of operating a server of type $k$ for a few seconds to several minutes~\cite{Lin2011}. Including (3) increases $\beta_k$ to the order of operating a server of type $k$ for an hour~\cite{Bodik2008}. Our model of risk associated with toggling a server is the vaguest. \citeauthor{Lin2011}~\cite{Lin2011} suggest that if this risk is included, $\beta_k$ is on the order of operating a server of type $k$ for several hours.

We call $\xi_k := \beta_k / e_k(0)$ the \emph{normalized switching cost}\index{normalized switching cost} where $e_k(0) = c_k \delta \Phi_k^{\text{min}}$ is the average energy cost of an idling server of type $k$ in a single time slot. Hence, $\xi_k$ approximately measures the minimum duration a server must be asleep to outweigh the switching cost. Competitive algorithms typically wait until this cost is amortized before taking action. Therefore, the normalized switching cost can be used to review how a chosen switching cost relates to the remainder of the model.

\section{Dynamic Routing}\label{section:application:dynamic_routing}

For dynamic routing, also called \emph{geographical load balancing}\index{geographical load balancing}, we consider a network of $\iota$ data centers. We are interested in dispatching incoming jobs from $\zeta$ different geographically centered locations to the data centers and simultaneously right-sizing each data center (i.e., determining the number of active servers)~\cite{Lin2012}. Let $d$ be the number of server types and $e$ the number of job types. We observe that this problem can be translated to a pure data center right-sizing problem by considering $\iota \cdot d$ dimensions and a total of $\zeta \cdot e$ load types. A dimension $(j,k)$ thus encompasses a data center $j \in [\iota]$ and a server type $k \in [d]$. A \emph{load type}\index{load type} $(s,i)$ encompasses a source $s \in [\zeta]$ and a job type $i \in [e]$.

Costs can be modeled in the same way that was presented in \cref{section:application:operating_cost}. Here, $\delta_{t,(j,k),(s,i)}$ can be interpreted as the network delay incurred by routing a request from source $s$ to data center $j$ during time slot $t$~\cite{Lin2012}.

\section{Dynamic Duration}\label{section:application:dynamic_duration}

In most scenarios, job types do not only incur different costs (as covered in \cref{section:application:dispatching:multiple_load_types}) but also have varying duration. We denote by $\delta$ the length of a time slot and by $\eta_{k,i}$ the average processing time of a job of type $i$ on a server of type $k$ assuming the server operates at full utilization. We impose the natural assumption that for any job type $i \in [e]$ there exists a server type $k \in [d]$ such that $\eta_{k,i} \leq \delta$.

If this were not the case, we could not guarantee that the jobs of this type finish in time on some server before servers are assigned to a new set of jobs at the beginning of the next time slot. This model can be extended to support jobs that reach across multiple time slots by keeping track of when jobs will be finished and extending each newly arriving load profile with all unfinished jobs. It can then be enforced that those jobs must be processed on the same server type throughout their lifetime by setting the hitting cost appropriately.

\paragraph{Sub-Time Jobs} To translate the problem that considers dynamic job durations to the problem discussed in \cref{section:application:dispatching:multiple_load_types}, we interpret the length of time during which a server operates at full utilization to process its assigned jobs as the (fractional) number of \emph{sub-time jobs}\index{sub-time jobs}. During time slot $t$ for a server type $k$ with $x$ active servers, the number of sub-time jobs that are processed on a single server of type $k$ is given as $\sum_{i=1}^e l_{t,k,i} \eta_{k,i}$. Note that this is similar to our previous definition of $l_{t,k}$ but is scaled with the time jobs take to be processed. \Cref{fig:dynamic_job_duration} shows how an assignment of sub-time jobs to a server relates to an assignment of jobs.

\begin{figure}
    \centering
    \input{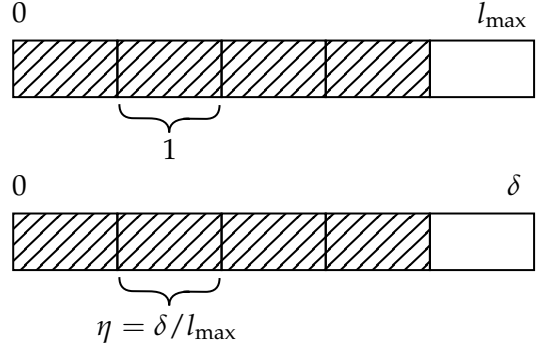}
    \caption{Relationship of jobs and sub-time jobs. Shown is an assignment of four jobs with equal runtime to a server with $l_{\text{max}} = 5$. For some fixed $\delta$, the jobs are translated to sub-time jobs by setting $\eta = \delta / l_{\text{max}}$. Jobs are visualized by the upper bar, their corresponding sub-time jobs are displayed by the lower bar. The figure shows that the number of sub-time jobs is proportional to the processing time of a job. As servers are assumed to be dynamically speed-scalable, the utilization of a server during some time slot is proportional to the number of sub-time jobs it is assigned to.}
    \label{fig:dynamic_job_duration}
\end{figure}

\paragraph{Utilization} The utilization of a server of type $k$ given an assignment of $l$ sub-time jobs is given as $l / \delta$. We set $g_{t,k}(l) = \infty$ if $l > \delta$. We also add the processing time to the perceived delay until a job is finished by increasing $\delta_{t,k,i}$ by $\eta_{k,i}$.

\paragraph{Delay} To model the average delay, we use the average job duration to determine the service rate. Formally, we set the service rate $\mu$ from \cref{eq:avg_delay} to $\mu = 1 / E Y$ where the expected duration of a job on a server of type $k$ during time slot $t$ is \begin{align*}
    E Y = \begin{cases}
        \sum_{i=1}^e l_{t,k,i} \eta_{k,i} / \sum_{i=1}^e l_{t,k,i} & \sum_{i=1}^e l_{t,k,i} > 0 \\
        0 & \text{otherwise}.
    \end{cases}
\end{align*} Note that we normalized the service rate to unit time. The arrival rate is therefore given as the number of jobs handled on a single server of type $k$ divided by the length of a time slot, i.e. $\lambda = \frac{1}{\delta} \sum_{i=1}^e l_{t,k,i} / X_{t,k}$. When servers of type $k$ receive no load, i.e., $\sum_{i=1}^e l_{t,k,i} = 0$, we set their associated delay to zero. For $\sum_{i=1}^e l_{t,k,i} > 0$ and using our original model from \cref{eq:delay}, we obtain the average delay across all jobs processed on a server of type $k$ during time slot $t$ as \begin{align*}
    d_{t,k} := \begin{cases}
        1 / \left(\frac{1}{E Y} - \lambda\right) & \frac{1}{E Y} > \lambda \\
        \infty & \text{otherwise}.
    \end{cases}
\end{align*} Here, we use our assumption that any job can be processed within a single time slot on servers of some type. Note that if $\eta_{k,i} \to \delta$ for some server type $k$ and job type $i$ the average delay will go to infinity as the build-up of long job queues becomes more likely.

\paragraph{Prohibitions} If we want to prohibit altogether that jobs of type $i$ are processed on servers of type $k$, it suffices to set $\delta_{t,k,i} = \infty$ as the revenue loss $r_{t,i}$ is assumed to be convex.

\section{Energy Quotas}\label{section:application:energy_quotas}

We motivated energy quotas in \cref{section:application:operating_cost:energy} and are now ready to present a model that allows for the required flexibility when modeling energy cost. To be as general as possible, we continue to assume a network of $\iota$ data centers with $d$ server types, resulting in a total of $\iota \cdot d$ dimensions. We also continue to examine $\zeta \cdot e$ load types where $\zeta$ is the number of geographically centered job sources and $e$ is the number of job types. We now extend our model from \cref{eq:multiple_load_types_load_balancing} to \begin{align*}
    f_t(x) := \min_{z \in \mathcal{Z}_t} \sum_{j=1}^{\iota} \left(e_{t,j}(x,z) + \sum_{k=1}^{d} h_{t,(j,k)}(x,z)\right)
\end{align*} where $e_{t,j}(x,z)$ is the total energy cost of data center $j$ during time slot $t$. $h_{t,(j,k)}(x,z)$ reduces to \begin{align*}
    \sum_{s=1}^{\zeta} \sum_{i=1}^e l_{t,(j,k),(s,i)} q_{t,(j,k),(s,i)}\left(\frac{l_{t,(j,k)}}{x_{(j,k)}}\right)
\end{align*} for $x > 0$ where $l_{t,(j,k),(s,i)} = \lambda_{t,(s,i)} z_{(s,i)}$ and $l_{t,(j,k)} = \sum_{s=1}^{\zeta} \sum_{i=1}^e l_{t,(j,k),(s,i)} \eta_{(j,k),(s,i)}$. The other cases remain as described in \cref{eq:multiple_load_types_load_balancing_unit}. Note that these definitions of $f_t$ and $h_{t,(j,k)}$ are simply adaptations of our previous definitions from \cref{section:application:dispatching} and, in particular, are convex, increasing, and non-negative.

Given $x \in \mathcal{X}$ and $z \in \mathcal{Z}$, in data center $j$, the average utilization of active servers of type $k$ during time slot $t$ is given as $l_{t,(j,k)} / x_{(j,k)}$. Hence, using optimal load balancing, we obtain the total energy consumption of data center $j$ with \begin{align*}
    \phi'_j(x,z) := \sum_{k=1}^{d} x_{(j,k)} \phi_{k}\left(\frac{l_{t,(j,k)}}{x_{(j,k)} \cdot \delta}\right)
\end{align*} for $x > 0$. The total energy cost of data center $j$ during time slot $t$ is therefore given as $e_{t,j}(x,z) := \nu_{t,j}(\phi'_j(x,z))$. We observe that, on the condition that $\nu_{t,j}$ and $\phi_k$ are convex increasing non-negative functions, the same holds for $e_{t,j}$. $\nu_{t,j}$ now receives the entire energy consumption at a location during time slot $t$ as input.

\paragraph{Maximum Quotas} We are thus able to model more complex energy costs. We consider by $p_{t,i,j} \in \mathbb{R}_{\geq 0} \cup \{\infty\}$ the energy from source $i \in [\xi]$ available at data center $j$ during time slot $t$. We denote by $c_{t,i}$ the average cost per unit of energy of energy source $i$ during time slot $t$ and assume without loss of generality that $c_{t,1} \leq c_{t,2} \leq \cdots$ holds for $t \in T$. A reasonable model of energy cost would then be to use the sources of energy in order of cost until the energy demand is satisfied. We define \begin{align*}
    \delta_{t,i,j} := (p - \sum_{i'=1}^{i-1} p_{t,i',j})^+
\end{align*} as the remaining energy requirement of data center $j$ during time slot $t$ after all energy sources up to source $i$ were used. Our model is then given by \begin{align*}
    \nu_{t,j}(p) := \sum_{i=1}^{\xi} c_{t,i} \min\{\delta_{t,i,j}, p_{t,i,j}\}
\end{align*} where we assume that the energy supply is \emph{sufficient}\index{sufficient energy supply}, i.e. $p \leq \max_{z \in \mathcal{Z}_t} \phi_j(m,z) \leq \sum_{i=1}^{\xi} p_{t,i,j}$ where we defined $m$ as the vector of upper bounds of each dimension of the decision space. It is easy to see that $\nu_{t,j}$ is continuous, increasing, and convex.

\paragraph{Making Profit} Let us now assume that some energy sources $i \in [\xi]$ are produced at data center $j$, and by selling unused supply, we make an average profit of $u_{t,i}$ per unit of energy during time slot $t$. We set $u_{t,i} = 0$ for all external energy sources. Again, we assume without loss of generality that $c_{t,1} + u_{t,1} \leq c_{t,2} + u_{t,2} \leq \cdots$ holds for $t \in T$. The energy cost of this extended model is then given by \begin{align*}
    \nu_{t,j}(p) &:= \sum_{i=1}^{\xi} c_{t,i} \min\{\delta_{t,i,j}, p_{t,i,j}\} - u_{t,i} (p_{t,i,j} - \delta_{t,i,j})^+ \\
                 &= \sum_{i=1}^{\xi} (c_{t,i} + u_{t,i}) \min\{\delta_{t,i,j}, p_{t,i,j}\} - \sum_{i=1}^{\xi} u_{t,i} p_{t,i,j}
\end{align*} where we continue to assume that the energy supply is sufficient. As the first sum is convex and increasing and as the subtracted sum is constant, we know that also $\nu_{t,j}$ must be convex and increasing. For the definition to be non-negative, we require $\sum_{i=1}^{\xi} u_{t,i} p_{t,i,j}$ to be less or equal to $\sum_{i=1}^{\xi} (c_{t,i} + u_{t,i}) \min\{\delta_{t,i,j}, p_{t,i,j}\}$. In words, we require that the possible profit does not exceed the total energy cost of a data center, whether this energy cost is related to energy sources with an associated cost or energy sources with an associated profit that are used rather than sold. Note that this requirement is natural in the context of data centers, as the main purpose behind on-site energy production is to sustain the data center and interaction with the grid merely happens to offset variances in power generated on-site.

\paragraph{Minimum Quotas} One could also imagine a scenario where we seek to impose minimum quotas, requiring that some energy sources make up at least a certain fraction of the total energy use of a data center. While such quotas are relevant, as research on green data centers shows (as discussed in \cref{section:application:operating_cost:energy}), they are inherently soft. In contrast, the maximum quotas are hard as they enforce a physical limitation of energy supply. Hence, in most cases, it is beneficial to model tendencies towards some energy sources using costs rather than strict quotas to largely decouple the energy cost model from energy availability. The cost model can then be adapted to result in the desired makeup of energy sources.


\paragraph{} We have seen an expressive framework for cost models that specifically model data centers but are general enough to support a network of heterogeneous data centers with heterogeneous loads and flexible energy costs. This reaches our main goal for this chapter.

\chapter{Theory}\label{chapter:theory}

In this chapter, we introduce the theoretical foundations for our subsequent work. We begin by formally introducing the metrics we use to assess the performance of the discussed algorithms. We then introduce the problem of smoothed convex optimization and related variants that the examined algorithms address.

\section{Performance Metrics}\label{section:theory:performance_metrics}

We say that an algorithm is \emph{optimal}\index{optimal algorithm} with respect to some performance metric if no algorithm can achieve a better score in the given metric given the same information. Crucially, optimality depends on the information given to the algorithm. We thus say that an offline algorithm of a minimization problem is optimal if its result always incurs the smallest possible cost while satisfying the given constraints. In contrast, an optimal online algorithm must not necessarily return the optimal offline solution. In fact, in many cases, online algorithms must necessarily perform worse than optimal offline algorithms due to the lack of provided information (in our case, the convex cost functions arrive over time). Naturally, these performance metrics also inform parts of our experimental analysis performed in \cref{chapter:case_studies}.

\subsection{Approximations}

We begin by considering the offline case. In most cases, we seek to find optimal solutions to the offline problem. However, for some problems where the computational complexity of optimal solutions is high for large instances, it is beneficial to consider efficient algorithms that achieve close to optimal performance, motivating the definition of approximation algorithms. Note that we limit our definitions of performance metrics to minimization problems.

\begin{definition}\index{approximation ratio}
\cite{Williamson2011} An $\alpha$-approximation algorithm for a minimization problem $c$ is a polynomial-time algorithm $ALG$ that for all instances of the problem produces a solution whose value is within a factor of $\alpha$ of the value of an optimal solution $OPT$, i.e., $c(ALG) \leq \alpha \cdot c(OPT)$.
\end{definition}

In other words, an $\alpha$-approximation guarantees that its results are at most a factor of $\alpha$ worse than the optimal solution. The integral smoothed convex optimization problem is an example where \citeauthor{Kappelmann2017}~\cite{Kappelmann2017} and \citeauthor{Albers2021_2}~\cite{Albers2021_2} recently made substantial progress on approximation algorithms.

\subsection{Competitiveness}

For online algorithms, it is natural to consider an adaptation of the idea of approximation algorithms. Here, we compare the result of an online algorithm with the optimal offline solution.

\begin{definition}\index{competitive ratio}
An $\alpha$-competitive online algorithm for a minimization problem $c$ is an algorithm $ALG$ that for all instances of the problem produces a solution whose value is within a factor of $\alpha$ of the value of an optimal offline solution $OPT$, i.e. $c(ALG) \leq \alpha \cdot c(OPT)$.
\end{definition}

We observe that this definition is analogous to our earlier definition of approximation algorithms in the offline case. However, in contrast to approximation algorithms, where the limiting factor was the algorithm's complexity, the competitiveness of online algorithms is fundamentally restricted by the information available to an online algorithm compared to its offline variants. In smoothed convex optimization, considerable work has focused on finding online algorithms with a constant competitive ratio in the number of dimensions $d$.

Crucially, the competitiveness of an online algorithm depends on the assumed adversary model. Commonly, three adversary models are used in the literature, which are described by \citeauthor{Borodin1990}~\cite{Borodin1990}. First, the \emph{oblivious adversary}\index{oblivious adversary} is the weakest adversary and only knows the algorithm's code but needs to construct the request sequence before any moves are made. Second, the \emph{adaptive online adversary}\index{adaptive online adversary} makes the next request based on the algorithm's previous answers but serves it immediately. Third, the \emph{adaptive offline adversary}\index{adaptive offline adversary} is the strongest adversary that serves the requests based on the algorithm's previous answers but, in the end, can choose the optimal request sequence among all possible request sequences. Note that as all adversaries know the algorithm's code, they are equivalent in the case of a deterministic algorithm. Also, note that randomization is not helpful when playing against an adaptive offline adversary. In the case of many smoothed online convex optimization problems, including right-sizing data centers, it is reasonable to assume an oblivious adversary as typically incoming requests arrive independently from previous server configurations in a data center.

\subsection{Regret}

Regret is another approach to measuring the performance of online algorithms.

\begin{definition}\index{regret (static)}
The (static) regret of an online algorithm $ALG$ for a minimization problem $c$ is $\rho(T)$ if for all instances of the problem the difference between the result of the algorithm and the static optimal offline solution $OPT_s$ does not exceed $\rho(T)$, i.e. $c(ALG) - c(OPT_s) \leq \rho(T)$.
\end{definition}

Commonly, the literature considers this definition of regret where the online algorithm is compared against a static offline solution, i.e., a solution where the agent is not allowed to move in the decision space. We say that an algorithm achieves \emph{no-regret}\index{no-regret} if $\rho$ is sublinear in the time horizon $T$. Observe that ideally, an algorithm achieves negative regret, in which case it performs better than the static optimum.

Ideally, online algorithms perform well with respect to the competitive ratio and regret. In other words, our online algorithms should both perform well compared against an agent that is moving in the decision space with perfect knowledge of the future (competitive ratio) and perform well against an agent that picks one optimal location in the decision space. In practice, for the example of dynamically right-sizing a data center, our algorithms are required to outperform a static number of servers to be viable alternatives. In contrast, to minimize energy waste and revenue loss, the strategies proposed by our algorithms must be as close as possible to the optimal dynamic strategies.

However, \citeauthor{Andrew2015}~\cite{Andrew2015} proved that no online algorithm for smoothed convex optimization can simultaneously achieve a constant competitive ratio and no-regret even when $d = 1$ and cost functions are linear. The competitive ratio can be arbitrarily poor for no-regret algorithms by oscillating the dynamic optimal solution between two points in the decision space. The no-regret algorithm will approach a static optimum, which can be arbitrarily worse than the dynamic optimum. In contrast, a constant-competitive algorithm generally sticks to a point in the decision space until it knows that the cost of movement is outweighed by the reduced cost of some other point that it then moves to. Hence, for constant-competitive algorithms, the regret can be arbitrarily large as the algorithm oscillates between the two points in the decision space, in each step having an arbitrary distance to the static optimum. An illustrative example with one dimension and a periodic sequence of two linear cost functions is depicted in \cref{fig:incompatibility_of_competitive_ratio_and_regret}.

\begin{figure}
    \centering
    \input{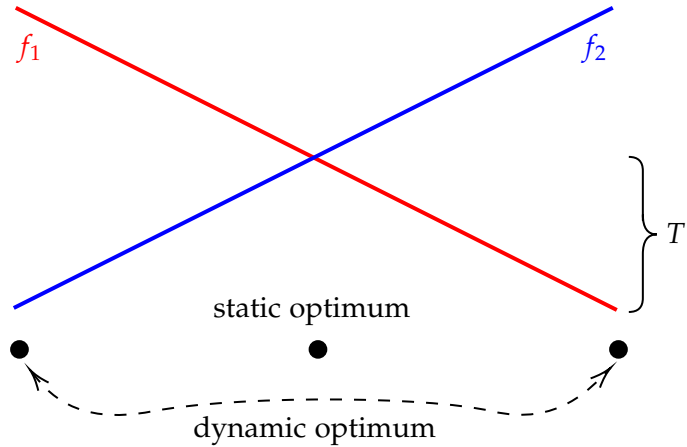}
    \caption{Incompatibility of competitive ratio and regret in one dimension. Consider an adversary playing two linear cost functions $f_1$ and $f_2$ with different minimizers. Then, the dynamic optimum oscillates between the two minimizers while the static optimum is given by the intersection of the two cost functions. Therefore, a no-regret algorithm may be arbitrarily far away from the dynamic offline optimum. In contrast, a competitive algorithm which sticks to either end of the decision space may exceed the static offline optimum by a delta that is not in $\mathcal{O}(T)$ \cite{Wierman2019}.}
    \label{fig:incompatibility_of_competitive_ratio_and_regret}
\end{figure}

There thus exist many variants of regret and the competitive ratio used to bridge between the two metrics. One approach considers an additive variant of the competitive ratio, which is called the \emph{competitive difference}.

\begin{definition}\index{competitive difference}
\cite{Chen2015} The competitive difference of an online algorithm $ALG$ for a minimization problem $c$ is $\rho(T)$ if for all instances of the problem, the difference between the result of the algorithm and the dynamic optimal offline solution $OPT$ does not exceed $\rho(T)$, i.e., $c(ALG) - c(OPT) \leq \rho(T)$.
\end{definition}

This definition is also known as \emph{dynamic regret}\index{dynamic regret}~\cite{Chen2018}. We next define a variant of regret that bridges between static and dynamic regret.

\begin{definition}\index{constrained dynamic regret}
\cite{Chen2018} The $L$-constrained dynamic regret of an online algorithm $ALG$ for a minimization problem $c$ is $\rho(T)$ if for all instances of the problem the difference between the result of the algorithm and the $L$-constrained optimal offline solution $OPT_L$ does not exceed $\rho(T)$, i.e. $c(ALG) - c(OPT_L) \leq \rho(T)$.

The $L$-constrained optimal offline solution minimizes $c$ subject to the additional constraint \begin{align*}
    \sum_{t=1}^T \norm{X_t - X_{t-1}} \leq L
\end{align*} for $X_t, X_{t-1} \in \mathcal{X}$.
\end{definition}

We now observe that given the optimal offline solution $OPT$ with schedule $\hat{X}_t$, the $L$-constrained dynamic regret is equivalent to dynamic regret for $L = \sum_{t=1}^T \norm{\hat{X}_t - \hat{X}_{t-1}}$. In contrast, given the static optimal offline solution $OPT_s$, the $L$-constrained dynamic regret is equivalent to static regret for $L = \norm{OPT_s - \mathbf{0}}$ which is the initial (and only) step of $OPT_s$~\cite{Chen2018}.

Another metric used to bridge the gap between the competitive ratio and regret is the \emph{$\alpha$-unfair competitive ratio} which penalizes movement in the decision space by an additional factor $\alpha$~\cite{Andrew2015}.

\begin{definition}\index{unfair competitive ratio}
\cite{Andrew2015} The $\alpha$-unfair competitive ratio of an online algorithm $ALG$ for a minimization problem $c$ is $\beta$ if for all instances of the problem the ratio of the result of the algorithm and the dynamic $\alpha$-unfair offline solution $OPT_{\alpha}$ does not exceed $\beta$, i.e. $c(ALG) \leq \beta \cdot c(OPT_{\alpha})$.

Here, the $\alpha$-unfair optimal offline solution $OPT_{\alpha}$ is defined as the minimizer of \begin{align*}
    \sum_{t=1}^T f_t(X_t) + \alpha \norm{X_t - X_{t-1}}.
\end{align*}
\end{definition}

Note that for $\alpha = 1$, the $\alpha$-unfair competitive ratio is equivalent to the competitive ratio. For large $\alpha$, the $\alpha$-unfair optimal offline solution $OPT_{\alpha}$ is similar to the $L$-constrained optimal offline solution $OPT_L$ in that the movement in the decision space is restricted.

\section{Problems}

Now that we have an overview of the commonly used performance metrics, we introduce the problems we consider in this work. We initially state the problems as offline problems, but as all problems follow the same structure, their corresponding online variant is obtained by deferring the convex cost functions. All other problem variables --- except for the time horizon $T$ --- such as the movement cost that penalizes movement in the decision space are known from the beginning.

\subsection{Smoothed Convex Optimization}

We begin by formally introducing the most general problem we consider and which we already motivated in \cref{chapter:introduction}.

\begin{problem}[Smoothed Convex Optimization (SCO)]\index{smoothed convex optimization}\label{problem:smoothed_convex_optimization}
Given a time horizon $T \in \mathbb{N}$, a convex decision space $\mathcal{X} \subset \mathbb{R}^d$, a norm $\norm{\cdot}$ on $\mathbb{R}^d$, and a sequence $F$ of non-negative convex functions $f_t$ for $t \in [T]$ with $f_t(x) = \infty$ for all $x \not\in \mathcal{X}$, find $X \in \mathcal{X}^T$ minimizing \begin{align*}
    c_{\text{SCO}}(X) = \sum_{t=1}^T f_t(X_t) + \norm{X_t - X_{t-1}}
\end{align*}
where $X_0 = \mathbf{0}$.
\end{problem}

In many practical applications of smoothed convex optimization, we seek to find integral solutions minimizing hitting and movement costs. This is especially true within the context of resource allocation, for example, for right-sizing data centers, where our resources are discrete. This observation motivates the definition of the following variant of SCO.

\begin{problem}[Integral Smoothed Convex Optimization (Int-SCO)]
We define integral smoothed convex optimization analogously to SCO with the added restriction that the points $x$ in $d$-dimensional space must be discrete, that is $\mathcal{X} \subset \mathbb{Z}^d$.
\end{problem}

In this work, we often refer to the convex cost functions of fractional problems as hitting costs, whereas we generally refer to them as operating costs in the context of integral problems.

In \cref{chapter:introduction}, we have seen that metrical task systems subsume Int-SCO. However, it was shown that, in general, the competitiveness of deterministic and randomized algorithms for metrical task systems must be proportional to the size of the decision space~\cite{Blum1992, Borodin1992}. Further, \citeauthor{Chen2018}~\cite{Chen2018} have shown that the competitiveness of any online algorithm for SCO is lower bounded by $\Omega(\sqrt{d})$. Therefore, many of the online algorithms for SCO that we examine in \cref{chapter:online_algorithms} further restrict hitting and movement costs.

Another similar problem is the ski rental problem. In the \emph{ski rental problem}\index{ski rental problem}, skis can be bought for a cost of $b$ units or rented for a cost of one unit per day. Each day of the ski season, the agent has to decide whether to rent the skis or end the sequence of decisions by buying the skis without knowing how long the ski season will last~\cite{Shah2021}. Consider the uni-dimensional decision space $\{0,b\}$, the $\ell_2$ norm as movement cost, and the sequence of hitting costs $f_t(0) = 1$ and $f_t(b) = 0$. The solution to this instance of SCO is a solution to the corresponding ski rental problem, yielding that the ski rental problem is a special case of SCO. \citeauthor{Karlin1990}~\cite{Karlin1990} showed that the best competitive ratio attainable by a randomized algorithm is $e/(e-1) \approx 1.58$, giving a lower bound for the competitive ratio of online algorithms for SCO.

\citeauthor{Goel2019}~\cite{Goel2019} proved that for $\alpha$-strongly convex hitting costs with respect to the $\ell_2$ norm and $\ell_2$-squared movement costs, the optimal competitiveness of any online algorithm is $\mathcal{O}(1/\sqrt{\alpha})$ as $\alpha \downarrow 0$. We discuss hitting costs and movement costs of this shape in greater detail in \cref{section:theory:beyond_convexity}. \citeauthor{Bansal2015}~\cite{Bansal2015} have shown that in the uni-dimensional setting, the optimal competitive ratio that a deterministic memoryless algorithm for SCO can attain is three.

\subsubsection{Complexity of the Offline Problem}

We now want to examine the complexity of Int-SCO in the offline case. That is, we know all arriving convex cost functions $f_t$ in advance. We prove Int-SCO NP-hard for varying $d$ by giving a polynomial-time reduction from the Knapsack problem. In \cref{section:theory:simplified_smoothed_convex_optimization}, we extend this proof of NP-hardness to the integral simplified smoothed convex optimization problem, further restricting the decision space and movement cost.

Given a set of items with an associated value and weight and an upper bound to the total weight, Knapsack is the problem of determining the number of copies of each item that maximizes the total value and conforms to the given upper bound on total weight. Formally we define Knapsack as follows.

\begin{problem}[Knapsack (KP)]\index{knapsack problem}
Given a number of items $n \in \mathbb{N}$, a value of each item $v \in \mathbb{N}^n$, a weight of each item $w \in \mathbb{N}^n$, and an upper bound to the total weight $W \in \mathbb{N}$, find $x \in \{0,1\}^n$ satisfying $\sum_{i = 1}^n w_i x_i \leq W$ and maximizing $\sum_{i=1}^n v_i x_i$.
\end{problem}

This variant of Knapsack is commonly called \emph{0-1 Knapsack} and restricts the number of copies of each item to zero or one. It is, however, easy to see that our proof can be generalized to a setting where we allow $x_i \in [m_i]_0$ for $m \in \mathbb{N}^n$. \citeauthor{Williamson2014}~\cite{Williamson2014} gives a proof for the NP-completeness of the Knapsack decision problem. It immediately follows that the Knapsack optimization problem is NP-hard.

Before reducing to Int-SCO, we reduce Knapsack to a related problem called Minimum Knapsack.

\begin{problem}[Minimum Knapsack (Min-KP)]\index{minimum knapsack problem}
Given a number of items $n \in \mathbb{N}$, a cost of each item $c \in \mathbb{N}^n$, a utility of each item $u \in \mathbb{N}^n$, and a lower bound to the total utility $U \in \mathbb{N}$, find $x \in \{0,1\}^n$ satisfying $\sum_{i = 1}^n u_i x_i \geq U$ and minimizing $\sum_{i=1}^n c_i x_i$.
\end{problem}

\begin{lemma}
Min-KP is NP-hard.
\end{lemma}
\begin{proof}
We prove the lemma by giving a reduction from KP.

Let $\mathcal{I}_{\text{KP}} = (n, v, w, W)$ be an instance of KP. Let $\mathcal{I}_{\text{Min-KP}}(U) = (n, c, u, U)$ be an instance of Min-KP with $c = w$, and $u = v$. Hence, $\mathcal{I}_{\text{Min-KP}}(U)$ minimizes the total weight $\sum_{i=1}^n w_i x_i$ such that $\sum_{i=1}^n v_i x_i \geq U$.

By finding solutions to $\mathcal{I}_{\text{Min-KP}}(U)$ repeatedly for varying $U$, we determine the maximal $U$ such that $\sum_{i=1}^n w_i x_i \leq W$. We observe that $U$ is upper bounded by $n \cdot v_{\text{max}}$. If $U$ were greater than $n \cdot v_{\text{max}}$ we would have $\sum_{i=1}^n v_{\text{max}} x_i \geq \sum_{i=1}^n v_i x_i > n \cdot v_{\text{max}}$ which contradicts $x \in \{0,1\}^n$. Hence, we can use binary search to find $U$ in $\mathcal{O}(\log n + \log v_{\text{max}})$ iterations. The other direction works analogously.

We have seen a total, polynomial-time reduction from KP to Min-KP. Hence, Min-KP is NP-hard.
\end{proof}

Next, we prove our central reduction from Min-KP to Int-SCO. To motivate this reduction, we first prove that the following (convex) integer optimization is, in fact, equivalent to Min-KP.

\begin{lemma}
\label{lemma:integer_minimization}
Let $\mathcal{I}_{\text{Min-KP}} = (n, c, u, U)$ be an instance of Min-KP. $x$ is the solution to $\mathcal{I}_{\text{Min-KP}}$ if and only if $x$ minimizes \begin{align*}
    c_{\text{SCO}}'(x) = \sum_{i=1}^n c_i x_i + M\left(U - \sum_{i=1}^n u_i x_i\right)^+
\end{align*} subject to $x \in \{0,1\}^n$ for some $M > \frac{n c_{\text{max}}}{u_{\text{min}}}$.
\end{lemma}
\begin{proof}
Suppose $x$ minimizes $c_{SCO}'(x)$. Now suppose $(U - \sum_{i=1}^n u_i x_i)^+ > 0$. Then $\sum_{i=1}^n u_i < U$ follows immediately. It is easy to see that if $x \equiv 1$, $\mathcal{I}$ has no solution because the lower bound on the utility $U$ is not met. Henceforth, we assume $x$ can be further increased. Then, $(U - \sum_{i=1}^n u_i x_i)^+ \geq u_{\text{min}}$. Therefore, $c_{SCO}'(x) > \sum_{i=1}^n c_i x_i + c_{\text{max}}$. We observe that $x$ is not optimal as $c_{SCO}'(x)$ could be minimized further by increasing $x$ such that $(U - \sum_{i=1}^n u_i x_i)^+ = 0$ since $\sum_{i=1}^n c_i x_i \leq n c_{\text{max}}$ holds for all $x$.

By leading our previous assumption to a contradiction, we conclude $(U - \sum_{i=1}^n u_i x_i)^+ = 0$ and therefore $U \leq \sum_{i=1}^n u_i x_i$. Further, $c_{SCO}'(x)$ minimizes $\sum_{i=1}^n c_i x_i$ for all remaining candidates for $x$. Hence, $x$ is the solution of $\mathcal{I}_{\text{Min-KP}}$.

On the other hand, suppose that $x$ is the solution to $\mathcal{I}_{\text{Min-KP}}$. Then $(U - \sum_{i=1}^n u_i x_i)^+ = 0$ and $\sum_{i=1}^n c_i x_i$ is minimized. Hence, $x$ minimizes $c_{SCO}'(x)$.
\end{proof}

For our construction we need that $c_{SCO}'$ is convex.

\begin{lemma}
\label{lemma:integer_minimization_convexity}
$c_{SCO}'$ is convex on $\{0,1\}^n$.
\end{lemma}
\begin{proof}
It is easy to see that $c_{SCO}'$ is continuous. Therefore, to show the convexity of $c_{SCO}'$ it suffices to prove midpoint-convexity, i.e. $c_{SCO}'\left(\frac{x+y}{2}\right) \leq \frac{c_{SCO}'(x)+c_{SCO}'(y)}{2}$ for all $x, y \in \mathbb{R}^n$.

To simplify the notation let $C(x) = \sum_{i=1}^n c_i x_i$ and let $U(x) = \sum_{i=1}^n u_i x_i$. To further simplify the notation we define $\frac{x+y}{2}$ to be applied component-wise to elements $i \in [n]$ of $x$ and $y$. We then obtain \small{
\begin{align*}
         &c_{SCO}'\left(\frac{x+y}{2}\right) \leq \frac{c_{SCO}'(x)+c_{SCO}'(y)}{2} \\
    \iff &C\left(\frac{x+y}{2}\right) + M\left(U - U\left(\frac{x+y}{2}\right)\right)^+ \leq \frac{C(x) + M(U - U(x))^+ + C(y) + M(U - U(y))^+}{2} \\
    \iff &C(x) + C(y) + 2M\left(U - U\left(\frac{x+y}{2}\right)\right)^+ \leq C(x) + M(U - U(x))^+ + C(y) + M(U - U(y))^+ \\
    \iff &2\left(U - U\left(\frac{x+y}{2}\right)\right)^+ \leq (U - U(x))^+ + (U - U(y))^+.
\end{align*}
}\normalsize

We immediately get the convexity of $U(\cdot)$ by the following equivalence. \begin{align*}
    U\left(\frac{x+y}{2}\right) &= \sum_{i=1}^n u_i \frac{x_i + y_i}{2} \\
                                &= \frac{\sum_{i=1}^n u_i x_i + \sum_{i=1}^n u_i y_i}{2} \\
                                &= \frac{U(x) + U(y)}{2}.
\end{align*}

Now, we consider three cases separately.

\begin{enumerate}
    \item If $U(x) > U$ and $U(y) > U$, then $U\left(\frac{x+y}{2}\right) > U$. Hence \begin{align*}
        2\left(U - U\left(\frac{x+y}{2}\right)\right)^+ = 0 = (U - U(x))^+ + (U - U(y))^+.
    \end{align*}
    \item If $U(x) \leq U$ and $U(y) \leq U$, then $U\left(\frac{x+y}{2}\right) \leq U$. Hence \begin{align*}
        2\left(U - U\left(\frac{x+y}{2}\right)\right)^+ &= 2U - 2U\left(\frac{x+y}{2}\right) \\
                                                        &= 2U - U(x) - U(y) \\
                                                        &= (U - U(x))^+ + (U - U(y))^+.
    \end{align*}
    \item For the only remaining case we assume w.l.o.g. that $U(x) \leq U$ and $U(y) > U$. If $U - U(x) < U(y) - U$, then $U\left(\frac{x+y}{2}\right) > U$ and we follow the first case. If, on the other hand, $U - U(x) \geq U(y) - U$, then $U\left(\frac{x+y}{2}\right) \leq U$ and we follow the second case.\qedhere
\end{enumerate}
\end{proof}

We now have everything in place to prove our main result of this section.

\begin{theorem}
Int-SCO is NP-hard.
\end{theorem}
\begin{proof}
We now give our reduction from Min-KP to Int-SCO.

Let $\mathcal{I}_{\text{Min-KP}} = (n, c, u, U)$ be an instance of Min-KP and set $d = n$. We define $\mathcal{I}_{\text{Int-SCO}} = (T, \mathcal{X}, \norm{\cdot}, f)$ as an instance of Int-SCO with $T = 1$, $\mathcal{X} = \{0,1\}^n$, $\norm{\cdot} = 0$, and $f_1(x) = c_{\text{SCO}}'(x)$. It is easy to see that $f_1$ is non-negative. By \cref{lemma:integer_minimization_convexity}, $\mathcal{I}_{\text{Int-SCO}}$ is a valid instance of Int-SCO.

The correctness of our construction follows from \cref{lemma:integer_minimization}. \begin{align*}
         &X \text{ is a solution to } \mathcal{I}_{\text{Int-SCO}} \\
    \iff &X \text{ minimizes } \sum_{t=1}^T f_t(X_t) + \norm{X_t - X_{t-1}} \text{ such that } X_t \in \mathcal{X}. \\
    \iff &X \text{ minimizes } c_{\text{SCO}}'(X_1) \text{ such that } X_1 \in \{0,1\}^n. \\
    \iff &X_1 \text{ is a solution to } \mathcal{I}_{\text{Min-KP}}.
\end{align*}

Our construction is total and polynomial in the size of $\mathcal{I}_{\text{Min-KP}}$. Hence, Int-SCO is NP-hard.
\end{proof}

We observe that the above reduction can be extended to Knapsack with arbitrary bounds $m_i$ by setting $\mathcal{X}$ of $\mathcal{I}_{\text{Int-SCO}}$ to $[m_1]_0 \times \dots \times [m_n]_0$.

\subsection{Simplified Smoothed Convex Optimization}\label{section:theory:simplified_smoothed_convex_optimization}

In many applications, for example, for right-sizing data centers where we are interested in determining the optimal number of servers to run at a particular time, it suffices to restrict $\mathcal{X}$ to $[m_0]_0 \times \dots \times [m_d]_0$ for some upper bound in each dimension $m \in \mathbb{N}^d$ and the switching cost $\norm{\cdot}$ to a Manhattan norm which is scaled in each dimension independently from time. To that end, we first define a restricted variant of (fractional) SCO, which we term \emph{simplified smoothed convex optimization}.

\begin{problem}[Simplified Smoothed Convex Optimization (SSCO)]\index{simplified smoothed convex optimization}\label{problem:simplified_smoothed_convex_optimization}
Given a time horizon $T \in \mathbb{N}$, upper bounds $m \in \mathbb{N}^d$, switching costs $\beta \in \mathbb{R}_{>0}^d$, and a sequence $F$ of non-negative convex functions $f_t$ for $t \in [T]$, find $X \in (\mathbb{R}_{\geq 0, \leq m_0} \times \dots \times \mathbb{R}_{\geq 0, \leq m_d})^T$ minimizing \begin{align}\label{eq:simplified_smoothed_convex_optimization}
    c_{\text{SSCO}}(X) = \sum_{t=1}^T f_t(X_t) + \sum_{k=1}^d \beta_k (X_{t,k} - X_{t-1,k})^+
\end{align}
where $X_0 = \mathbf{0}$.
\end{problem}

We observe that $c_{\text{SSCO}}$ pays the switching cost whenever $x$ increases. Decreasing $x$ does not increase the paid switching cost. This observation motivates the following lemma that shows that we could equivalently pay the switching cost for decreasing $x$.

\begin{lemma}
\label{lemma:inverse_switching_cost}
For all $T \in \mathbb{N}$ and $ X_t \in \mathbb{R}$ where $X_0 = X_{T+1} = \mathbf{0}$, the following equivalence holds:
\begin{align*}
    \sum_{t=1}^T (X_t - X_{t-1})^+ = \sum_{t=1}^T (X_t - X_{t+1})^+.
\end{align*}
\end{lemma}
\begin{proof}
The left side of the equation sums all increases in $x$ from $t \in \{0, \dots, T\}$ starting from $X_0 = \mathbf{0}$. The right side of the equation sums all decreases in $x$ from $t \in \{1, \dots, T+1\}$ ending with $X_{T+1} = \mathbf{0}$. As the schedule $X$ begins and ends with the configuration $\mathbf{0}$, the two sums are equivalent.
\end{proof}

To complete the proof that any instance of SSCO is an instance of SCO, we have to show that our switching cost is indeed a valid norm. Given an instance $\mathcal{I}_{\text{SSCO}} = (T, m, \beta, F)$ with $F = (f_1, \dots, f_T)$ we define the corresponding instance of SCO as $\mathcal{I}_{\text{SCO}} = (T, \mathcal{X}, \norm{\cdot}, \widetilde{F})$ where $\mathcal{X} = \mathbb{R}_{\geq 0, \leq m_0} \times \dots \times \mathbb{R}_{\geq 0, \leq m_d}$, $\widetilde{F}$ is slightly modified version of $F$ which is formally defined in the following, and $\norm{x} = \sum_{k=1}^d \frac{\beta_k}{2} |x_k|$ as the dimension-dependently scaled Manhattan norm of $x$. It is easy to see that $\norm{\cdot}$ is indeed a valid norm. The next lemma proves that $X \in \mathcal{X}^T$ is a solution to $\mathcal{I}_{\text{SCO}}$ if and only if it is a solution to $\mathcal{I}_{\text{SSCO}}$.

\begin{lemma}\label{lemma:switching_cost_l1_norm_vs_pos_movement}
For any $T \in \mathbb{N}, \beta \in \mathbb{R}_{>0}^d$, and $X_t \in \mathbb{R}^d$ with $X_0 = X_{T+1} = \mathbf{0}$, the following equivalence holds:
\begin{align}\label{eq:switching_cost_l1_norm_vs_pos_movement}
    \sum_{t=1}^{T+1} \norm{X_t - X_{t-1}} = \sum_{t=1}^T \sum_{k=1}^d \beta_k (X_{t,k} - X_{t-1,k})^+.
\end{align}
\end{lemma}
\begin{proof}
By \cref{lemma:inverse_switching_cost}, the above equivalence holds iff \begin{align*}
    \sum_{t=1}^{T+1} \sum_{k=1}^d \beta_k |X_{t,k} - X_{t-1,k}| = \sum_{t=1}^T \sum_{k=1}^d \beta_k ((X_{t,k} - X_{t-1,k})^+ + (X_{t,k} - X_{t+1,k})^+).
\end{align*}
It is easy to see that this always holds as $(X_{t,k} - X_{t-1,k})^+$ (increasing value) and $(X_{t,k} - X_{t+1,k})^+$ (decreasing value) are the two components of $|X_{t,k} - X_{t-1,k}|$.
\end{proof}

Note that the last summand of the left side of \cref{eq:switching_cost_l1_norm_vs_pos_movement} is $\norm{X_{T+1} - X_T} = \norm{X_T}$ which is not considered in the cost function of SCO. To correct for this under-approximation of the switching cost and to ensure that the cost of a schedule $X$ is equivalent between $\mathcal{I}_{\text{SSCO}}$ and $\mathcal{I}_{\text{SCO}}$ we slightly modify the hitting cost $f_T$ at time $T$ to \begin{align*}
    \widetilde{f}_T(x) := f_T(x) + \norm{x} = f_T(x) + \sum_{k=1}^d \frac{\beta_k}{2} |x|.
\end{align*} The remaining hitting costs remain the same, i.e. $\widetilde{f_t} := f_t$ for all $t \in [T-1]$. We set $\widetilde{F} = (\widetilde{f}_1, \dots \widetilde{f}_T)$. We also observe that the $\norm{\cdot}$ is convex, increasing, and non-negative, implying that this slight modification maintains the invariant that the hitting costs likewise are convex, increasing, and non-negative. This slight modification of the final hitting cost is only relevant in the offline setting where the time horizon is known.

With the same motivation we used for the restriction of SCO to Int-SCO, we now restrict SSCO to an integral variant.

\begin{problem}[Integral Simplified Smoothed Convex Optimization (Int-SSCO)]
We define integral simplified smoothed convex optimization analogously to SSCO with the added restriction that the points $x$ in $d$-dimensional space must be discrete, that is $x \in [m_0]_0 \times \dots \times [m_d]_0$.
\end{problem}

\citeauthor{Albers2018}~\cite{Albers2018} have shown for Int-SSCO in the uni-dimensional setting that the optimal competitive ratio is 3 for deterministic algorithms and 2 for randomized algorithms. As Int-SSCO is subsumed by Int-SCO, these bounds also hold for Int-SCO.

\subsubsection{Complexity of the Offline Problem}

We next extend our proof of NP-hardness of Int-SCO for varying $d$ to Int-SSCO. We cannot reuse our original proof as the switching cost of SSCO is required to be positive.

\begin{theorem}
\label{theorem:int_ssco_np_hardness}
Int-SSCO is NP-hard.
\end{theorem}
\begin{proof}
Again, we use a reduction from Min-KP.

Let $\mathcal{I}_{\text{Min-KP}} = (n, c, u, U)$ be an instance of Min-KP and set $d = n$. We define $\mathcal{I}_{\text{Int-SSCO}} = (T, m, \beta, f)$ as an instance of Int-SSCO with $T = 1$, $m \equiv 1$, $\beta \equiv 1$, and $f_1(x) = c_{\text{SCO}}'(x) + n - \sum_{i=1}^n x_i$.

It is easy to see that $f_1$ is non-negative. In \cref{lemma:ssco_reduction_convexity}, we prove that $f_1$ is convex. Assuming the convexity of $f_1$, $\mathcal{I}_{\text{Int-SSCO}}$ is a valid instance of Int-SSCO.

We now prove the correctness of our construction. Again, we use \cref{lemma:integer_minimization}. \begin{align*}
         &X \text{ is a solution to } \mathcal{I}_{\text{Int-SSCO}} \\
    \iff &X \text{ minimizes } \sum_{t=1}^T f_t(X_t) + \sum_{k=1}^d \beta_k (X_{t,k} - X_{t-1,k})^+ \text{ such that } X_t \in [m_0]_0 \times \dots \times [m_d]_0. \\
    \iff &X \text{ minimizes } f_1(X_1) + \sum_{i=1}^n X_{1,i} \text{ such that } X_1 \in \{0,1\}^n. \\
    \iff &X \text{ minimizes } c_{\text{SCO}}'(X_1) + n + \sum_{i=1}^n X_{1,i} - X_{1,i} \text{ such that } X_1 \in \{0,1\}^n. \\
    \iff &X \text{ minimizes } c_{\text{SCO}}'(X_1) \text{ such that } X_1 \in \{0,1\}^n. \\
    \iff &X_1 \text{ is a solution to } \mathcal{I}_{\text{Min-KP}}.
\end{align*}

Our construction is still total and polynomial in the size of $\mathcal{I}_{\text{Min-KP}}$. Hence, Int-SSCO is NP-hard.
\end{proof}

\begin{lemma}
\label{lemma:ssco_reduction_convexity}
$f_1$ from \cref{theorem:int_ssco_np_hardness} is convex on $\{0,1\}^n$.
\end{lemma}
\begin{proof}
To show convexity of $f_1$, it suffices to show that $h(x) = n - \sum_{i=1}^n x_i$ is convex as the convexity of $c_{SCO}'$ was already established in \cref{lemma:integer_minimization_convexity} and $f_1(x) = c_{SCO}'(x) + h(x)$. Further, it is enough to prove $h$ midpoint-convex as $h$ is continuous. We observe that \begin{align*}
         &&h\left(\frac{x + y}{2}\right) &\leq \frac{h(x) + h(y)}{2} \\
    \iff &&n - \sum_{i=1}^n \frac{x_i + y_i}{2} &\leq n - \frac{\sum_{i=1}^n x_i + \sum_{i=1}^n y_i}{2}
\end{align*} holds for any $x, y \in \{0,1\}^n$, proving the lemma.
\end{proof}

\subsection{Smoothed Balanced Load Optimization}

We now turn to a variant of Int-SSCO introduced by \citeauthor{Albers2021_2}~\cite{Albers2021_2} that further restricts the structure of the convex cost functions. This restriction is motivated by the usual cost model of heterogeneous data centers with homogeneous loads we examined in detail in \cref{section:application:dispatching:optimal_load_balancing} where the incoming load (or set of jobs) is distributed equally among all active servers.

Given a sequence of convex increasing non-negative costs $g_{t,k}(l)$ of each instance in dimension $k$ given its load $l$ during time slot $t$, the overall cost for dimension $k$ during time slot $t$ is given as \begin{align*}
    h_{t,k}(x,z) := \begin{cases} 
        x g_{t,k}\left(\frac{l_{t,k}}{x}\right) & x > 0 \\
        \infty                                  & x = 0 \land l_{t,k} > 0 \\
        0                                       & x = 0 \land l_{t,k} = 0
    \end{cases}
\end{align*} where $l_{t,k} = \lambda_t z$ for some sequence of load profiles $\lambda_t \in \mathbb{N}_0$. Here, $x$ is the position in the decision space in dimension $k$, and $z \in [0,1]$ is the fraction of the load $\lambda_t$ that is assigned to dimension $k$~\cite{Albers2021_2}. Given the set of all possible assignments to $d$ dimensions $\mathcal{Z} := \{z \in [0,1]^d \mid \sum_{k=1}^d z_k = 1\}$, the overall hitting cost is defined as the convex optimization \begin{align}
\label{eq:sblo_hitting_cost}
    f_t(x) := \min_{z \in \mathcal{Z}} \sum_{k=1}^d h_{t,k}(x_k,z_k).
\end{align} Intuitively, the load profiles $\lambda_t$ are balanced across all dimensions so as to minimize cost. We also observe that the formulation of $f_t$ from \cref{eq:sblo_hitting_cost} is equivalent to our formulation from \cref{eq:heterogeneous_load_balancing}.

\begin{problem}[Smoothed Balanced Load Optimization (SBLO)]\index{smoothed balanced load optimization}\label{problem:sblo}
Given a time horizon $T \in \mathbb{N}$, upper bounds $m \in \mathbb{N}^d$, switching costs $\beta \in \mathbb{R}_{>0}^d$, a sequence $\Lambda$ of load profiles $\lambda_t \in \mathbb{N}_0$, and a sequence $G$ of convex increasing non-negative functions $g_{t,k}$ for $t \in [T], k \in [d]$, find $X \in ([m_0]_0 \times \dots \times [m_d]_0)^T$ minimizing \begin{align*}
    c_{\text{SBLO}}(X) = \sum_{t=1}^T f_t(X_t) + \sum_{k=1}^d \beta_k (X_{t,k} - X_{t-1,k})^+
\end{align*}
where $X_0 = \mathbf{0}$ and $f_t$ is given by \cref{eq:sblo_hitting_cost}.
\end{problem}

In the online variant of SBLO (and in its variants), load profiles and convex cost functions arrive over time. It is easy to see that any instance of SBLO is in fact an instance of Int-SSCO.

\subsection{Smoothed Load Optimization}

Lastly, we consider an even simpler problem proposed by \citeauthor{Albers2021}~\cite{Albers2021} where we assume that each instance can only handle a single job during each time slot. Instead of using convex functions to model cost, we assume that cost increases linearly and independently from time with the number of active servers. In addition, we impose the constraint that the number of servers must still be enough to handle the incoming load. Without this restriction, the optimal strategy would always be not to run any servers at all.

\begin{problem}[Smoothed Load Optimization (SLO)]\index{smoothed load optimization}\label{problem:slo}
Given a time horizon $T \in \mathbb{N}$, upper bounds $m \in \mathbb{N}^d$, switching costs $\beta \in \mathbb{R}_{>0}^d$, a sequence $\Lambda$ of load profiles $\lambda_t \in \mathbb{N}_0$, and the non-negative operating costs $c \in \mathbb{R}_{\geq 0}^d$, find $X \in ([m_0]_0 \times \dots \times [m_d]_0)^T$ minimizing \begin{align*}
    c_{\text{SLO}}(X) = \sum_{t=1}^T \sum_{k=1}^d c_k X_{t,k} + \beta_k (X_{t,k} - X_{t-1,k})^+
\end{align*}
where $X_0 = \mathbf{0}$ such that for all $t \in [T]$ \begin{align*}
    \sum_{k=1}^d X_{t,k} \geq \lambda_t.
\end{align*}
\end{problem}

In contrast to our definition of SBLO, SLO balances the load implicitly among all active servers, which is possible because we assume that each active server can only handle a single job during one time slot. Again, it is easy to see that SLO is an instance of SBLO by setting $g_{t,k}(l) := c_k$ for $l \leq 1$ and $g_{t,k}(l) := \infty$ otherwise.

\citeauthor{Albers2021}~\cite{Albers2021} show that an online algorithm for SLO cannot attain a competitive ratio smaller than $2d$.

\section{Beyond Convexity}\label{section:theory:beyond_convexity}

We have seen that the optimal competitiveness of online algorithms for smoothed convex optimization is fundamentally limited to be dimension-dependent as long as arbitrary convex hitting costs and arbitrary norms as movement costs are allowed. In the literature, many promising approaches are based on restricting the class of hitting costs (and movement costs) to achieve a dimension-independent competitive ratio. This section focuses mainly on hitting costs, introducing these restrictions, and investigating how they relate to our data center model.

\subsection{Continuity and Differentiability}\label{section:theory:beyond_convexity:continuity_and_differentiability}

A very first natural restriction is to assume that hitting costs are continuous. In fact, theorem 10.1 of~\cite{Rockafellar1970} proves that given a convex function $f : \mathcal{X} \to \mathbb{R}$, $f$ is continuous on the interior of its domain, $\mathcal{X}^{\circ}$. Throughout this work, we will thus assume the hitting costs to be continuous on the interior of the decision space $\mathcal{X}$ without an explicit mention.

As continuity does not represent a limitation, we investigate the continuous differentiability of the hitting costs. We call a function \emph{smooth}\index{smooth function} if it is infinitely-many times continuously differentiable. However, to obtain a smooth hitting cost in the application of right-sizing data centers, one would have to drastically reduce the complexity of the model we discussed in \cref{chapter:application}. For example, \citeauthor{Bansal2015}~\cite{Bansal2015} focus entirely on the energy cost, which is a good approximation in practice as energy represents the largest fraction of the overall cost. Still, in general, the assumption of smoothness is too strong.

\subsection{Stronger Assumptions}

Some of the algorithms, which we discuss, require a more restricted class of convex cost functions. We thus introduce some terminology that is commonly used in convex optimization to describe properties of well-behaved functions.

\paragraph{Lipschitz Continuity} We begin with the fundamental notion of Lipschitz continuity.

\begin{definition}\index{Lipschitz continuity}
\cite{Gupta2020} A function $f : K \to \mathbb{R}$ is called $L$-Lipschitz over a convex set $K \subseteq \mathbb{R}^d$ with respect to the norm $\norm{\cdot}$ if \begin{align*}
    \norm{f(x) - f(y)} \leq L \norm{x - y}
\end{align*} holds for all $x, y \in K$.
\end{definition}

Intuitively, the absolute slope of an $L$-Lipschitz function cannot be greater than $L$. Alternatively, in other words, the function values cannot change arbitrarily fast.

\paragraph{Lipschitz Smoothness} In the context of convex optimization, there is a notion of smoothness that is distinct from the smoothness that we discussed in \cref{section:theory:beyond_convexity:continuity_and_differentiability}. For descent methods, it is beneficial if the difference in gradients of two points shrinks with the distance between the points. Formally, we define smoothness as follows.

\begin{definition}\index{Lipschitz smoothness}
A function $f : K \to \mathbb{R}$ is called $\beta$-Lipschitz smooth over a convex set $K \subseteq \mathbb{R}^d$ with respect to the norm $\norm{\cdot}$ if its gradient $\nabla f$ is $\beta$-Lipschitz over $K$. Thus, $f$ is $\beta$-Lipschitz smooth if \begin{align*}
    \norm{\nabla f(x) - \nabla f(y)} \leq \beta \norm{x - y}
\end{align*} holds for all $x, y \in K$. This is equivalent to saying that \begin{align*}
    f(y) \leq f(x) + \langle\nabla f(x), y-x\rangle + \frac{\beta}{2}\norm{x-y}^2
\end{align*} holds for all $x, y \in K$~\cite{Gupta2020}.
\end{definition}

Intuitively, this ensures that at any point $x \in K$, a quadratic can be fit above the curve of $f$. This ensures that a descent method does not ``overshoot" when approaching the minimum because the gradient decreases as the minimum is approached.

\paragraph{Strong Convexity} In contrast, descent methods converge faster if gradients are large, very far away from the optimal solution. The notion of strong convexity describes this property.

\begin{definition}\index{strong convexity}
\cite{Gupta2020} A function $f : K \to \mathbb{R}$ is called $\alpha$-strongly convex over a convex set $K \subseteq \mathbb{R}^d$ with respect to the norm $\norm{\cdot}$ if $g(x) = f(x) - \frac{\alpha}{2}\norm{x-y}^2$ is convex over $K$. Equivalently, $f$ is $\alpha$-strongly convex if \begin{align*}
    f(y) \geq f(x) + \langle\nabla f(x), y-x\rangle + \frac{\alpha}{2}\norm{x-y}^2
\end{align*} holds for all $x, y \in K$.
\end{definition}

Intuitively, at any point $x \in K$, a quadratic can be fit under the curve of $f$. In other words, $f$ grows at least quadratically as one moves away from the minimizer. This contrasts our definition of $\beta$-Lipschitz smoothness. When a function is $\alpha$-strongly convex and $\beta$-Lipschitz smooth, descent methods converge quickly as the gradient is large when far away and small when close to the optimal solution.

Note that constant and even linear functions are not strongly convex. Recall that the energy consumption models we discussed in \cref{eq:energy_model:1} and \cref{eq:energy_model:2} were linear in the utilization, implying that the overall operating cost of a server is not strongly convex. Even the non-linear energy consumption model from \cref{eq:energy_model:3} is not strongly convex as its first-order derivative is zero for $s = 0$.

\paragraph{Local Polyhedrality} Still, strong convexity represents a significant restriction. A similar but not quite as strong is the property of local polyhedrality.

\begin{definition}\index{local polyhedrality}
\cite{Goel2018} A function $f : K \to \mathbb{R}$ with minimizer $\hat{x}$ is called $\alpha$-locally polyhedral over a convex set $K \subseteq \mathbb{R}^d$ with respect to the norm $\norm{\cdot}$ if there exists some $\epsilon > 0$ such that \begin{align*}
    f(x) - f(\hat{x}) \geq \alpha \norm{x - \hat{x}}
\end{align*} holds for all $x \in K$ with $\norm{x - \hat{x}} \leq \epsilon$.
\end{definition}

Local polyhedrality indicates that at any point $x \in K$, a linear function with slope $\alpha$ can be fit below the curve of $f$~\cite{Goel2018}. In other words, $f$ grows at least linearly as one moves away from the minimizer. Similar to strong convexity, constant functions are not locally polyhedral. Nevertheless, local polyhedrality encompasses many functions, among others the cost functions we described in \cref{chapter:application} modeling the cost of a data center~\cite{Goel2018}.

\chapter{Offline Algorithms}\label{chapter:offline_algorithms}

To later measure the performance of online algorithms, we must first describe efficient offline algorithms that can be used to find optimal (or nearly optimal) solutions. In this chapter, we, therefore, describe algorithms for fractional and integral smoothed convex optimization. We begin with a general investigation of (fractional) convex optimization. Then, we turn to specific algorithms described in the literature for the uni-dimensional and multi-dimensional settings.

In our general argument, we want to use that the cost function $c_{\text{SCO}}$ is convex on $\mathcal{X}^T$ implying that algorithms for convex optimization can be used to solve the fractional offline case optimally.

\begin{lemma}
$c_{\text{SCO}}(X) = \sum_{t=1}^T f_t(X_t) + \norm{X_t - X_{t-1}}$ is convex on $\mathcal{X}^T$.
\end{lemma}
\begin{proof}
By definition $f_t$ is convex on $\mathcal{X}$. Every norm $\norm{\cdot}$ on $\mathcal{X}$ is convex by the triangle inequality as shown by the following argument: \begin{align*}
    \forall x, y \in \mathcal{X}.\ \forall \lambda \in [0,1].\ \norm{\lambda x + (1 - \lambda) y} \leq \norm{\lambda x} + \norm{(1 - \lambda) y} = \lambda \norm{x} + (1 - \lambda) \norm{y}.
\end{align*} In total we get that the sum of convex functions is also convex.
\end{proof}

\section{Convex Optimization}\label{section:offline_algorithms:convex_optimization}

As the offline variant of SCO simply is the problem of minimizing $c_{\text{SCO}}$, it is easy to see that we can use a convex optimization solver to obtain the optimal schedule $\hat{X}$. Throughout our discussion of algorithms, we denote schedules by $X$ and optimal solutions by $\hat{\cdot}$. Further, as we do not impose any constraints beyond the bounds of the decision space $\mathcal{X} \subset \mathbb{R}^d$ it suffices to find the local optimum of $c_{\text{SCO}}$ with respect to $\mathcal{X}$. By the convexity of $c_{\text{SCO}}$ we know that any local optimum will also be a global optimum~\cite{Bubeck2015}. We can choose the algorithm for finding the local optimum based on our knowledge of the properties of $c_{\text{SCO}}$. If, for example, $c_{\text{SCO}}$ is differentiable on $\mathcal{X}$ we could use gradient descent. In general, even if we cannot be sure of the differentiability of $c_{\text{SCO}}$ we can always find a subgradient in the interior of $\mathcal{X}$~\cite{Bubeck2015}. We denote by $\partial c_{\text{SCO}}(x)$ the set of subgradients of $c_{SCO}$ at the point $x \in \mathcal{X}$.

We are interested in finding good approximations of $\hat{x}$. To that end, we call a solution $\hat{x}$ \emph{$\epsilon$-optimal}\index{$\epsilon$-optimal solution} if $\norm{g(\hat{x})}^2 \leq \epsilon$ where $g \in \partial c_{\text{SCO}}(\hat{x})$. Using the projected subgradient method we are able to find an $\epsilon$-optimal solution $\hat{x}$ in $\mathcal{O}(1 / \epsilon^2)$ iterations~\cite{Boyd2003}. Note that this convergence rate is dimension-independent.

In our implementation, we use two different algorithms for derivative-free local optimization. We use the Subplex method, which extends the Nelder-Mead method if there are no equality or inequality constraints beyond the bounds on the decision space~\cite{Rowan1990}. If we impose such constraints, we use the COBYLA (Constrained Optimization BY Linear Approximations) algorithm, which iteratively constructs linear approximations of the objective functions and approximates constraints using a simplex in $d+1$ dimensions~\cite{Powell1994, Powell1998}. We use implementations from the NLopt library~\cite{Johnson}.

To be able to solve the optimization problem efficiently, it is crucial to quickly find a point $x \in \mathcal{X}$ of which the associated cost is finite. Until such a point is found, the algorithms sample values from the decision space without direction. Our primary use of convex optimization solvers is to determine a minimizer of the hitting cost. Furthermore, in the application of dynamically right-sizing data centers, hitting costs are always finite for the upper bound of the decision space. We, therefore, use this upper bound as the first guess for finding the minimizer even if this guess may be farther away from the minimizer than the lower bound. A convex optimization that does not terminate thus indicates that the problem is infeasible, for example, because the provided load profile is infeasible.

Note that we cannot draw this conclusion if we solve for a different objective, such as \cref{eq:heterogeneous_load_balancing_convex_program} or \cref{eq:multiple_load_types_load_balancing_convex_program}. In these cases, there is no guarantee that the solver will find a feasible point unless it searches over the entire decision space. In the given example, however, choosing the uniform distribution across all active server types performs well in practice. This heuristic ensures that no server type without at least one active server is assigned load, which would incur an infinite cost by definition.

We have seen a general method of solving high-dimensional (fractional) SCO, though the convergence rate could be increased using an accelerated gradient method if the hitting costs are strongly convex or conditional gradient descent if the hitting costs are smooth~\cite{Bubeck2015}. Denoting the complexity of computing the hitting cost $f_t$ by $\mathcal{O}(C)$ and the convergence rate of a convex optimization with $\alpha$ dimensions by $\mathcal{O}(O_{\epsilon}^{\alpha})$ we obtain a total complexity of $\mathcal{O}(T C O_{\epsilon}^{T d})$. We observe, however, that our method does not extend to the integral case, which is of particular importance for the application of right-sizing data centers. We, therefore, devote much of the remaining chapter to the discussion of algorithms for integral variants of SCO.

\section{Uni-Dimensional}

We now limit our attention to the uni-dimensional setting, i.e. $\mathcal{X} \subset \mathbb{R}$.

\subsection{Backward-Recurrent Capacity Provisioning}\label{section:offline_algorithms:ud:capacity_provisioning}

Before discussing algorithms for Int-SSCO, we discuss a simple backward-recurrent algorithm for SSCO proposed by \citeauthor{Lin2011}~\cite{Lin2011} based on capacity provisioning. They extend this idea to formulate an online algorithm which we discuss in \cref{section:online_algorithms:ud:lazy_capacity_provisioning}.

\citeauthor{Lin2011}~\cite{Lin2011} observed that the optimal offline solution can be characterized by two bounds corresponding to charging the switching cost for powering-up and powering-down servers, respectively. Let $\tau \in [T]$ be a time slot. Then the optimal offline solution $\hat{X}_{\tau}$ during time slot $\tau$ is lower bounded by $X_{\tau,\tau}^L$ where $X_{\tau}^L$ is the smallest vector minimizing \begin{align}\label{eq:ud:brcp:lower}
    c_{\tau}^L(X) = \sum_{t=1}^{\tau} f_t(X_t) + \beta (X_t - X_{t-1})^+.
\end{align} Conversely, $\hat{X}_{\tau}$ is upper bounded by $X_{\tau,\tau}^U$ where $X_{\tau}^U$ is the largest vector minimizing \begin{align}\label{eq:ud:brcp:upper}
    c_{\tau}^U(X) = \sum_{t=1}^{\tau} f_t(X_t) + \beta (X_{t-1} - X_t)^+.
\end{align} Overall we have $X_{\tau,\tau}^L \leq \hat{X}_{\tau} \leq X_{\tau,\tau}^U$~\cite{Lin2011}. Note that the switching cost is paid for powering up a server for the lower bound, while for the upper bound, the switching cost is paid for powering down a server. Further, it is easy to see that the bounds for time slot $\tau$ do not depend on any time slots $t > \tau$.

An optimal offline algorithm can be described as determining the optimal schedule moving backward in time. We begin by setting $\hat{X}_{T+1} = 0$. For each previous time slot $\tau$, we set $\hat{X}_{\tau} = \hat{X}_{\tau + 1}$ unless this violates the bounds, in which case we make the smallest possible change: \begin{align*}
    \hat{X}_{\tau} = \begin{cases}
        0 & \tau > T \\
        (\hat{X}_{\tau+1})_{X_{\tau,\tau}^L}^{X_{\tau,\tau}^U} & \tau \leq T
    \end{cases}
\end{align*} where $(\hat{X}_{\tau+1})_{X_{\tau,\tau}^L}^{X_{\tau,\tau}^U}$ is the projection of $\hat{X}_{\tau+1}$ onto $[X_{\tau,\tau}^L, X_{\tau,\tau}^U]$~\cite{Lin2011}. The resulting algorithm is described in \cref{alg:brcp}. We can use algorithms for convex optimization to compute the upper and lower bounds. Hence, we obtain a complexity of $\mathcal{O}(T^2 C O_{\epsilon}^T)$ for $\epsilon$-optimal upper and lower bounds. An example of Backward-Recurrent Capacity Provisioning is given in \cref{fig:backward_recurrent_capacity_provisioing_vs_lazy_capacity_provisioning}.

\begin{algorithm}
    \caption{Backward-Recurrent Capacity Provisioning~\cite{Lin2011}}\label{alg:brcp}
    \SetKwInOut{Input}{Input}
    \Input{$\mathcal{I}_{\text{SSCO}} = (T \in \mathbb{N}, m \in \mathbb{N}, \beta \in \mathbb{R}_{>0}, (f_1, \dots, f_T) \in (\mathbb{R}_{\geq 0} \to \mathbb{R}_{\geq 0})^T)$}
    $\hat{X}_{T+1} \gets 0$\;
    \For{$\tau \gets T$ \KwTo $1$}{
        find $X_{\tau,\tau}^L$ using the optimization described by \cref{eq:ud:brcp:lower}\;
        find $X_{\tau,\tau}^U$ using the optimization described by \cref{eq:ud:brcp:upper}\;
        $\hat{X}_{\tau} \gets (\hat{X}_{\tau+1})_{X_{\tau,\tau}^L}^{X_{\tau,\tau}^U}$\;
    }
    \Return $\hat{X}$\;
\end{algorithm}

\subsection{Graph-Based Optimal Integral Algorithm}\label{section:offline_algorithms:ud:graph_based}

We now turn to the uni-dimensional integral case of simplified smoothed convex optimization, i.e., our decision space is given as $\mathcal{X} := [m]_0$ where $m \in \mathbb{N}$ is the maximum number of servers a data center can employ at the same time. As our decision space is discrete and finite, it is natural to model our problem as a weighted directed graph which is comprised of vertices $v_{t,j}$ for each $t \in [T]$ and $j \in [m]_0$ describing the state that during time slot $t$ exactly $j$ servers are active. For the initial and final state $X_0 = X_{T+1} = 0$ we have two additional vertices $v_{0,0}$ and $v_{T+1,0}$. For each vertice associated with time slot $t \in [T]_0$ we add an edge to all vertices associated with the subsequent time slot $t + 1$. With each edge from $v_{t,i}$ to $v_{t+1,j}$ we associate the switching cost incurred by the represented action, i.e. $\beta (j - i)^+$, and the hitting cost incurred by the state represented by the vertice $v_{t+1,j}$, i.e. $f_{t+1}(j)$. The structure of this graph is presented in \cref{fig:underlying_graph_of_uni_dimensional_integral_offline_algorithm}.

\begin{figure}
    \centering
    \resizebox{\textwidth}{!}{\input{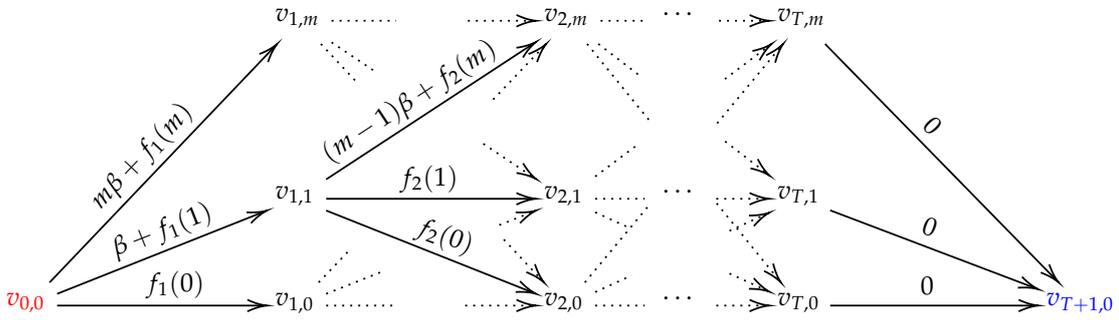}}
    \caption{Underlying graph of uni-dimensional integral offline algorithm. The algorithm finds a shortest path from $v_{0,0}$ (red) to $v_{T+1,0}$ (blue).}
    \label{fig:underlying_graph_of_uni_dimensional_integral_offline_algorithm}
\end{figure}

\citeauthor{Albers2018}~\cite{Albers2018} show that an optimal schedule is given by the shortest path from $v_{0,0}$ to $v_{T+1,0}$ within our constructed graph. We observe that our graph follows a particular structure. Each vertice representing an action during time slot $t$ can only be reached by a path that includes precisely one vertice representing time slots $0$ through $t - 1$ and, crucially, whether this path is optimal up to time $t$ does not depend on any time slot past $t$. Bellman's principle of optimality states that any part of an optimal path must itself be optimal~\cite{Bellman1954}. Based on this principle and using our previous observation, we can use dynamic programming to sequentially determine the optimal paths to the vertices of time slot $t$.

Moreover, we observe a second property of the graph of interest. Namely, for each time slot $t \in [T]$ (in the following called columns), the graph has precisely $m$ vertices (in the following called rows), which can only be reached through edges from vertices representing time slot $t - 1$ and all outgoing edges lead to vertices representing time slot $t + 1$. We are, therefore, able to iteratively solve the problem of finding a shortest path for subgraphs of our original graph using binary search. During each iteration, we only consider a constant number of rows.

To simplify the selection of rows, the algorithm assumes $m$ to be a power of two. Instances $\mathcal{I} = (T, m, \beta, F)$ where $m$ is not a power of two can be transformed into an instance $\mathcal{I}' = (T, m', \beta, F')$ with $m' = 2^{\lceil \log_2 m \rceil}$, $F' = (f'_1, \dots f'_T)$, and \begin{align*}
    f'_t(x) = \begin{cases}
        f_t(x) & x \leq m \\
        x (f_t(m) + \epsilon) & \text{otherwise} \\
    \end{cases}
\end{align*} for $\epsilon > 0$. The algorithm using $\log_2 m - 1$ iterations is described in \cref{alg:ud:optimal_graph_search}. During each iteration the algorithm finds the shortest path in a subgraph comprised of only five rows in $\mathcal{O}(T C)$ time. Overall, we thus find an optimal schedule in $\mathcal{O}(T C \log_2 m)$ time. \citeauthor{Albers2018}~\cite{Albers2018} show that in the final iteration, the algorithm obtains an optimal schedule for the original problem instance.

\begin{algorithm}
    \caption{Uni-Dimensional Optimal Graph Search~\cite{Albers2018}}\label{alg:ud:optimal_graph_search}
    \KwIn{$\mathcal{I}_{\text{Int-SSCO}} = (T \in \mathbb{N}, m \in \mathbb{N}, \beta \in \mathbb{R}_{>0}, (f_1, \dots, f_T) \in ([m]_0 \to \mathbb{R}_{\geq 0})^T)$ with $m$ a power of two}
    \eIf{$m > 2$}{$K \gets \log_2 m - 2$ }{$K \gets 0$ }
    $V^K \gets \{v_{0,0}, v_{T+1,0}\} \cup \{v_{t,\xi m / 4} \mid t \in [T], \xi \in [4]_0\}$\;
    $\hat{X}^K \gets \text{\ref{proc:ud:optimal_graph_search:shortest_path}}(\mathcal{I}_{Int-SSCO}, V^K)$\;
    \For{$k \gets K - 1$ \KwTo $0$}{
        $V_t^k \gets \{\hat{X}_t^{k+1} + \xi 2^k \mid \xi \in \{-2, -1, 0, 1, 2\}\} \cap [m]_0$\;
        $V^k \gets \{v_{0,0}, v_{T+1,0}\} \cup \{v_{t,j} \mid t \in [T], j \in V_t^k\}$\;
        $\hat{X}^k \gets \text{\ref{proc:ud:optimal_graph_search:shortest_path}}(\mathcal{I}_{Int-SSCO}, V^k)$\;
    }
    \Return $\hat{X}^0$\;
\end{algorithm}

\begin{function}
	\caption{ShortestPath($\mathcal{I}, V$)}\label{proc:ud:optimal_graph_search:shortest_path}
	\For{$t \gets 1$ \KwTo $T$}{
        \ForEach{$v_{t,j} \in V$}{
            $\hat{c}^{v_{t,j}} \gets \infty$\;
            \ForEach{$v_{t-1,i} \in V$}{
                $c^{v_{t,j}} \gets \hat{c}^{v_{t-1,i}} + f_t(j) + \beta (j - i)^+$\;
                \If{$c^{v_{t,j}} < \hat{c}^{v_{t,j}}$}{
                    $\hat{c}^{v_{t,j}} \gets c^{v_{t,j}}$\;
                    $\hat{X}^{v_{t,j}} \gets \hat{X}^{v_{t-1,i}}$\;
                }
            }
            $\hat{X}_t^{v_{t,j}} \gets j$\;
        }
    }
    $\hat{v} \gets \argmin_{v_{T,j} \in V} \hat{c}^{v_{T,j}}$\;
	\Return $\hat{X}^{\hat{v}}$\;
\end{function}

\section{Multi-Dimensional}

\subsection{Graph-Based Optimal Integral Algorithm}

We now lift the restriction on $d$ and also consider multi-dimensional instances of Int-SSCO. Again, an intuitive approach is to model the offline problem using a graph. Previously, with $d = 1$, the vertices where arranged in a two-dimensional grid (time being the first dimension). We now arrange the vertices in a $(d+1)$-dimensional grid. We call $x = (x_1, \dots, x_d) \in \mathcal{M}$ a \emph{configuration}\index{configuration} (also called a state in the literature) where $M_k := [m_k]_0$ and $\mathcal{M} := M_1 \times \dots \times M_d = \mathcal{X}$ is the set of all configurations. For each configuration $x$ and each time slot $t$ we introduce two vertices. $v_{t,x}^{\uparrow}$ represents the configuration in the beginning of time slot $t$ while $v_{t,x}^{\downarrow}$ represents the configuration at the end of time slot $t$. Thus, the first dimension has $2 T$ \emph{layers} where each layer only consists of powering-up or powering-down vertices.

In our graph we have edges $e_{t,x,k}^{\uparrow}$ representing the powering-up of a server of type $k$ in the beginning of time slot $t$, edges $e_{t,x}^{\text{op}}$ representing operating configuration $x$ during time slot $t$, edges $e_{t,x,k}^{\downarrow}$ representing the powering-down of a server of type $k$ at the end of time slot $t$, and edges $e_{t,x}^{\rightarrow}$ transitioning to the next time slot. For each $k \in [d]$ and $x = (x_1, \dots, x_d) \in [m_1]_0 \times \dots \times [m_k - 1]_0 \times \dots \times [m_d]_0$ let $x' = (x_1, \dots, x_k + 1, \dots, x_d)$. We add an edge $e_{t,x,k}^{\uparrow}$ between $v_{t,x}^{\uparrow}$ and $v_{t,x'}^{\uparrow}$ with weight $\beta_k$ and another edge $e_{t,x,k}^{\downarrow}$ between $v_{t,x'}^{\downarrow}$ and $v_{t,x}^{\downarrow}$ with weight $0$. For each time slot $t \in [T]$ and $x \in \mathcal{M}$, we add the edge $e_{t,x}^{\text{op}}$ from $v_{t,x}^{\uparrow}$ to $v_{t,x}^{\downarrow}$ with weight $f_t(x)$. Lastly, for each $t \in [T-1]$ and $x \in \mathcal{M}$, we add the edge $e_{t,x}^{\rightarrow}$ from $v_{t,x}^{\downarrow}$ to $v_{t+1,x}^{\uparrow}$ with weight $0$.  To simplify the algorithm we add an additional vertice $v_{T+1,\mathbf{0}}^{\uparrow}$ which can be reached through the edge $e_{t,\mathbf{0}}^{\rightarrow}$ from $v_{t,\mathbf{0}}^{\downarrow}$. The structure of this graph is presented in \cref{fig:underlying_graph_of_the_multi_dimensional_integral_offline_algorithm}.

\begin{figure}
    \centering
    \resizebox{\textwidth}{!}{\input{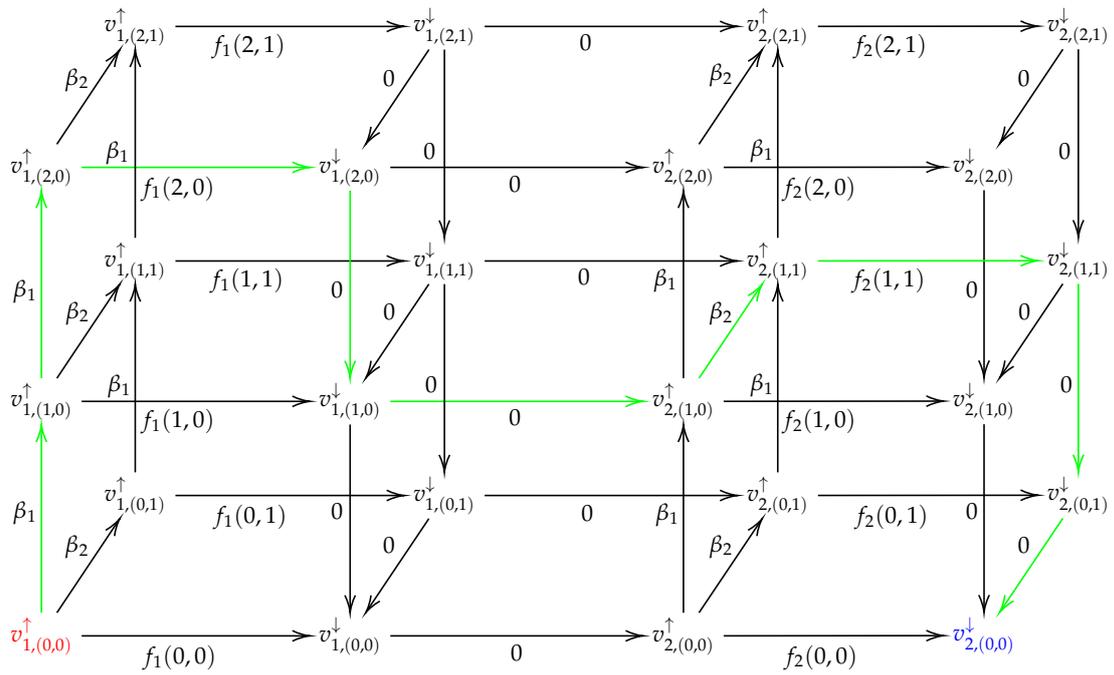}}
    \caption{Underlying graph of the multi-dimensional integral offline algorithm. In this example, $T = 2$, $d = 2$, $m_1 = 2$, and $m_2 = 1$. The algorithm finds a shortest path from $v_{1,\mathbf{0}}^{\uparrow}$ (red) to $v_{2,\mathbf{0}}^{\downarrow}$ (blue). A shortest path corresponding to the optimal schedule $X_1 = (2,0)$ and $X_2 = (1,1)$ is shown in green \cite{Albers2021_2}.}
    \label{fig:underlying_graph_of_the_multi_dimensional_integral_offline_algorithm}
\end{figure}

Any path from $v_{0,\mathbf{0}}^{\uparrow}$ to $v_{T,\mathbf{0}}^{\downarrow}$ must traverse exactly one edge $e_{t,x}^{\text{op}}$ for each time slot $t \in [T]$. The \emph{induced schedule} $X^P$ by a path $P$ assigns each time slot $t$ the configuration $x$ of the traversed edge $e_{t,x}^{\text{op}}$. \citeauthor{Albers2021_2}~\cite{Albers2021_2} show that a shortest path from $v_{0,\mathbf{0}}^{\uparrow}$ to $v_{T,\mathbf{0}}^{\downarrow}$ induces an optimal schedule. They also show that the cost of any induced schedule $X^P$ is given by the cost of the path $P$ when the sub-path between $v_{t,X_t^P}^{\downarrow}$ and $v_{t+1,X_{t+1}^P}^{\uparrow}$ is the shortest sub-path for all $t \in [T-1]$.

Again, we are able to use dynamic programming to obtain a shortest path from $v_{0,\mathbf{0}}^{\uparrow}$ to $v_{T,\mathbf{0}}^{\downarrow}$. Using Bellman's principle of optimality, we conclude that any such shortest path $P$ must consist of shortest sub-paths between the vertices of two subsequent layers. The algorithm works as follows: For each layer of each time slot $t$, we sequentially update the shortest path to each vertice of that layer. In layers consisting of powering-up vertices, we begin with the vertice $v_{t,\mathbf{0}}^{\uparrow}$ and then sequentially increase the values of each dimension beginning with dimension $1$. It is easy to see that by updating the vertices in this order, the predecessors of any newly reached vertice have already been updated. Conversely, in layers consisting of powering-down vertices, we begin with the vertice $v_{t,(m_1,\dots,m_d)}^{\downarrow}$ and iterate the dimensions from dimension $d$ through dimension $1$. The resulting algorithm is described in \cref{alg:md:optimal_graph_search}. We denote by $\hat{X}^v$ the optimal schedule up to vertice $v$ and by $\hat{c}^v$ the cost of the optimal schedule up to vertice $v$.

\begin{algorithm}
    \caption{Multi-Dimensional Optimal Graph Search~\cite{Albers2021_2}}\label{alg:md:optimal_graph_search}
    \KwIn{$\mathcal{I}_{\text{Int-SSCO}} = (d \in \mathbb{N}, T \in \mathbb{N}, m \in \mathbb{N}^d, \beta \in \mathbb{R}_{>0}^d, (f_1, \dots, f_T) \in (\mathcal{M} \to \mathbb{R}_{\geq 0})^T)$}
    \For{$t \gets 1$ \KwTo $T$}{
        $(\hat{X}, \hat{c}) \gets \text{\ref{proc:md:optimal_graph_search:handle_first_layer}}(\mathcal{I}_{Int-SSCO}, \mathcal{M}, \hat{X}, \hat{c}, t, 1, \{\mathbf{0}\})$\;
        $(\hat{X}, \hat{c}) \gets \text{\ref{proc:md:optimal_graph_search:handle_second_layer}}(\mathcal{I}_{Int-SSCO}, \mathcal{M}, \hat{X}, \hat{c}, t, d, \{(m_1,\dots,m_d)\})$\;
    }
    \Return $\hat{X}^{v_{T,\mathbf{0}}^{\downarrow}}$\;
\end{algorithm}

Here, the functions \ref{proc:md:optimal_graph_search:handle_first_layer} and \ref{proc:md:optimal_graph_search:handle_second_layer} update the vertices of the respective layer during time slot $t$. $k$ denotes the dimension that is expanded in the current iteration and $\mathcal{B}$ is the set of configurations from previous iterations the current expansion is based upon. $E \subset V \times \mathbb{R}$ is the set of all predecessors of the vertice $v_{t,x}^{\xi}$ along with the cost of the respective edge.

\begin{function}
	\caption{HandleFirstLayer($\mathcal{I}, \mathcal{M}, \hat{X}, \hat{c}, t, k, \mathcal{B}$)}\label{proc:md:optimal_graph_search:handle_first_layer}
	\lIf{$k > d$}{
	    \Return $(\hat{X}, \hat{c})$
	}
	$\mathcal{B}' \gets \mathcal{B}$\;
	\ForEach{$y \in \mathcal{B}$}{
	    \ForEach{$j \in M_k$}{
	        $x \gets y_{k \gets j}$\;
	        $E \gets \{(v_{t,x_{l \gets P_k(x_l)}}^{\uparrow}, \beta_l (x_l - P_k(x_l))) \mid l \in [k]_0, x_l > 0\}$\;
	        \If{$t > 1$}{$E \gets \{(v_{t-1,x}^{\downarrow}, 0)\} \cup E$\;}
	        $(\hat{X}, \hat{c}) \gets \text{\ref{proc:md:optimal_graph_search:update_paths}}(\hat{X}, \hat{c}, v_{t,x}^{\uparrow}, E)$\;
	        $\mathcal{B}' \gets \mathcal{B}' \cup \{x\}$\;
        }
    }
    \Return $\text{\ref{proc:md:optimal_graph_search:handle_first_layer}}(\mathcal{I}_{Int-SSCO}, \mathcal{M}, \hat{X}, \hat{c}, t, k+1, \mathcal{B}')$\;
\end{function}

\begin{function}
	\caption{HandleSecondLayer($\mathcal{I}, \mathcal{M}, \hat{X}, \hat{c}, t, k, \mathcal{B}$)}\label{proc:md:optimal_graph_search:handle_second_layer}
	\lIf{$k < 1$}{
	    \Return $(\hat{X}, \hat{c})$
	}
	$\mathcal{B}' \gets \mathcal{B}$\;
	\ForEach{$y \in \mathcal{B}$}{
	    \ForEach{$j \in M_k$}{
	        $x \gets y_{k \gets j}$\;
	        $E \gets \{(v_{t,x}^{\uparrow}, f_t(x))\} \cup \{(v_{t,x_{l \gets N_k(x_l)}}^{\downarrow}, 0) \mid l \in [k]_0, x_l < m_l\}$\;
	        $(\hat{X}, \hat{c}) \gets \text{\ref{proc:md:optimal_graph_search:update_paths}}(\hat{X}, \hat{c}, v_{t,x}^{\downarrow}, E)$\;
	        $\mathcal{B}' \gets \mathcal{B}' \cup \{x\}$\;
        }
    }
    \Return $\text{\ref{proc:md:optimal_graph_search:handle_second_layer}}(\mathcal{I}_{Int-SSCO}, \mathcal{M}, \hat{X}, \hat{c}, t, k-1, \mathcal{B}')$\;
\end{function}

$x_{k \gets j}$ denotes the update of configuration $x$ in dimension $k$ to the value $j$. We assume $M_k$ to be in ascending order for the first layer and in descending order for the second layer. $P_k(j)$ and $N_k(j)$ denote the previous and next values to $j$ in $M_k$, respectively. We keep the definitions abstract to allow for the generalization of this algorithm to an approximation algorithm. In the case of the optimal algorithm $P_k(j) = j-1$ and $N_k(j) = j+1$ for all $k \in [d]$. It is easy to verify that during the last iteration $\mathcal{B} = \mathcal{M}$.

\begin{function}
	\caption{UpdatePaths($\hat{X}, \hat{c}, v_{t,x}^{\xi}, E$)}\label{proc:md:optimal_graph_search:update_paths}
	$\hat{c}^{v_{t,x}^{\xi}} \gets \infty$\;
	\ForEach{$(v_{\tau,y}^{\kappa}, c^{v_{\tau,y}^{\kappa}}) \in E$}{
	    $c^{v_{t,x}^{\xi}} \gets \hat{c}^{v_{\tau,y}^{\kappa}} + c^{v_{\tau,y}^{\kappa}}$\;
	    \If{$c^{v_{t,x}^{\xi}} < \hat{c}^{v_{t,x}^{\xi}}$}{
	        $\hat{c}^{v_{t,x}^{\xi}} \gets c^{v_{t,x}^{\xi}}$\;
	        $\hat{X}^{v_{t,x}^{\xi}} \gets \hat{X}^{v_{\tau,y}^{\kappa}}$\;
	        \lIf{$\kappa = {\uparrow} \land \xi = {\downarrow}$}{$\hat{X}_t^{v_{t,x}^{\xi}} \gets x$}
	    }
	}
    \Return $(\hat{X}, \hat{c})$\;
\end{function}

\ref{proc:md:optimal_graph_search:update_paths} determines the shortest path to $v$ through its predecessors $E$ and updates the optimal schedule $\hat{X}^v$ (if necessary) and optimal cost $\hat{c}^v$. \Cref{fig:underlying_graph_of_the_multi_dimensional_integral_offline_algorithm} shows an example for a shortest path in the underlying graph.

\ref{proc:md:optimal_graph_search:update_paths} runs in $\mathcal{O}(|E|)$ time. Due to the structure of the graph, we can follow $|E| \in \mathcal{O}(d)$ and thus \ref{proc:md:optimal_graph_search:handle_first_layer} runs in $\mathcal{O}(|\mathcal{M}| d)$ time while \ref{proc:md:optimal_graph_search:handle_second_layer} runs in $\mathcal{O}(|\mathcal{M}| C d)$ time. Therefore, the overall time complexity of \cref{alg:md:optimal_graph_search} is in $\mathcal{O}(T |\mathcal{M}| C d)$ where $\mathcal{O}(T |\mathcal{M}|)$ is the size of the underlying graph. Note, that this running time is not polynomial in the size of the problem instance which is given by $\mathcal{O}(T + \sum_{k=1}^d \log_2 m_k)$ even when $d$ is assumed a constant as $|\mathcal{M}| \in \mathcal{O}(\prod_{k=1}^d m_k)$. For this reason, \citeauthor{Albers2021_2}~\cite{Albers2021_2} developed an approximation algorithm which we discuss next.

\subsection{Graph-Based Integral Polynomial-Time Approximation Scheme}

\citeauthor{Albers2021_2}~\cite{Albers2021_2} extend the graph-based algorithm from the previous section to a polynomial-time approximation scheme for SBLO. Their algorithm is a generalization of an approximation algorithm that was previously proposed by \citeauthor{Kappelmann2017}~\cite{Kappelmann2017} for the uni-dimensional setting. The idea is to restrict the possible values of $x_{t,k}$. In addition to requiring $x_{t,k} \in [m_k]_0$, we also require $x_{t,k}$ to be a power of some $\gamma > 1$. The possible numbers of active servers of type $k$ are now given as \begin{align*}
    M_k^{\gamma} := &\{0, m_k\} \cup \{\lfloor\gamma^i\rfloor \in [m_k]_0 \mid i \in \mathbb{N}\} \cup \{\lceil\gamma^i\rceil \in [m_k]_0 \mid i \in \mathbb{N}\} \\
                    &\{0, \lfloor\gamma^1\rfloor, \lceil\gamma^1\rceil, \lfloor\gamma^2\rfloor, \lceil\gamma^2\rceil, \dots, m_k\}.
\end{align*} It is easy to see that the ratio between two consecutive values in the ordered set $M_k^{\gamma}$ is not larger than $\gamma$. Also, $|M_k^{\gamma}| \in \mathcal{O}(\log_{\gamma} m_k)$. We define $\mathcal{M}^{\gamma} := M_1^{\gamma} \times \dots \times M_d^{\gamma}$ which results in $|\mathcal{M}^{\gamma}| \in \mathcal{O}(\prod_{k=1}^d \log_{\gamma} m_k)$. \Cref{fig:a_schedule_found_by_the_multi_dimensional_integral_offline_approximation_algorithm} shows an example for a schedule found by the approximation algorithm compared to an optimal schedule.

\begin{figure}
    \centering
    \input{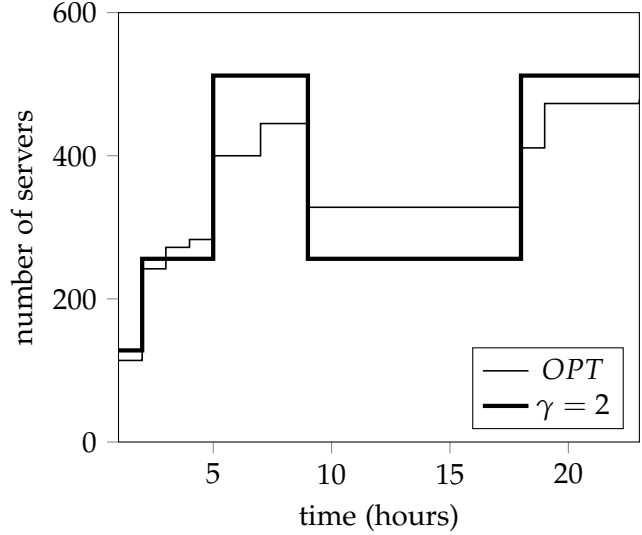}
    \caption{A comparison of the optimal schedule and a schedule found by the approximation algorithm with $\gamma = 2$. The example uses the Facebook 2009-0 trace under our second model. Note that the approximation algorithm is limited to server configurations that are powers of two.}
    \label{fig:a_schedule_found_by_the_multi_dimensional_integral_offline_approximation_algorithm}
\end{figure}

Let $G^{\gamma}$ be the described graph parametrized by $\gamma > 1$. \citeauthor{Albers2021_2}~\cite{Albers2021_2} show that, assuming we are given an instance of SBLO, given a shortest path $P^{\gamma}$ in $G^{\gamma}$, its induced schedule $X^{P^{\gamma}}$ is a $(2\gamma + 1)$-approximation of the optimal schedule $\hat{x}$. Further, the total number of vertices in $G^{\gamma}$ is \begin{align*}
    \mathcal{O}(T |\mathcal{M}|^{\gamma}) = \mathcal{O}(T \prod_{k=1}^d \log_{1+\epsilon} m_k)
\end{align*} as $|M_k^{\gamma}| \in \mathcal{O}(\log_{\gamma} m_k) = \mathcal{O}(\log_{1 + \epsilon} m_k)$. Using the same graph search algorithm that was described in the previous section we are able to find $X^{P^{\gamma}}$ in $\mathcal{O}(T C d \prod_{k=1}^d \log_{1 + \epsilon} m_k)$ time. Hence, setting $\gamma = 1 + \epsilon / 2$ for some $\epsilon > 0$ yields a $(1+\epsilon)$-approximation. The algorithm is described in \cref{alg:md:approximate_graph_search}.

\begin{algorithm}
    \caption{Multi-Dimensional Approximate Graph Search~\cite{Albers2021_2}}\label{alg:md:approximate_graph_search}
    \KwIn{$\mathcal{I}_{\text{Int-SSCO}} = (d \in \mathbb{N}, T \in \mathbb{N}, m \in \mathbb{N}^d, \beta \in \mathbb{R}_{>0}^d, (f_1, \dots, f_T) \in (\mathcal{M} \to \mathbb{R}_{\geq 0})^T)$}
    \For{$t \gets 1$ \KwTo $T$}{
        $(\hat{X}, \hat{c}) \gets \text{\ref{proc:md:optimal_graph_search:handle_first_layer}}(\mathcal{I}_{Int-SSCO}, \mathcal{M}^{\gamma}, \hat{X}, \hat{c}, t, 1, \{\mathbf{0}\})$\;
        $(\hat{X}, \hat{c}) \gets \text{\ref{proc:md:optimal_graph_search:handle_second_layer}}(\mathcal{I}_{Int-SSCO}, \mathcal{M}^{\gamma}, \hat{X}, \hat{c}, t, d, \{(m_1,\dots,m_d)\})$\;
    }
    \Return $\hat{X}^{v_{T,\mathbf{0}}^{\downarrow}}$\;
\end{algorithm}

One important question regarding the approximation algorithm remains. Namely, how to compute the powers of gamma $\gamma^i$ such that $\lfloor\gamma^i\rfloor \in [m_k]$ or $\lceil\gamma^i\rceil \in [m_k]$ for some $k \in [d]$. The simplest solution is to iteratively increase $i \in \mathbb{N}$ until $\gamma^i$ is greater than $\max_{k \in [d]} m_k$. We then keep track of all $\lfloor\gamma^i\rfloor$ and $\lceil\gamma^i\rceil$ that were generated in this way. We observe that this simple approach takes $\mathcal{O}(\log_{\gamma} m_k)$, not altering the runtime of the algorithm. We have thus discussed a polynomial-time approximation scheme for multi-dimensional instances of Int-SSCO.

\begin{figure}
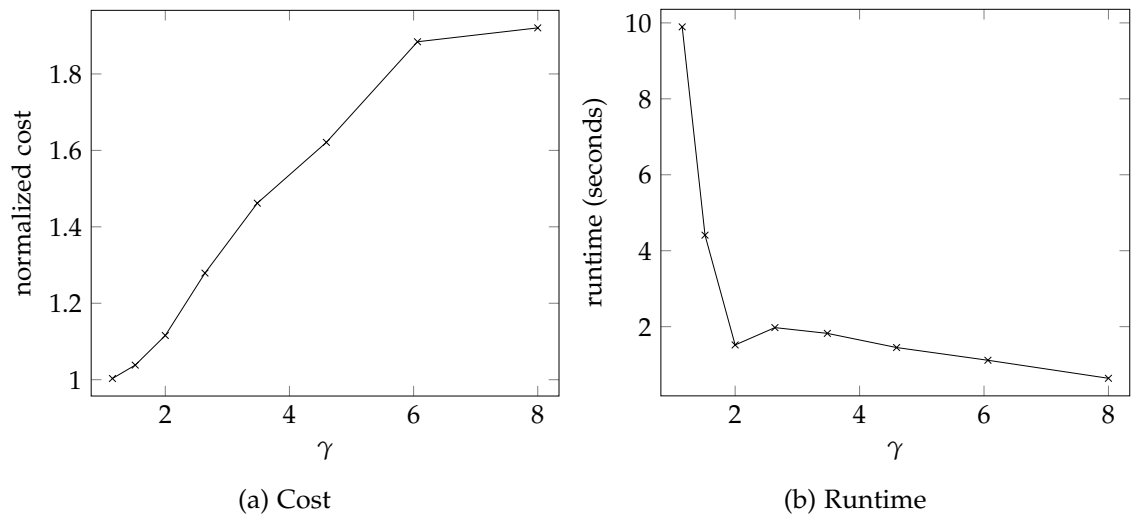

    \begin{subfigure}[b]{.5\linewidth}
    \resizebox{\textwidth}{!}{\input{thesis/figures/approx_graph_search_cost}}
    \caption{Cost}\label{fig:approx_graph_search:cost}
    \end{subfigure}
    \begin{subfigure}[b]{.5\linewidth}
    \resizebox{\textwidth}{!}{\input{thesis/figures/approx_graph_search_runtime}}
    \caption{Runtime}\label{fig:approx_graph_search:runtime}
    \end{subfigure}
    \caption{Normalized cost and runtime of the approximate graph search algorithm for the Facebook 2009-1 trace.}
\end{figure}

In practice, we observe that cost grows linearly with $\gamma$. The normalized cost when used with the Facebook 2009-1 trace is shown in \cref{fig:approx_graph_search:cost}. In contrast, the most significant decrease in runtime occurs for $\gamma \in (1, 2]$ as shown in \cref{fig:approx_graph_search:runtime}.


\chapter{Online Algorithms}\label{chapter:online_algorithms}

In this chapter, we discuss the online algorithms we implemented in our work. Similar to the previous chapter on offline algorithms, we begin our discussion in \cref{section:online_algorithms:ud} with algorithms for the uni-dimensional setting. As was discussed in \cref{chapter:theory}, these algorithms give strong guarantees yielding a constant competitive ratio. Next in \cref{section:online_algorithms:md}, we extend our discussion to the multi-dimensional setting. Here, the guarantees are not as strong. We thus begin in \cref{section:online_algorithms:md:lazy_budgeting} by considering lazy budgeting algorithms for smoothed convex optimization problems with particular cost functions. As mentioned previously in \cref{chapter:theory}, while there are algorithms with sublinear regret (gradient descent), there cannot be any algorithm achieving a dimension-independent constant competitive ratio unless the class of allowed cost functions is restricted~\cite{Chen2018}. In \cref{section:online_algorithms:md:descent_methods}, we thus discuss gradient methods that perform well with regard to either the competitive ratio or regret with a restricted class of cost functions. Still, sublinear regret and a constant competitive ratio cannot be achieved simultaneously, even for linear cost functions~\cite{Andrew2015}. We, therefore, end this chapter in \cref{section:online_algorithms:md:predictions} with a discussion of algorithms that use predictions to circumvent this fundamental limitation. \Cref{appendix:taxonomy} includes an overview of all discussed online algorithms.

Throughout this chapter, we denote by $\tau \in [T]$ the current time slot. In contrast to offline algorithms that know the hitting costs $f_t$ for all $t \in [T]$, an online algorithm only knows the hitting costs $f_t$ up to $\tau$, i.e. $t \in [\tau]$.

\section{Uni-Dimensional}\label{section:online_algorithms:ud}

\subsection{Lazy Capacity Provisioning}\label{section:online_algorithms:ud:lazy_capacity_provisioning}

\subsubsection{Fractional Algorithm}

We begin by returning to the notion of capacity provisioning that we introduced in \cref{section:offline_algorithms:ud:capacity_provisioning}, yielding a backward-recurrent algorithm finding an optimal schedule for SSCO. This algorithm computed bounds $X_{\tau}^L$ and $X_{\tau}^U$ on the optimal solution, which only depend on the schedule up to time slot $\tau$. However, the optimal offline algorithm stayed within these bounds moving backward in time which is impossible for an online algorithm. \citeauthor{Lin2011}~\cite{Lin2011} present a similar algorithm moving forward in time called \emph{lazy capacity provisioning}. We compute the schedule $X_{\tau}$ during time slot $\tau$ by setting $X_{\tau} = X_{\tau-1}$ unless this violates the bounds in which case we make the smallest possible change: \begin{align*}
    X_{\tau} = \begin{cases} 
        0 & \tau \leq 0 \\
        (X_{\tau-1})_{X_{\tau,\tau}^L}^{X_{\tau,\tau}^U} & \tau \geq 1
    \end{cases}
\end{align*} where $(X_{\tau-1})_{X_{\tau,\tau}^L}^{X_{\tau,\tau}^U}$ is the projection of $X_{\tau-1}$ onto $[X_{\tau,\tau}^L, X_{\tau,\tau}^U]$~\cite{Lin2011}. \Cref{fig:backward_recurrent_capacity_provisioing_vs_lazy_capacity_provisioning} shows how this update rule differs from \nameref{alg:brcp}. The resulting algorithm is described in \cref{alg:ud:lcp}. Similar to the offline algorithm, we can use algorithms for convex optimization to compute the upper and lower bounds. Hence, we obtain a complexity of $\mathcal{O}(\tau C O_{\epsilon}^{\tau})$ for $\epsilon$-optimal upper and lower bounds. This complexity is worrying as it depends on $\tau$, which may grow very large. However, \citeauthor{Lin2011}~\cite{Lin2011} prove the following lemma, which implies that it suffices to compute the lower and upper bounds using only the history since the last time slot where both bounds were either decreased or increased.

\begin{figure}
    \centering
    \input{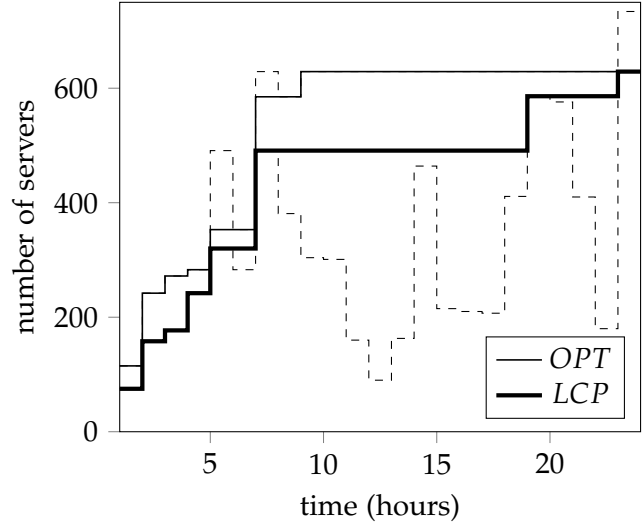}
    \caption{Backward-Recurrent Capacity Provisioning vs. Lazy Capacity Provisioning. LCP stays within the bounds ``lazily'' moving forwards in time. The optimal solution stays within the bounds moving backwards in time.}
    \label{fig:backward_recurrent_capacity_provisioing_vs_lazy_capacity_provisioning}
\end{figure}

\begin{lemma}
\cite{Lin2011} If there exists an index $t \in [1, \tau-1]$ such that $X_{\tau,t+1}^U < X_{\tau,t}^U$ or $X_{\tau,t+1}^L > X_{\tau,t}^L$, then $(\hat{X}_{\tau,1},\dots,\hat{X}_{\tau,t}) := (X_{\tau,1}^L,\dots,X_{\tau,t}^L) = (X_{\tau,1}^U,\dots,X_{\tau,t}^U)$, and no matter what the future arrival is, solving the optimization in $[1,\tau']$ for $\tau' > \tau$ is equivalent to solving two optimizations: one over $[1,t]$ with initial condition $X_0$ and final condition $\hat{X}_{\tau,t}$ and the second over $[t+1,\tau']$ with initial condition $\hat{X}_{\tau,t}$.
\end{lemma}

While not changing the worst-case complexity, this significantly improves the practical complexity in the application of right-sizing data centers as diurnal load patterns typically ensure that less than a day needs to be considered~\cite{Lin2011}. We denote by $X_{\tau}^{L,(t,x_0)}$ and $X_{\tau}^{U,(t,x_0)}$ the bounds resulting from optimizations beginning at time slot $t$ with initial condition $x_0$. \citeauthor{Lin2011}~\cite{Lin2011} showed that lazy capacity provisioning is $3$-competitive and also proved that this result is tight.

\begin{algorithm}
    \caption{Lazy Capacity Provisioning~\cite{Lin2011}}\label{alg:ud:lcp}
    \SetKwInOut{Input}{Input}
    \Input{$\mathcal{I}_{\text{SSCO}} = (\tau \in \mathbb{N}, m \in \mathbb{N}, \beta \in \mathbb{R}_{>0}, (f_1, \dots, f_{\tau}) \in (\mathbb{R}_{\geq 0} \to \mathbb{R}_{\geq 0})^{\tau})$}
    $t_0 \gets 0$\;
    $x_0 \gets 0$\;
    \For{$t \gets \tau-1$ \KwTo $2$}{
        \If{$X_{t,t}^U < X_{t,t-1}^U \lor X_{t,t}^L > X_{t,t-1}^L$}{
            $t_0 \gets t$\;
            $x_0 \gets X_{t,t-1}^U$\;
            \KwBreak
        }
    }
    find $X_{\tau,\tau}^{L,(t_0,x_0)}$ using the optimization described by \cref{eq:ud:brcp:lower}\;
    find $X_{\tau,\tau}^{U,(t_0,x_0)}$ using the optimization described by \cref{eq:ud:brcp:upper}\;
    \Return $(X_{\tau-1})_{X_{\tau,\tau}^{L,(t_0,x_0)}}^{X_{\tau,\tau}^{U,(t_0,x_0)}}$\;
\end{algorithm}

\subsubsection{Integral Algorithm}

\citeauthor{Albers2018}~\cite{Albers2018} applied lazy capacity provisioning to the integral variant Int-SSCO using their graph-based offline algorithm discussed in \cref{section:offline_algorithms:ud:graph_based} to compute the integral lower and upper bounds. It is apparent that this immediately yields a deterministic online algorithm for Int-SSCO. \citeauthor{Albers2018}~\cite{Albers2018} showed that similar to lazy capacity provisioning, their algorithm is $3$-competitive. Due to the changed method of determining the bounds, its runtime is $\mathcal{O}(\tau^2 C \log_2 m)$. Note that it is impossible to cache the intermediate results of the dynamic program (see \cref{alg:ud:optimal_graph_search}) as the binary search over possible configurations considers different vertices depending on the obtained schedule, which changes over time. Thus, for large $\tau$, it may be beneficial to use caching instead of binary search resulting in a worst-case runtime of $\mathcal{O}(\tau C m)$. By using the same method of shortening the used history that was proposed by \citeauthor{Lin2011}~\cite{Lin2011}, we can reduce this time complexity drastically in practice (for large $\tau$). Thus, the adopted algorithm is still described by \cref{alg:ud:lcp}. We simply need to slightly modify the graph-based algorithm computing optimal offline solutions to allow for initial conditions other than $0$.

\subsection{Memoryless Algorithm}\label{section:online_algorithms:ud:memoryless}

\citeauthor{Bansal2015}~\cite{Bansal2015} showed that for SSCO, a competitive ratio of $3$ can also be attained by a memoryless algorithm. In a memoryless online algorithm for smoothed convex optimization, the configuration $X_{\tau}$ only depends on the preceding configuration $X_{\tau-1}$ and the current hitting cost $f_{\tau}$. This generally allows for a more space and time-efficient algorithm, which is important when choosing a small time slot length $\delta$ to be more responsive to changes in load.

The algorithm proposed by \citeauthor{Bansal2015}~\cite{Bansal2015} works as follows. Let $\hat{x}$ be the minimizer of $f_{\tau}(x)$, i.e. $\hat{x} = \argmin_{x \in \mathcal{X}} f_{\tau}(x)$. The algorithm moves into the direction of the minimizer until it either reaches $\hat{x}$, or it reaches a configuration $x$ where its switching cost equals twice the hitting cost of $x$. \Cref{fig:memoryless_algorithm} gives an example of a step of the algorithm. We observe that this is equivalent to the following convex optimization: \begin{align}\label{eq:ud:memoryless}\begin{aligned}
    &\min_{x \in \mathcal{X}} &&f_{\tau}(x) \\
    &\text{subject to}        &&\beta |x - X_{\tau-1}| \leq \frac{f_{\tau}(x)}{2}.
\end{aligned}\end{align} Originally, \citeauthor{Bansal2015}~\cite{Bansal2015} proposed this algorithm for a restricted variant of uni-dimensional SSCO where the decision space is unbounded, i.e., $\mathcal{X} = \mathbb{R}$, and the switching costs are given by the $\ell_1$ norm. In particular, they choose $\beta = 1$. First, it is easy to see that we can adapt the algorithm for a bounded decision space by bounding the feasible region of the optimization problem in \cref{eq:ud:memoryless}. Second, we observe that $\beta$ can simply be interpreted as the weight that we associate with smoothing (i.e., minimizing movement) instead of minimizing hitting costs. This is shown by the following equation, which is obtained by dividing the cost of \cref{eq:simplified_smoothed_convex_optimization} by $\beta$: \begin{align}\label{eq:simplified_smoothed_convex_optimization:without_beta}
    \sum_{t=1}^T \frac{1}{\beta} f_t(X_t) + \sum_{k=1}^d (X_{t,k} - X_{t-1,k})^+.
\end{align} The cost associated with this equation is the cost of \cref{eq:simplified_smoothed_convex_optimization} linearly scaled by $1 / \beta$. Especially, this argument shows the following lemma.

\begin{lemma}\label{lemma:switching_cost_absolute_vs_positive_movement}
A schedule is optimal with respect to \cref{eq:simplified_smoothed_convex_optimization:without_beta} if and only if it is optimal with respect to \cref{eq:simplified_smoothed_convex_optimization}.
\end{lemma}

Therefore, without loss of optimality, we can incorporate the weight of the switching cost $\beta$ into the hitting costs. Further, note that in their model, \citeauthor{Bansal2015}~\cite{Bansal2015} consider the absolute movement, i.e. $|X_{t,k} - X_{t-1,k}|$, rather than only positive movements, i.e. $(X_{t,k} - X_{t-1,k})^+$. However, with \cref{lemma:switching_cost_l1_norm_vs_pos_movement} in \cref{chapter:theory}, we have shown that these switching costs only differ by a constant factor (namely $1/2$).

The resulting algorithm is simply given by determining $\hat{x}$ based on the convex optimization in \cref{eq:ud:memoryless}, see \cref{alg:ud:memoryless}. Thus, the time (and space) complexity of this memoryless algorithm is $\mathcal{O}(C O_{\epsilon}^1)$ for finding $\epsilon$-optimal solutions.

\begin{figure}
    \centering
    \input{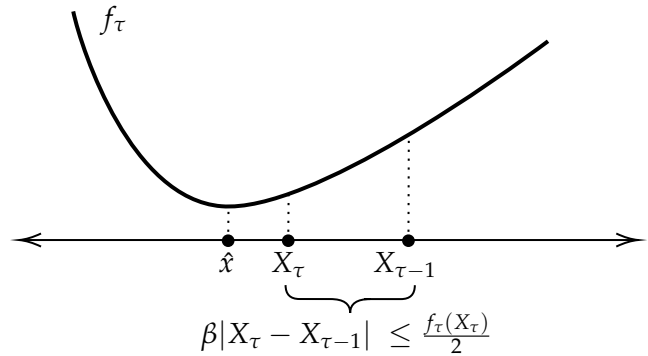}
    \caption{The memoryless algorithm moves towards the minimizer of the hitting cost, balancing hitting and movement costs.}
    \label{fig:memoryless_algorithm}
\end{figure}

\begin{algorithm}
    \caption{Memoryless algorithm~\cite{Bansal2015}}\label{alg:ud:memoryless}
    \SetKwInOut{Input}{Input}
    \Input{$\mathcal{I}_{\text{SSCO}} = (\tau \in \mathbb{N}, m \in \mathbb{N}, \beta \in \mathbb{R}_{>0}, (f_1, \dots, f_{\tau}) \in (\mathbb{R}_{\geq 0} \to \mathbb{R}_{\geq 0})^{\tau})$}
    \Return $\hat{x}$ such that that $\hat{x}$ is the result of the optimization in \cref{eq:ud:memoryless}\;
\end{algorithm}

\subsection{Probabilistic Algorithm}\label{section:online_algorithms:ud:probabilistic}

Next, we discuss a $2$-competitive algorithm developed by \citeauthor{Bansal2015}~\cite{Bansal2015}, which works by maintaining a probability distribution over configurations. Using this probability distribution, they describe a randomized algorithm which they subsequently translate into a deterministic algorithm. We discuss how to gather a randomized and then a deterministic algorithm from a probability distribution. Then, we describe how \citeauthor{Bansal2015}~\cite{Bansal2015} determine the probability distribution and how it can be computed in practice.

\subsubsection{From Probability Distribution to Deterministic Algorithm}

Let's suppose we have given a probability distribution $p$ over configurations $x \in \mathcal{X}$. A randomized algorithm is then described by initially picking a number $\gamma \in [0,1]$ uniformly at random and then maintaining the invariant that at time $\tau$ the chosen configuration $x_{\tau}$ has the property that the probability mass to the left of $x_{\tau}$ with respect to $p$ is exactly $\gamma$~\cite{Bansal2015}. Crucially, this approach only works in the fractional setting. Also, note that $\gamma$ is chosen only once prior to running the algorithm. This describes how we obtain a randomized algorithm from a probability distribution over configurations.

Next, \citeauthor{Bansal2015}~\cite{Bansal2015} show the following theorem, which describes how we can obtain a deterministic algorithm from a randomized algorithm. \begin{theorem}
   ~\cite{Bansal2015} For the problem of (fractional) online convex optimization, if there exists a $\rho$-competitive randomized algorithm $\mathcal{R}$ then there exists a $\rho$-competitive deterministic algorithm $\mathcal{D}$.
\end{theorem}
\begin{proof}
\citeauthor{Bansal2015}~\cite{Bansal2015} prove this theorem using Jensen's inequality. In the setting of a probability space, \emph{Jensen's inequality}\index{Jensen's inequality} claims that given a convex function $\varphi$ and a random variable $X$ we have \begin{align}
    E(\varphi(X)) \geq \varphi(E X)
\end{align} provided both expectations exist, i.e. $E |X|$ and $E |\varphi(X)| < \infty$~\cite{Durrett2010}.

Let $X_{\tau}$ be a random variable denoting the configuration of the randomized algorithm $\mathcal{R}$ at time $\tau$. Then, the deterministic algorithm $\mathcal{D}$ of \citeauthor{Bansal2015}~\cite{Bansal2015} sets their configuration to $x_{\tau} = E X_{\tau}$. The cost of $\mathcal{D}$ is thus given by $f_{\tau}(x_{\tau}) + (x_{\tau} - x_{\tau-1})^+$ and the cost of $\mathcal{R}$ is given by $E(f_{\tau}(X_{\tau})) + E((X_{\tau} - X_{\tau-1})^+)$. We observe that both $f_{\tau}$ and $(\cdot)^+$ are convex functions, implying that the cost of $\mathcal{R}$ is at least $f_{\tau}(E X_{\tau}) + (E X_{\tau} - E X_{\tau-1})^+$ which equals the cost of $\mathcal{D}$. Summing over all $t$ completes the proof.
\end{proof}

Hence, we have seen that a deterministic algorithm can be obtained from a randomized algorithm by, in each time slot, choosing the expected configuration of the randomized algorithm.

\subsubsection{Assumptions}

In the description of their algorithm, \citeauthor{Bansal2015}~\cite{Bansal2015} consider a restricted variant of uni-dimensional SSCO. Similar to their memoryless algorithm, which we discussed in \cref{section:online_algorithms:ud:memoryless}, they consider an unbounded decision space, i.e., $\mathcal{X} = \mathbb{R}$, and the $\ell_1$ norm as switching costs. Further, for their description of this probabilistic algorithm, they assume that the minimizer $\hat{x}$ of $f_{\tau}$ is unique and bounded and that the hitting costs $f_{\tau}$ are continuous and smooth, i.e., are infinitely many times continuously differentiable. In particular, they assume the first-order and second-order derivatives of $f_{\tau}$ are well-defined and continuous. In \cref{section:theory:beyond_convexity}, we discussed the assumption of differentiability and how it relates to our data center model.

Our implementation generalizes their algorithm to instances of SSCO with a bounded decision space $\mathcal{X}$, variable switching costs $\beta$, and piecewise linear functions. The second assumption, namely that the minimizer of the hitting cost is bounded, is natural in a data center setting as revenue loss increases for small configurations, whereas energy costs increase for large configurations.

In summary, the final algorithm is $2$-competitive for arbitrary instances of uni-dimensional SSCO with the restriction that hitting costs must either be piecewise linear or smooth.

\subsubsection{The Probability Distribution}

For any time $\tau$, the algorithm maintains a probability distribution $p_{\tau}$ over configurations $x \in \mathcal{X}$. So $\int_a^b p_{\tau}(x) \,dx$ represents the probability that $X_{\tau} \in [a,b]$ for any two $a, b \in \mathcal{X}$. At each time step $\tau$ we first find the minimizer of $f_{\tau}$, $\hat{x} = \argmin_{x \in \mathcal{X}} f_{\tau}(x)$. Then, we find a point $x_r \geq \hat{x}$ such that \begin{align}\label{eq:ud:probabilistic:right}
    \frac{1}{2} \int_{\hat{x}}^{x_r} \diff[2]{f_{\tau}}{y}(y) \,dy = \beta \int_{x_r}^{\infty} p_{\tau-1}(y) \,dy
\end{align} and a point $x_l \leq \hat{x}$ such that \begin{align}\label{eq:ud:probabilistic:left}
    \frac{1}{2} \int_{x_l}^{\hat{x}} \diff[2]{f_{\tau}}{y}(y) \,dy = \beta \int_{-\infty}^{x_l} p_{\tau-1}(y) \,dy.
\end{align} Note that we use \cref{lemma:switching_cost_absolute_vs_positive_movement} to linearly scale the hitting cost $f_{\tau}$ by $1 / \beta$ to allow for $\beta \neq 1$. We then simply moved the constant factor outside of the derivative and integral. The probability distribution is updated as follows: \begin{align}\label{eq:ud:probabilistic:update}
    p_{\tau}(x) = \begin{cases}
        p_{\tau-1}(x) + \frac{1}{2 \beta} \diff[2]{f_{\tau}}{x}(x) & x \in [x_l,x_r] \\
        0 & \text{otherwise}
    \end{cases}
\end{align} where $p_0$ is a discrete distribution concentrating all probability mass in the point $0$. Note that \citeauthor{Bansal2015}~\cite{Bansal2015} do not assume any particular initial distribution, yet in our original problem statement we assumed $X_0 = \mathbf{0}$. The continuous extension of this distribution can be approximated as $p_0 \sim \text{Unif}(0, \epsilon)$ for a suitably small $\epsilon > 0$. In our implementation we choose $\epsilon = 10^{-5}$.

\subsubsection{The Algorithm}

To begin with, recall that the algorithm developed by \citeauthor{Bansal2015}~\cite{Bansal2015} operates on an unbounded decision space. To translate the algorithm to a setting with a bounded decision space, it is easy to see that we need to ensure that the underlying probability distribution does not assign positive probability to $x \not\in \mathcal{X}$. This can be achieved by introducing the additional restrictions $0 \leq x_l$ and $x_r \leq m$ which requires the assumption $\hat{x} \in [0,m]$. This is not a restriction as we simply define $\hat{x}$ as the minimizer of the hitting cost $f_{\tau}$ with respect to the decision space $\mathcal{X}$.

We use a convex optimization (as described in \cref{section:offline_algorithms:convex_optimization}) to find the minimizer of the hitting cost $\hat{x}$. To determine $x_l$ and $x_r$, we use Brent's method with suitably defined functions and intervals to find a root. We will define these intervals and functions and describe how they can be computed in the following. Before beginning their description, note that we can only use a bracketed root finding method as we assumed that our decision space is bounded. If the decision space was not bounded, $x_l$ could be determined using a search for a local optimum minimizing $x$ starting from $\hat{x}$ under the equality constraint given in \cref{eq:ud:probabilistic:left}. This works because the equality constraint reduces the dimensionality of the optimization to $0$, resulting in a single feasible point. The analogous approach can be used to determine $x_r$.

We begin by describing how $x_r$ can be determined by a local search for a root. First, note that $x_r \in [\hat{x},m]$. Next, we restate \cref{eq:ud:probabilistic:right} as a function of $x_r$: \begin{align*}
    && \frac{1}{2} \int_{\hat{x}}^{x_r} \diff[2]{f_{\tau}}{y}(y) \,dy =&\ \beta \int_{x_r}^{\infty} p_{\tau-1}(y) \,dy \\
    \iff&& \left(\diff{f_{\tau}}{x}(x_r) - \diff{f_{\tau}}{x}(\hat{x})\right) =&\ 2 \beta \int_{x_r}^{\infty} p_{\tau-1}(y) \,dy \\
    \iff&& g(x_r) := \diff{f_{\tau}}{x}(x_r) - 2 \beta \int_{x_r}^{\infty} p_{\tau-1}(y) \,dy =&\ 0.
\end{align*} Here, we used the fundamental theorem of calculus and that the first order derivative of $f_{\tau}$ at $\hat{x}$ is $0$ as $\hat{x}$ is the minimizer of $f_{\tau}$. We observe that $g(\hat{x}) \leq 0$ and $g(m) \geq 0$. As $g$ is continuous, we know that we can be sure to find a root $x_r$ of $g$ on the interval $[\hat{x},m]$.

We take the analogous approach to determine $x_l \in [0,\hat{x}]$. Using \cref{eq:ud:probabilistic:left}, we obtain: \begin{align*}
    && \frac{1}{2} \int_{x_l}^{\hat{x}} \diff[2]{f_{\tau}}{y}(y) \,dy =&\ \beta \int_{-\infty}^{x_l} p_{\tau-1}(y) \,dy \\
    \iff&& \left(\diff{f_{\tau}}{x}(\hat{x}) - \diff{f_{\tau}}{x}(x_l)\right) =&\ 2 \beta \int_{-\infty}^{x_l} p_{\tau-1}(y) \,dy \\
    \iff&& h(x_l) := 2 \beta \int_{-\infty}^{x_l} p_{\tau-1}(y) \,dy - \diff{f_{\tau}}{x}(x_l) =&\ 0.
\end{align*} Again, we observe that $h(0) \leq 0$, $h(\hat{x}) \geq 0$, and $h$ is continuous, implying that we can be sure to find a root $x_l$ of $h$ on the interval $[0,\hat{x}]$. \Cref{fig:probabilistic_algorithm} illustrates the choice of $x_l$ and $x_r$.

\begin{figure}
    \centering
    \input{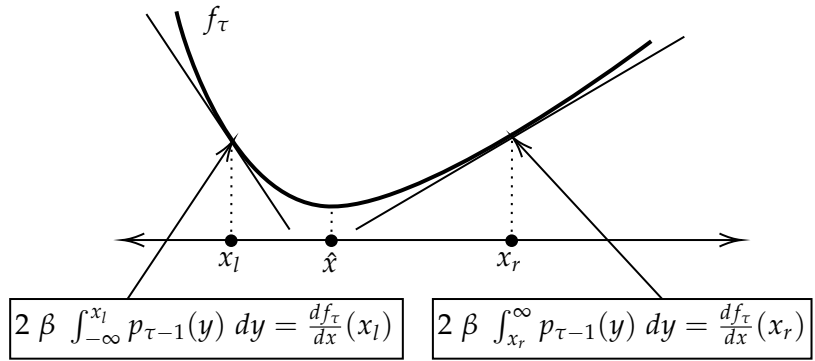}
    \caption{Visualization of the choice of $x_l$ and $x_r$ of the probabilistic algorithm.}
    \label{fig:probabilistic_algorithm}
\end{figure}

\paragraph{Root Finding} The previous arguments show that a bracketed root finding method can be used to find $x_l$ and $x_r$. We use \emph{Brent's method}\index{Brent's method} for root finding. Brent's method combines the bisection method with higher-order methods to guarantee convergence to the root, yet at a higher rate than if only bisection were used~\cite{Press2007}. \citeauthor{Press2007}~\cite{Press2007} ``recommend it as the method of choice for general one-dimensional
root finding where a function’s values only (and not its derivative or functional form)
are available''. We denote the convergence rate of approximating the root with tolerance $\epsilon$ by $\mathcal{O}(R_{\epsilon})$.

\paragraph{Numerical Differentiation} We use the \emph{five-point stencil}\index{five-point stencil} \begin{align*}
    \diff{f_{\tau}}{x}(x) \approx \frac{-f_{\tau}(x - 2h) + 8 f_{\tau}(x + h) - 8 f_{\tau}(x - h) + f_{\tau}(x - 2h)}{12h}
\end{align*} to find a finite difference approximation of order $\mathcal{O}(h)$ of the first order derivative of $f_{\tau}$ at configurations $x \in \mathcal{X}$~\cite{Sauer2011}. To match the accuracy of our convex optimizations we set $h := \epsilon / 10$. To approximate the second order derivative of $f_{\tau}$ at a configuration $x \in \mathcal{X}$ we use \begin{align*}
    \diff[2]{f_{\tau}}{x}(x) \approx \frac{-f_{\tau}(x + 2h) + 16 f_{\tau}(x+h) - 30 f_{\tau}(x) + 16 f_{\tau}(x-h) - f_{\tau}(x - 2h)}{12 h^2}
\end{align*} which yields an approximation of order $\mathcal{O}(h^4)$~\cite{Sauer2011}. Thus, we set $h := (\epsilon / 10)^{-1/4}$. We are thus able to compute these approximations in $\mathcal{O}(C)$ time.

Next, we describe how the constraints from \cref{eq:ud:probabilistic:right} and \cref{eq:ud:probabilistic:left} can be computed numerically.

\paragraph{Numerical Integration} In our implementation, we need to compute both finite and semi-infinite integrals over the probability distribution $p$.

We use the \emph{Tanh-sinh quadrature}\index{Tanh-sinh quadrature} (also known as the double exponential method) to compute finite integrals. \citeauthor{Bailey2005}~\cite{Bailey2005} describe the convergence and error of this method in more detail. They conclude that ``overall, the tanh-sinh scheme appears to be the best
for integrands of the type most often encountered in experimental math research'' and highlight that it has ``excellent accuracy and runtime performance''~\cite{Bailey2005}.

We use the \emph{Gauss-Laguerre quadrature}\index{Gauss-Laguerre quadrature} for semi-infinite integrals, which approximates values of integrals of the kind \begin{align}\label{eq:gauss_laguerre}
    \int_0^{\infty} e^{-x} f(x) \,dx
\end{align}~\cite{Weisstein}. It is easy to see that we can eliminate the weights by multiplying the integrand $g$ with $e^x$. Let $\text{GL}(g)$ denote the approximation of \cref{eq:gauss_laguerre} obtained by the Gauss-Laguerre quadrature. We are then able to compute any right-open integral over the interval $[{a,\infty})$ with integrand $p$ by setting $g(x) := p(a+x)$ and any left-open integral over the interval $({-\infty,b}]$ with integrand $p$ by setting $g(x) := p(b-x)$.

Crucially, for numeric stability in both integration schemes, we need that the integrands are continuous. It is easy to see that, in general, this is not the case for our probability distribution $p$. We describe in the following paragraph how integrals can be suitably discretized to allow for stable numeric results. Moreover, as the integration schemes are not universal, probability distributions exist for which the used quadratures cannot find the integral. In such a case, one would have to resort to another integration scheme. We denote the convergence rate of approximating the integral with tolerance $\epsilon$ by $\mathcal{O}(I_{\epsilon})$.

\paragraph{Piecewise Linear Hitting Costs} We have already discharged the assumptions that $\beta = 1$ and that the decision space is unbounded. It remains to extend this algorithm to allow for piecewise linear hitting costs.

In our description of the adaption to piecewise linear hitting costs $f_{\tau}$, we refer to the non-continuous or non-smooth points of $f_{\tau}$ as \emph{breakpoints}\index{breakpoint}. For a piecewise linear function, we can discretize the integral into a summation, replace the first-order derivative at a breakpoint by the difference of consecutive points, and replace the second-order derivative at a point by the difference in slopes of consecutive points~\cite{Bansal2015}. Let $B_{f_{\tau}}$ denote the set of breakpoints of $f_{\tau}$. Further, we denote by $x_{\tau,l}$ and $x_{\tau,r}$ the values of $x_l$ and $x_r$ at time $\tau$, respectively. It is easy to see that the set of breakpoints of $p_{\tau}$ is then given by \begin{align*}
    B_{p_{\tau}} := \{0, m\} \cup \left(B_{f_1} \cup \{x_{1,l}, x_{1,r}\}\right) \cup \dots \cup \left(B_{f_{\tau}} \cup \{x_{\tau,l}, x_{\tau,r}\}\right).
\end{align*}

Let $B_{p_{\tau}}^I := B_{p_{\tau}} \cap I$. The integral of $p_{\tau}$ over $I \subseteq \mathbb{R} \cup \{-\infty, \infty\}$ can then be computed using the quadrature methods described previously by integrating piecewise: \begin{align*}
    \int_I p_{\tau}(y) \,dy = \int_{\min I}^{\min B_{p_{\tau}}^I} p_{\tau}(y) \,dy + \int_{\max B_{p_{\tau}}^I}^{\max I} p_{\tau}(y) \,dy + \sum_{(i, j) \in \text{sort}(B_{p_{\tau}}^I)} \int_i^j p_{\tau}(y) \,dy
\end{align*} where $\text{sort}(A)$ denotes the pairs of consecutive elements of some set $A \subset \mathbb{R}$ in ascending order. We can continue to use finite difference approximations for the first-order and second-order derivatives.

\paragraph{} Note that the computation of $p_{\tau}$ requires $\mathcal{O}(\tau)$ many approximations of the second-order derivative of $f_{t}$. Therefore, any evaluation of $p_{\tau}$ requires $\mathcal{O}(\tau C)$ time and we are thus able to compute the $\epsilon$-optimal integral in $\mathcal{O}(\tau C I_{\epsilon} |B_{p_{\tau}}|) = \mathcal{O}(\tau^2 C I_{\epsilon} |B_{f_0}|)$ time, assuming $B_{f_{\tau}}$ add a constant number of new breakpoints in each time step and with $B_{f_0}$ denoting the number of breakpoints that is shared between multiple $f_{\tau}$. Overall, the described convex optimizations can be solved $\epsilon$-optimally in $\mathcal{O}(\tau^2 C I_{\epsilon} |B_{f_0}| R_{\epsilon} O_{\epsilon}^1)$ time. This is also the time complexity of the algorithm. This shows that, similarly to lazy capacity provisioning, the computational complexity grows polynomially with time. However, unlike lazy capacity provisioning, we cannot regularly reset the history, rendering this algorithm computationally inefficient for short time slot lengths. The algorithm is described in \cref{alg:ud:probabilistic}.

\begin{algorithm}
    \caption{Probabilistic algorithm~\cite{Bansal2015}}\label{alg:ud:probabilistic}
    \SetKwInOut{Input}{Input}
    \Input{$\mathcal{I}_{\text{SSCO}} = (\tau \in \mathbb{N}, m \in \mathbb{N}, \beta \in \mathbb{R}_{>0}, (f_1, \dots, f_{\tau}) \in (\mathbb{R}_{\geq 0} \to \mathbb{R}_{\geq 0})^{\tau})$}
    $\hat{x} \gets \argmin_{x \in \mathcal{X}} f_{\tau}(x)$\;
    find $x_r$ using the optimization described by \cref{eq:ud:probabilistic:right} subject to $x \in \mathcal{X}$\;
    find $x_l$ using the optimization described by \cref{eq:ud:probabilistic:left} subject to $x \in \mathcal{X}$\;
    set $p_{\tau}$ based on the update rule in \cref{eq:ud:probabilistic:update}\;
    \Return $\int_{x_l}^{x_r} y \cdot p_{\tau}(y) \,dy$\;
\end{algorithm}

Updating the probability distribution $p$ can be done in constant time as this does not require any function evaluations. As discussed in the beginning of this subsection, given a uniformly picked $\gamma \in [0,1]$, the randomized algorithm chooses $X_{\tau}$ (randomly) such that $\int_{-\infty}^{X_{\tau}} p_{\tau}(y) \,dy = \gamma$. In other words, $P_{\tau}(X_{\tau}) \sim \text{Unif}(0,1)$ where $P_{\tau}$ is the cumulative distribution function of $p_{\tau}$. By the universality of the uniform, $X_{\tau}$ is $P_{\tau}$-distributed, i.e. $X_{\tau} \sim P_{\tau}$. Hence, $E X_{\tau} = \int_{x_l}^{x_r} y \cdot p_{\tau}(y) \,dy$ computes the configuration for time slot $\tau$ as $p_{\tau}(y) = 0$ for $y \not\in [x_l, x_r]$. We then return $E X_{\tau}$. Similarly to the previously discussed integrals, this integral can be computed $\epsilon$-optimally in $\mathcal{O}(\tau^2 C I_{\epsilon} |B_{f_0}|)$ time, not affecting the asymptotic time complexity of the algorithm.

\subsection{Randomly Biased Greedy Algorithm}\label{section:online_algorithms:ud:rbg}

We have seen many constant-competitive online algorithms for uni-dimensional smoothed convex optimization. However, we have not yet paid much attention to minimizing regret. This is also partially because sublinear regret can be achieved easily using online gradient descent. As this approach generalizes to the multi-dimensional setting, we discuss this approach in \cref{section:online_algorithms:md:descent_methods}, where we look more generally at descent methods.

Still, one important question regarding the uni-dimensional setting remains, namely, whether there is an algorithmic framework that achieves both a constant-competitive ratio and sublinear regret. While it is impossible to achieve both simultaneously, it is possible to develop an algorithmic framework that balances both performance metrics.

In their paper, where the incompatibility between the competitive ratio and regret was first introduced, \citeauthor{Andrew2015}~\cite{Andrew2015} already proposed an algorithmic framework for SCO balancing the two notions in the uni-dimensional setting. Their approach is to scale the norm used to penalize movement in the decision space with a parameter $\theta \geq 1$. If $\theta = 1$, i.e., the algorithm solves the original problem, their algorithm is $2$-competitive and has linear regret. In contrast, for $\theta > 1$, movement in the decision space is penalized more. This allows reducing the regret to an arbitrary amount (which still depends linearly on $T$) while maintaining a constant competitive ratio~\cite{Andrew2015}.

Note that while regret is understood as introduced in \cref{section:theory:performance_metrics}, the competitive ratio is understood with respect to the modified problem with lookahead $1$. In general, with \emph{lookahead}\index{lookahead} $i$, the environment plays actions $i$ steps before the agent follows suit. In other words, a step at time $i$ is evaluated using the cost function from time $t-i$, and the initial step $X_i$ is $\mathbf{0}$. For lookahead $i$, we consider the modified overall cost \begin{align*}
    \sum_{t=1}^T f_t(X_{t+i}) + \norm{X_{t+i} - X_{t+i-1}}.
\end{align*} Note that in this definition we assume that $f_t$ is known after $X_t$ is played, whereas in our original definition of SOCO we assumed that $f_t$ is known before $X_t$ is played. Hence, for $i=1$, this modified problem is equivalent to SOCO and similar to metrical task systems where, in both cases, the environment plays first. In contrast, with online convex optimization, the agent plays first, resulting in a lookahead of $i=0$~\cite{Andrew2015}. Note that online convex optimization is the more restricted setting as the agent has less knowledge than in metrical task systems, implying that given an algorithm for lookahead $i$, the corresponding algorithm for lookahead $0$ is as competitive. Given an algorithm with lookahead $i$, the corresponding algorithm with lookahead $0$ is obtained simply by shifting determined configurations $i$ time slots into the past.

\citeauthor{Bansal2015}~\cite{Bansal2015} mention in their paper that the claims on this algorithm were withdrawn, however, \citeauthor{Andrew2015}~\cite{Andrew2015} clarified their proof, showing it is correct as stated in the original paper~\cite{Wierman}.

The algorithm of \citeauthor{Andrew2015}~\cite{Andrew2015} is initialized with a random parameter $r$ which is uniformly sampled from $({-1,1})$, i.e. $r \sim \text{Unif}(-1, 1)$. They define the work function \begin{align}\label{eq:randomly_biased_greedy:work_function}
    w_{\tau}(x) = \min_{y \in \mathcal{X}} w_{\tau-1}(y) + f_{\tau}(y) + \theta \norm{x - y}
\end{align} where $w_0(x) = \theta \norm{x}$. During each time slot $\tau$ the algorithm moves to the configuration $x$ minimizing $w_{\tau-1}(x) + r \theta \norm{x}$. The resulting algorithm is described in \cref{alg:ud:rbg}.

\begin{algorithm}
    \caption{Randomly Biased Greedy~\cite{Andrew2015}}\label{alg:ud:rbg}
    \SetKwInOut{Input}{Input}
    \Input{$\mathcal{I}_{\text{SCO}} = (\tau \in \mathbb{N}, \mathcal{X} \subset \mathbb{R}, \norm{\cdot}, (f_1, \dots, f_{\tau}) \in (\mathcal{X} \to \mathbb{R}_{\geq 0})^{\tau}), \theta \geq 1, r \sim \text{Unif}(-1,1)$}
    $x \gets \argmin_{x \in \mathcal{X}} w_{\tau-1}(x) + r \theta \norm{x}$\;
    set $w_{\tau}$ as described in \cref{eq:randomly_biased_greedy:work_function}\;
    \Return $x$\;
\end{algorithm}

Observe that as expected, the algorithm returns $X_1 = \mathbf{0}$ for the initial step.

\citeauthor{Andrew2015}~\cite{Andrew2015} show that given a $\theta \geq 1$ their algorithm attains the $\alpha$-unfair competitive ratio $(1+\theta) / \min \{\theta, \alpha\}$ and regret $\mathcal{O}(\max \{T / \theta, \theta\})$. Hence, for $\alpha = 1$ and $\theta = 1$ the algorithm is $2$-competitive. For any $\alpha > 0$, the optimal $\alpha$-unfair competitive ratio is $1 + 1 / \alpha$ and obtained by setting $\theta = \alpha$. In contrast, $\theta = 1 / \epsilon$ for some $\epsilon > 0$ yields the minimal regret $\mathcal{O}(\epsilon T)$~\cite{Andrew2015}. In general, when $T$ is known in advance, \citeauthor{Andrew2015}~\cite{Andrew2015} show that for $\theta \in \mathcal{O}(\sqrt{T})$, their algorithm achieves a $\mathcal{O}(\sqrt{T})$ $\alpha$-unfair competitive ratio and $\mathcal{O}(\sqrt{T})$-regret.

It is easy to see that the work function itself is convex as it can be interpreted as the inf-projection of the convex function $f(x,y) = w_{\tau-1}(y) + f_{\tau}(y) + \theta \norm{x - y}$. This fact is shown in proposition 2.22 of~\cite{Burke2015}. The evaluation of $w_{\tau}$ requires $\mathcal{O}(\tau)$ recursive evaluations of the work function each of which uses a convex optimization and evaluates the hitting costs. Thus, an $\epsilon$-optimal evaluation of $w_{\tau}$ can be obtained in $\mathcal{O}(C (O_{\epsilon}^1)^{\tau})$ time. Hence, the overall time complexity of the algorithm is in $\mathcal{O}(C (O_{\epsilon}^1)^{\tau+1})$. We improve the practical runtime by memoizing the work function.

\subsection{Randomized Integral Relaxation}

\citeauthor{Albers2018}~\cite{Albers2018} use the probabilistic algorithm described in \cref{section:online_algorithms:ud:probabilistic} in their randomized algorithm for Int-SSCO achieving the optimal competitive ratio $2$ against an oblivious adversary. Although they used the probabilistic algorithm in their paper, their proof generalizes to any $2$-competitive fractional algorithm, so, in particular, the randomly biased greedy algorithm can be used as well. Roughly, the algorithm works by solving the relaxed problem using the probabilistic algorithm of \citeauthor{Bansal2015}~\cite{Bansal2015} and then randomly rounding the resulting fractional schedule.

Let $\bar{\mathcal{I}} = (T, m, \beta, \bar{F})$ with $\bar{F} = (\bar{f}_1, \dots, \bar{f}_T)$ be the fractional relaxation of the instance $\mathcal{I} = (T, m, \beta, F)$ of Int-SSCO and let $\bar{\mathcal{X}} = [0,m]$ denote the decision space of $\bar{\mathcal{I}}$. \citeauthor{Albers2018}~\cite{Albers2018} define the relaxed operating costs $\bar{f}_{\tau} : \bar{\mathcal{X}} \to \mathbb{R}_{\geq 0}$ as the linear interpolation of the integral operating costs $f_{\tau}$: \begin{align*}
    \bar{f}_{\tau} := \begin{cases}
        f_{\tau}(x) & x \in [m]_0 \\
        (\lceil x \rceil - x) f_{\tau}(\lfloor x \rfloor) + (x - \lfloor x \rfloor) f_{\tau}(\lceil x \rceil) & \text{otherwise}.
    \end{cases}
\end{align*} Note that $\bar{f_{\tau}}$ are continuous and piecewise linear with the set of breakpoints $[m]_0$. Hence, we are able to use \cref{alg:ud:probabilistic} to obtain the configuration $\bar{X}_{\tau} \in \bar{\mathcal{X}}$ at time $\tau$ for the relaxed problem instance  $\bar{\mathcal{I}}$. Further, let $\text{frac}(x) = x - \lfloor x \rfloor$ be the fractional part of $x$ and let $\bar{X}'_{\tau-1} = (\bar{X}_{\tau-1})_{\lfloor\bar{X}_{\tau}\rfloor}^{\lceil\bar{X}_{\tau}\rceil}$ be the projection of the preceding relaxed configuration onto the discrete interval of the current relaxed configuration.

The randomized algorithm distinguishes between time slots where the configuration is increased and time slots where the configuration is decreased. In the first case, i.e. $\bar{X}_{\tau-1} \leq \bar{X}_{\tau}$, if $X_{\tau-1} = \lceil\bar{X}_{\tau}\rceil$  the configuration remains unchanged. Otherwise, $X_{\tau}$ is set to $\lceil\bar{X}_{\tau}\rceil$ with probability \begin{align*}
    p_{\tau}^{\uparrow} := \frac{\bar{X}_{\tau} - \bar{X}'_{\tau-1}}{1 - \text{frac}(\bar{X}'_{\tau-1})}
\end{align*} and to $\lfloor\bar{X}_{\tau}\rfloor$ with probability $1 - p_{\tau}^{\uparrow}$. Conversely, if $\bar{X}_{\tau-1} > \bar{X}_{\tau}$, the configuration remains unchanged if $X_{\tau-1} = \lfloor\bar{X}_{\tau}\rfloor$, and otherwise with probability \begin{align*}
    p_{\tau}^{\downarrow} := \frac{\bar{X}'_{\tau-1} - \bar{X}_{\tau}}{\text{frac}(\bar{X}'_{\tau-1})}
\end{align*} the configuration is set to $\lfloor\bar{X}_{\tau}\rfloor$ and with probability $p_{\tau}^{\downarrow}$ the configuration is set to $\lceil\bar{X}_{\tau}\rceil$. The resulting algorithm is shown in \cref{alg:ud:randomized}.

\begin{algorithm}
    \caption{Randomized integral relaxation~\cite{Albers2018}}\label{alg:ud:randomized}
    \SetKwInOut{Input}{Input}
    \Input{$\mathcal{I}_{\text{Int-SSCO}} = (\tau \in \mathbb{N}, m \in \mathbb{N}, \beta \in \mathbb{R}_{>0}, (f_1, \dots, f_{\tau}) \in (\mathbb{N}_0 \to \mathbb{R}_{\geq 0})^{\tau})$}
    $\bar{X}_{\tau} \gets \text{\cref{alg:ud:probabilistic}}(\bar{\mathcal{I}}_{\text{Int-SSCO}})$\;
    \eIf{$\bar{X}_{\tau-1} \leq \bar{X}_{\tau}$}{
        \eIf{$X_{\tau-1} = \lceil\bar{X}_{\tau}\rceil$}{
            \Return $\lceil\bar{X}_{\tau}\rceil$\;
        }{
            $\gamma \sim \text{Unif}(0,1)$\;
            \eIf{$\gamma \leq p_{\tau}^{\uparrow}$}{
                \Return $\lceil\bar{X}_{\tau}\rceil$\;
            }{
                \Return $\lfloor\bar{X}_{\tau}\rfloor$\;
            }
        }
    }{
        \eIf{$X_{\tau-1} = \lfloor\bar{X}_{\tau}\rfloor$}{
            \Return $\lfloor\bar{X}_{\tau}\rfloor$\;
        }{
            $\gamma \sim \text{Unif}(0,1)$\;
            \eIf{$\gamma \leq p_{\tau}^{\downarrow}$}{
                \Return $\lfloor\bar{X}_{\tau}\rfloor$\;
            }{
                \Return $\lceil\bar{X}_{\tau}\rceil$\;
            }
        }
    }
\end{algorithm}

We use the universality of the uniform to simulate Bernoulli-distributed random variables with parameters $p_{\tau}^{\uparrow}$ and $p_{\tau}^{\downarrow}$, respectively. Any pseudo-random number generator can be used to produce the uniformly distributed $\gamma$. It is easy to see that the time complexity is given by the time complexity of \cref{alg:ud:probabilistic}, i.e., $\mathcal{O}(\tau^2 m C I_{\epsilon} R_{\epsilon} O_{\epsilon}^1)$ with $|B_{f_0}| \in \mathcal{O}(m)$, or the complexity of \cref{alg:ud:rbg}, i.e., $\mathcal{O}(C (O_{\epsilon}^1)^{\tau+1})$ depending on which algorithm is used for the relaxed problem.

\section{Multi-Dimensional}\label{section:online_algorithms:md}

\subsection{Lazy Budgeting}\label{section:online_algorithms:md:lazy_budgeting}

To begin with our discussion of the multi-dimensional setting, we examine two algorithms developed by \citeauthor{Albers2021}~\cite{Albers2021} for a restricted class of convex cost functions. Their first algorithm is $2d$-competitive for SLO, i.e., load and time-independent costs, and their second algorithm is $(2d+1)$-competitive for SBLO. Note that we only defined SBLO and SLO for the integral case. As no online algorithm for SLO can attain a competitive ratio smaller than $2d$, their first algorithm is optimal, and their second algorithm is nearly optimal~\cite{Albers2021, Albers2021_2}.

The idea behind both algorithms is to calculate optimal schedules up to the current time slot. Depending on this schedule, the algorithm decides if a server is powered up. To perform the smoothing, the algorithm remembers how long a server was idling and powers this server down if the idle duration surpasses a threshold. \citeauthor{Albers2021}~\cite{Albers2021} do not name their algorithms, yet, within this work, we refer to them as \emph{lazy budgeting methods}\index{lazy budgeting}.

\subsubsection{Lazy Budgeting for Smoothed Load Optimization}

Let $\mathcal{I} = (d, \tau, m, \beta, \Lambda, c)$ be an instance of SLO. \citeauthor{Albers2021}~\cite{Albers2021} focus on a setting without inefficient server types. A server type $k$ is called \emph{inefficient} if there exists another server type $k'$ where both the operating cost and the switching cost is lower, i.e. $c_k \geq c_{k'}$ and $\beta_k \geq \beta_{k'}$. In practice, this is not a restriction as a server of an inefficient server type is only ever powered up if all more efficient servers are already active as there is no trade-off between operating and switching costs. Furthermore, typically servers with a lower operating cost also have a higher switching cost. In addition, \citeauthor{Albers2021} assume that there are no duplicated server types, i.e., server types with equal operating and switching costs.

We assume that the server types are sorted in descending order by their operating costs, i.e. $c_1 > \dots > c_d$. As we excluded inefficient server types, switching costs are in ascending order, $\beta_1 < \dots < \beta_d$.

The algorithm of \citeauthor{Albers2021}~\cite{Albers2021} separates a problem instance into $m := \sum_{k=1}^d m_k$ lanes. Recall that with SLO, we assume that each active server can handle a single job during each time slot. The algorithm uses that at time slot $t$ there is a job in line $j$ if and only if $j \leq \lambda_{t}$. Thus, all servers represented by lines $j > \lambda_{t}$ are either inactive or idling.

Let $y_{t,j}$ denote the server type that handles the $j$-th lane during time slot $t$ given some underlying schedule $X$. We say $y_{t,j} = 0$ if there is no active server in lane $j$ during time slot $t$, which is only the case if $j > \lambda_{t}$ as we can assume that $\lambda_{t} \leq m$ holds for all time slots $t \in [T]$. \citeauthor{Albers2021}~\cite{Albers2021} give the following formal definition: \begin{align*}
    y_{t,j} := \begin{cases}
        \max \{k \in [d] \mid \sum_{k' = k}^d X_{t,k'} \geq j\} & k \in \left[\sum_{k=1}^d X_{t,k}\right] \\
        0 & \text{otherwise}.
    \end{cases}
\end{align*} By this definition, we prefer to use servers of the server type with the lowest operating cost (and largest switching cost). In other words, the server types handling each lane $y_{t,1}, \dots, y_{t,m}$ are sorted in descending order, i.e. $y_{t,j} \geq y_{t,j'}$ for $j < j'$. We denote by $\hat{y}_{t,j}^{\tau}$ the server type in lane $j$ during time slot $t$ induced by some optimal schedule $\hat{X}^{\tau}$ up to time $\tau$ and by $\widetilde{y}_{t,j}$ the server type in lane $j$ during time slot $t$ as assigned by the algorithm.

The algorithm begins by finding an optimal schedule $\hat{X}^{\tau}$ up to time slot $\tau$. This schedule is chosen such that the server type in a lane of $\hat{X}^{\tau}$ is never reduced compared to the previously used optimal schedule $\hat{X}^{\tau-1}$, i.e. $\hat{y}_{t,j}^{\tau} \geq \hat{y}_{t,j}^{\tau-1}$ for all time slots $t \in [\tau]$ and lanes $j \in [m]$. Moreover, we assume that $\hat{X}^{\tau}$ is a schedule that powers up servers as late as possible and powers down servers as early as possible. This is necessary in case $c_k = 0$ for some server type $k \in [d]$. We observe that these properties are fulfilled by all optimal schedules that \cref{alg:md:optimal_graph_search} obtains. We cache the results of the optimal graph-based algorithm such that in every iteration of the online algorithm, only one dynamic update needs to be performed. The asymptotic time complexity of this dynamic update is thus given as $\mathcal{O}(|\mathcal{M}| C d)$ where $|\mathcal{M}| \in \mathcal{O}(\prod_{k=1}^d m_k)$.

Now, the algorithm ensures that no server type is used for lane $j \in [m]$ that is smaller than the server type used by $\hat{X}^{\tau}$, i.e. $\widetilde{y}_{\tau,j} \geq \hat{y}_{\tau,j}^{\tau}$. If $\widetilde{y}_{\tau-1,j} < \hat{y}_{\tau,j}^{\tau}$, a server of type $\widetilde{y}_{\tau-1,j}$ is powered down and a server of type $\hat{y}_{\tau,j}^{\tau}$ is powered up. A server of type $k$ that is not replaced by a greater server type remains active for $\bar{t}_k := \lfloor \beta_k / c_k \rfloor$ time slots~\cite{Albers2021}. If $\hat{X}^{\tau}$ uses a smaller server type $k' \leq k$ in the meantime, then the server of type $k$ will run for at least $\bar{t}_{k'}$ further time slots.

The algorithm computes $\widetilde{y}_{\tau,j}$ directly. The corresponding number of active servers of type $k$ of the underlying schedule $\widetilde{X}$ can be obtained by $\widetilde{X}_{\tau,j} = |\{j \in [m] \mid \widetilde{y}_{\tau,j} = k\}|$. The resulting algorithm is shown in \cref{alg:md:lazy_budgeting:det_slo}. Here, $h_j$ denotes the time until the server handling line $j \in [m]$ is powered down.

\begin{algorithm}
    \caption{Lazy Budgeting for SLO~\cite{Albers2021}}\label{alg:md:lazy_budgeting:det_slo}
    \KwIn{$\mathcal{I}_{\text{SLO}} = (d \in \mathbb{N}, \tau \in \mathbb{N}, m \in \mathbb{N}^d, \beta \in \mathbb{R}_{>0}^d, \Lambda \in \mathbb{N}_0^{\tau}, c \in \mathbb{R}_{\geq 0}^d)$}
    Update the previously found optimal schedule $\hat{X}^{\tau-1}$ to $\hat{X}^{\tau}$ such that $\hat{y}_{t,j}^{\tau} \geq \hat{y}_{t,j}^{\tau-1}$ for all  $j \in [m]$\;
    \For{$j \gets 1$ \KwTo $m$}{
        \eIf{$\widetilde{y}_{\tau-1,j} < \hat{y}_{\tau,j}^{\tau}$ \KwOr $t \geq h_j$}{
            $\widetilde{y}_{\tau,j} \gets \hat{y}_{\tau,j}^{\tau}$\;
            $h_j \gets \tau + \bar{t}_{\hat{y}_{\tau,j}^{\tau}}$\;
        }{
            $\widetilde{y}_{\tau,j} \gets \widetilde{y}_{\tau-1,j}$\;
            $h_j \gets \max \{h_j, \tau + \bar{t}_{\hat{y}_{\tau,j}^{\tau}}\}$ where $\bar{t}_0 = 0$\;
        }
    }
    \ForEach{$k \in [d]$}{
        $X_{\tau,k} \gets |\{j \in [m] \mid \widetilde{y}_{\tau,j} = k\}|$\;
    }
    \Return $X_{\tau}$\;
\end{algorithm}

\begin{function}
	\caption{BuildLanes($x, d, m$)}\label{proc:md:lazy_budgeting:build_lanes}
	$y \gets \mathbf{0}$\;
	\For{$j \gets 1$ \KwTo $m$}{
	    \If{$j \leq \sum_{k=1}^d x_k$}{
	        \For{$k \gets 1$ \KwTo $d$}{
        	    \If{$\sum_{k'=k}^d x_{k'} \geq j$}{
        	        $y_j \gets k$\;
        	    }
        	}
	    }
	}
    \Return $y$\;
\end{function}

Schedules can be converted to lanes in $\mathcal{O}(m d^2)$ time as described by \ref{proc:md:lazy_budgeting:build_lanes} and lanes can be converted back to schedules in $\mathcal{O}(m)$ time by iterating over all lanes. Hence, the overall asymptotic time complexity of the algorithm is described by adding the time required to find the optimal schedule, i.e., $\mathcal{O}(m d^2 + C d \prod_{k=1}^d m_k)$.

\subsubsection{Randomized Lazy Budgeting for Smoothed Load Optimization}

\citeauthor{Albers2021}~\cite{Albers2021} describe how the competitive ratio of the previously described \cref{alg:md:lazy_budgeting:det_slo} can be improved to $\frac{e}{e-1}d \approx 1.582d$ against an oblivious adversary by randomizing the running time of a server. Before execution, the randomized algorithm chooses $\gamma \in [0,1]$ according to the probability density function \begin{align*}
    f_{\gamma}(x) = \begin{cases}
        e^x / (e-1) & x \in [0,1] \\
        0 & \text{otherwise}.
    \end{cases}
\end{align*} Then, the running time of a server of type $k \in [d]$, $\bar{t}_k$, is set to $\lfloor \gamma \cdot \beta_k / c_k \rfloor$.

To sample $\gamma$ we first seek to find the cumulative distribution function $F_{\gamma}$. For $x \in [0,1]$ we have \begin{align*}
    F_{\gamma}(x) &= \int_0^x f_{\gamma}(t) \,dt \\
                  &= \frac{1}{e-1} \int_0^x e^t \,dt \\
                  &= \frac{1}{e-1} (e^x - 1).
\end{align*} By the universality of the uniform, realizations of $F_{\gamma}$ can be simulated given $U \sim \text{Unif}(0,1)$ as $F_{\gamma}^{-1}(U) \sim F_{\gamma}$. It is now easy to see that $F_{\gamma}^{-1}(x) = \ln (x (e - 1) + 1)$.

This completes the description of the implementation of the randomized variant of lazy budgeting for SLO. The asymptotic time complexity is given by the time complexity of \cref{alg:md:lazy_budgeting:det_slo}.

\subsubsection{Lazy Budgeting for Smoothed Balanced-Load Optimization}

In a subsequent paper, \citeauthor{Albers2021_2}~\cite{Albers2021_2} modified their lazy budgeting method to SBLO, allowing for more complex cost functions. The method remains similar: First, an optimal schedule is found which ends at the current time slot. Then, the algorithm ensures for each server type that the number of active servers is at least as large as the number of active servers in the optimal schedule. Lastly, to make the algorithm competitive, if the accumulated idle operating cost of a server of type $k$, $g_{t,k}(0)$ exceeds the switching cost $\beta_k$, a server of type $k$ is powered down.

First, \citeauthor{Albers2021_2}~\cite{Albers2021_2} developed a $(2d+1)$-competitive online algorithm for a setting where the operating costs are time-independent. This allows to determine the runtime of a server in advance as both $g_{t,k}(0)$ and $\beta_k$ are known. Then, they extend their algorithm to a setting that allows for time-dependent operating costs where their algorithm attains a competitive ratio of $2d+1+\epsilon$ for any $\epsilon > 0$.

\paragraph{Time-Independent Operating Costs}

Again, we denote by $\hat{X}^{\tau}$ the optimal schedule with information up to time slot $\tau$ and by $X$ the schedule obtained by the algorithm. If the optimal schedule has more active servers of some type $k \in [d]$ than the schedule obtained by the algorithm, i.e. $\hat{X}_{\tau,k}^{\tau} > X_{\tau-1,k}$, $(\hat{X}_{\tau,k}^{\tau} - X_{\tau-1,k})^+$ servers of type $k$ are powered up. After being active for $\bar{t}_k := \lfloor \beta_k / g_k(0) \rfloor$ time slots, a server of type $k$ is powered down again. The resulting algorithm is shown in \cref{alg:md:lazy_budgeting:sblo_a}. Here, $h_{t,k}$ denotes the number of servers of type $k$ that were powered up at time $t$. We assume that $h_{t,k}$ is initialized with $0$ for all $t \in \mathbb{Z}$ and $k \in [d]$.

\begin{algorithm}
    \caption{Lazy Budgeting for SBLO (for time-independent operating costs)~\cite{Albers2021_2}}\label{alg:md:lazy_budgeting:sblo_a}
    \KwIn{$\mathcal{I}_{\text{SBLO}} = (d \in \mathbb{N}, \tau \in \mathbb{N}, m \in \mathbb{N}^d, \beta \in \mathbb{R}_{>0}^d, \Lambda \in \mathbb{N}_0, G \in (\mathbb{R}_{\geq 0} \to \mathbb{R}_{\geq 0}^d)^d)$}
    Update $\hat{X}^{\tau-1}$ to $\hat{X}^{\tau}$\;
    \For{$k \gets 1$ \KwTo $d$}{
        $X_{\tau,k} \gets X_{\tau-1,k} - h_{\tau - \bar{t}_k, k}$\;
        \If{$X_{\tau,k} < \hat{X}_{\tau,k}^{\tau}$}{
            $X_{\tau,k} \gets \hat{X}_{\tau,k}^{\tau}$\;
            $h_{t,k} \gets \hat{X}_{\tau,k}^{\tau} - X_{\tau,k}$\;
        }
    }
    \Return $X_{\tau}$\;
\end{algorithm}

Similar to our implementation of Lazy Budgeting for SLO (\cref{alg:md:lazy_budgeting:det_slo}), we cache intermediate results of the algorithm computing the optimal schedule up to time slot $\tau$ (\cref{alg:md:optimal_graph_search}). It is therefore enough to perform a single dynamic update to obtain $\hat{X}^{\tau}$ which is possible in $\mathcal{O}(|\mathcal{M}| C d)$ time where $|\mathcal{M}| \in \mathcal{O}(\prod_{k=1}^d m_k)$. It is easy to see that this is also the time complexity of \cref{alg:md:lazy_budgeting:sblo_a}.

\paragraph{Time-Dependent Operating Costs}

Next, \citeauthor{Albers2021_2}~\cite{Albers2021_2} extend \cref{alg:md:lazy_budgeting:sblo_a} to an algorithm which supports time-dependent operating costs and achieves a competitive ratio of $2d + 1 + c(\mathcal{I})$ where $c(\mathcal{I}) := \sum_{k=1}^d \max_{t \in [T]} \frac{g_{t,k}(0)}{\beta_k}$. Now, the idle operating cost $g_{t,k}(0)$ is not constant over time anymore. Since at time $\tau$ we only know the operating costs up to time $\tau$, we cannot pre-determine the number of time slots that a server of type $k$ is active until powered down, yet at the time slot when a server is powered down we know all relevant cost functions. \citeauthor{Albers2021_2}~\cite{Albers2021_2} formally define the maximal number of time slots such that the sum of the idle operating costs beginning from the next time slot, $\tau+1$, is smaller than or equal to the switching cost as \begin{align*}
    \bar{t}_{t,k} := \max \left\{\bar{t} \in [T - t] \mid \sum_{t' = t+1}^{t+\bar{t}} g_{t,k}(0) \leq \beta_k\right\}
\end{align*} In contrast to \cref{alg:md:lazy_budgeting:sblo_a} where servers of type $k$ that were powered up at time $\tau$ were active for $\bar{t}_{\tau,k}$ time slots, they are now active for $\bar{t}_{\tau,k} + 1$ time slots.

Let $W_{\tau,k}$ denote the set of all time slots $t$ such that servers of type $k$ that were powered up at time $t$ are powered down during time slot $\tau$, i.e., $t + \bar{t}_{t,k} + 1 = \tau$. Again, we denote by $h_{t,k}$ the number of servers of type $k$ that were powered up at time $t$. Using the same powering-up policy as \cref{alg:md:lazy_budgeting:sblo_a}, the updated algorithm is described in \cref{alg:md:lazy_budgeting:sblo_b}.

\begin{algorithm}
    \caption{Lazy Budgeting for SBLO (for time-dependent operating costs)~\cite{Albers2021_2}}\label{alg:md:lazy_budgeting:sblo_b}
    \KwIn{$\mathcal{I}_{\text{SBLO}} = (d \in \mathbb{N}, \tau \in \mathbb{N}, m \in \mathbb{N}^d, \beta \in \mathbb{R}_{>0}^d, \Lambda \in \mathbb{N}_0^{\tau}, G \in (\mathbb{R}_{\geq 0} \to \mathbb{R}_{\geq 0}^d)^{d^{\tau}})$}
    Calculate $\hat{X}^{\tau}$\;
    \For{$k \gets 1$ \KwTo $d$}{
        $W_{\tau,k} \gets \left\{t \in [\tau-1] \mid \sum_{t'=t+1}^{\tau-1} g_{t',k}(0) \leq \beta_k < \sum_{t'=t+1}^{\tau} g_{t',k}(0)\right\}$\;
        $X_{\tau,k} \gets X_{\tau-1,k} - \sum_{t \in W_{\tau,k}} h_{t, k}$\;
        \If{$X_{\tau,k} < \hat{X}_{\tau,k}^{\tau}$}{
            $X_{\tau,k} \gets \hat{X}_{\tau,k}^{\tau}$\;
            $h_{t,k} \gets \hat{X}_{\tau,k}^{\tau} - X_{\tau,k}$\;
        }
    }
    \Return $X_{\tau}$\;
\end{algorithm}

$W_{\tau,k}$ can be computed in $\mathcal{O}(\tau^2 C)$ time. Thus, the overall time complexity of the algorithm is $\mathcal{O}(\tau^2 |\mathcal{M}| C d)$.

\paragraph{Reducing the Competitive Ratio}

We now describe how \citeauthor{Albers2021_2}~\cite{Albers2021_2} improve the competitive ratio of their algorithm to $2d + 1 + \epsilon$ for any $\epsilon > 0$. The idea is to consider a modified problem instance $\widetilde{\mathcal{I}} = (d, \widetilde{\tau}, m, \beta, \widetilde{\Lambda}, \widetilde{G})$ which divides each time slot $t$ of the original problem instance $\mathcal{I}$ into $\widetilde{n}_t$ sub time slots, allowing for up to $\widetilde{n}_t$ intermediate state changes. Our goal is to choose $\widetilde{\mathcal{I}}$ such that $c(\widetilde{\mathcal{I}})$ becomes arbitrarily small.

In our description, we refer to time slots of $\mathcal{I}$ by $t$ and to sub time slots of $\widetilde{\mathcal{I}}$ by $u$. The total number of sub time slots is given by $\widetilde{\tau} := \sum_{t=1}^{\tau} \widetilde{n}_t$. We denote by $U(t) = [u+1 : u+\widetilde{n}_t] = \{u+1, u+2, \dots, u+\widetilde{n}_t\} \subseteq [\widetilde{\tau}]$ with $u = \sum_{t'=1}^{t-1} \widetilde{n}_t$ the set of sub time slots corresponding to time slot $t$. In contrast, let $U^{-1}(u)$ be the time slot $t \in [\tau]$ such that $u \in U(t)$.

The operating cost $g_{t,k}(l)$ of a server of type $k$ during time slot $t$ under load $l$ is divided equally among all sub time slots $\widetilde{n}_t$, i.e. $\widetilde{g}_{u,k}(l) := g_{U^{-1}(u),k}(l) / \widetilde{n}_{U^{-1}(u)}$. The job volume does not change, so $\widetilde{\lambda}_u := \lambda_{U^{-1}(u)}$.

\citeauthor{Albers2021_2}~\cite{Albers2021_2} show that for \begin{align*}
    \widetilde{n}_t = \left\lceil\frac{d}{\epsilon} \cdot \max_{k \in [d]} \frac{g_{t,k}(0)}{\beta_k}\right\rceil
\end{align*} the cost of the resulting schedule is at most $2d + 1 + \epsilon$ times larger than the cost of an optimal solution. For $\epsilon \to 0$, the competitive ratio converges to $2d + 1$.

Let $\widetilde{X}$ be the schedule obtained by \cref{alg:md:lazy_budgeting:sblo_b} for $\widetilde{\mathcal{I}}$. Then, the schedule $X$ for the original problem instance $\mathcal{I}$ can be obtained by setting $X_{\tau} = \widetilde{X}_{\mu(\tau)}$ where $\mu(\tau) = \argmin_{u \in U(\tau)} \widetilde{f}_{u}(\widetilde{X}_u)$ is the configuration that minimizes the operating cost during $U(t)$~\cite{Albers2021_2}.

The resulting algorithm is shown in \cref{alg:md:lazy_budgeting:sblo_c}. First, the modified problem instance $\widetilde{\mathcal{I}}$ is created, and the next $\widetilde{n}_{\tau}$ time steps are simulated using \cref{alg:md:lazy_budgeting:sblo_b}. Then, the resulting schedule $X$ is constructed from $\widetilde{X}$.

\begin{algorithm}
    \caption{Lazy Budgeting for SBLO~\cite{Albers2021_2}}\label{alg:md:lazy_budgeting:sblo_c}
    \KwIn{$\mathcal{I}_{\text{SBLO}} = (d \in \mathbb{N}, \tau \in \mathbb{N}, m \in \mathbb{N}^d, \beta \in \mathbb{R}_{>0}^d, \Lambda \in \mathbb{N}_0^{\tau}, G \in (\mathbb{R}_{\geq 0} \to \mathbb{R}_{\geq 0}^d)^{d^{\tau}})$}
    $\widetilde{n}_{\tau} \gets \left\lceil d / \epsilon \cdot \max_{k \in [d]} g_{\tau,k}(0) / \beta_k\right\rceil$\;
    Extend the modified problem instance $\widetilde{\mathcal{I}}$ by $\widetilde{n}_{\tau}$ additional time slots\;
    Update $\widetilde{X}$ by executing the next $\widetilde{n}_{\tau}$ time slots in \cref{alg:md:lazy_budgeting:sblo_b}\;
    $X_{\tau} \gets \widetilde{X}_{\mu(\tau)}$ with $\mu(\tau) = \argmin_{u \in U(\tau)} \widetilde{f}_u(\widetilde{X}_u)$\;
    \Return $X_{\tau}$\;
\end{algorithm}

The number of sub time slots $\widetilde{n}_t$ can be determined in $\mathcal{O}(C d)$ time. Creating the modified problem instance takes $\mathcal{O}(\widetilde{n}_{\tau})$ time. Simulating \cref{alg:md:lazy_budgeting:sblo_b} for $\widetilde{n}_{\tau}$ sub time slots takes $\mathcal{O}(\widetilde{n}_{\tau} \widetilde{\tau}^2 |\mathcal{M}| C d)$ time. The obtained schedule can be constructed in $\mathcal{O}(\widetilde{n}_{\tau} C)$ time. Hence, the overall asymptotic time complexity of \cref{alg:md:lazy_budgeting:sblo_c} is given by $\mathcal{O}(\widetilde{n}_{\tau} \widetilde{\tau}^2 |\mathcal{M}| C d)$, which is inversely proportional to $\epsilon^3$.

\subsection{Descent Methods}\label{section:online_algorithms:md:descent_methods}

When we seek fractional solutions, and the cost functions are differentiable, a promising method is to descent towards the minimizer of the current cost function. This is an approach that is commonly used in online convex optimization, yielding algorithms with sublinear regret~\cite{Andrew2015}. As mentioned in \cref{section:online_algorithms:ud:rbg} on the Randomly Biased Greedy algorithm, movement costs are not considered in the online convex optimization setting, and the agent picks a point in the decision space before the cost function is revealed. A commonly used algorithm that achieves no-regret in the setting of online convex optimization is \emph{online gradient descent} (OGD).

\subsubsection{Online Gradient Descent}

Online gradient descent works by selecting an arbitrary initial point $X_1 \in \mathcal{X}$ and then choosing, at time $\tau > 1$, $X_{\tau} = \Pi_{\mathcal{X}}(X_{\tau-1} - \eta_{\tau-1} \nabla f_{\tau-1}(X_{\tau-1}))$ where $\eta_t$ are the learning rates and $\Pi_{\mathcal{X}}(x)$ is the euclidean projection of $x$ onto $\mathcal{X}$~\cite{Andrew2015}. The algorithm is described in \cref{alg:md:ogd}. Similar to RBG, OGD operates with lookahead $1$. In the multi-dimensional setting, the \emph{euclidean projection}\index{euclidean projection} of a point $x \in \mathbb{R}^d$ onto a convex set $K$ is given as $\Pi_{K}(x) = \argmin_{y \in K} \norm{y - x}_2$.

\begin{algorithm}
    \caption{Online Gradient Descent~\cite{Andrew2015}}\label{alg:md:ogd}
    \SetKwInOut{Input}{Input}
    \Input{$\mathcal{I}_{\text{SCO}} = (\tau \in \mathbb{N}, \mathcal{X} \subset \mathbb{R}^d, \norm{\cdot}, (f_1, \dots, f_{\tau}) \in (\mathcal{X} \to \mathbb{R}_{\geq 0})^{\tau}), \eta > 0$}
    $X_{\tau} \gets \Pi_{\mathcal{X}}(X_{\tau-1} - \eta \nabla f_{\tau-1}(X_{\tau-1}))$\;
    \Return $X_{\tau}$\;
\end{algorithm}

With appropriate learning rates OGD achieves no-regret for online convex otpimization. For example, OGD obtains $\mathcal{O}(\sqrt{T})$-regret for $\eta_t \in \Theta(1 / \sqrt{t})$~\cite{Andrew2015}. \citeauthor{Andrew2015}~\cite{Andrew2015} showed that an OGD algorithm with $\mathcal{O}(\rho_2(T))$-regret in the online convex optimization setting, and $\sum_{t=1}^T \eta_t \in \mathcal{O}(\rho_1(T))$, achieves $\mathcal{O}(\rho_1(T) + \rho_2(T))$-regret in the online smoothed convex otpimization setting. In particular, for learning rates $\eta_t \in \Theta(1 / \sqrt{t})$, OGD obtains $\mathcal{O}(\sqrt{T})$-regret when movement costs are considered and the agent picks a point after the hitting costs were revealed~\cite{Andrew2015}.

Using finite difference methods, $\nabla f_{\tau-1}(X_{\tau-1})$ can be computed in $\mathcal{O}(d C)$ time. The euclidean projection can be computed $\epsilon$-optimally in $\mathcal{O}(O_{\epsilon}^d)$ time. Thus, the overall asymptotic time complexity of OGD is $\mathcal{O}(d C O_{\epsilon}^d)$.

\subsubsection{Online Mirror Descent}\label{section:online_algorithms:md:descent_methods:omd}

Thus, concerning regret, where algorithms seek to minimize the hitting cost immediately, the smoothing property of SCO does not require the development of new algorithms. In contrast, competitive algorithms need to wait before moving to the minimizer of the hitting cost until the movement costs are amortized. In other words, in each step, a competitive algorithm has to decide how far to move into the direction of the minimizer to balance hitting cost and movement cost. Crucially, where to move depends on the geometry of the cost function. As is shown in \cref{fig:level_sets_of_the_hitting_costs}, rather than moving towards the minimizer directly, it is advantageous to move to a projection onto some sub-level set of the cost function. This approach minimizes the movement costs that are required to reach a point with the same hitting cost. This approach can balance hitting costs and movement costs and is discussed in \cref{section:online_algorithms:md:descent_methods:obd}, where we introduce the online balanced descent framework.

\begin{figure}
    \centering
    \input{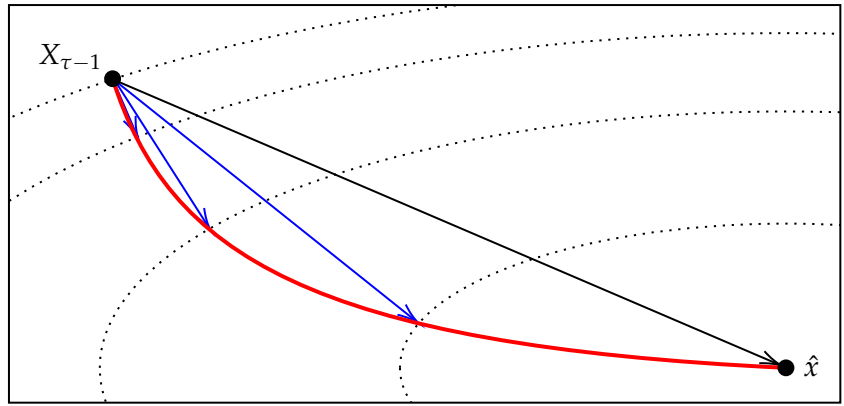}
    \caption{Level sets of $f_{\tau}$ in two dimensions. The blue arrows show projections of $X_{\tau-1}$ onto some level set. The red line visualizes the projection of $X_{\tau-1}$ onto all level sets $\{x \in \mathbb{R}^2 \mid f_{\tau}(x) = l\}$ for $l \in [\hat{x}, f_{\tau}(X_{\tau-1})]$. The step in the direction of the minimizer of $f_{\tau}$ is shown in black. Note that it is not optimal to move directly in the direction of the minimizer as there likely exists a closer point on the same level set. Online balanced descent picks a point on the red line.}
    \label{fig:level_sets_of_the_hitting_costs}
\end{figure}

The mirror descent framework is an extension of gradient descent, allowing to adapt to the underlying ``geometry'' of a problem~\cite{Gupta2020}. The original gradient descent algorithm uses Euclidean geometry. This is shown by a slightly modified form of its update rule: \begin{align*}
    X_{\tau} = \argmin_{x \in \mathcal{X}} \eta_{\tau-1} \langle\nabla f_{\tau-1}(X_{\tau-1}), x\rangle + \frac{1}{2} \norm{x - X_{\tau-1}}_2^2.
\end{align*} Observe that OGD uses the squared Euclidean distance as a regularizer (i.e., a function ensuring that we remain close to the point $X_{\tau-1}$) which can be replaced by another distance to obtain different algorithms~\cite{Gupta2020}.

\paragraph{Proximal Point View} The Bregman divergence is a commonly used class of distance functions. Given a strictly convex \emph{distance-generating function}\index{distance-generating function} $h$, the Bregman divergence measures the deviation of $h$ from its linear approximation.

\begin{definition}\index{Bregman divergence}
\cite{Chen2018} The Bregman divergence from a point $x$ to a point $y$ with respect to a strictly convex function $h$ is given as \begin{align*}
    D_h(x,y) = h(x) - h(y) - \langle\nabla h(y), x - y\rangle.
\end{align*}
\end{definition}

\begin{figure}
    \centering
    \input{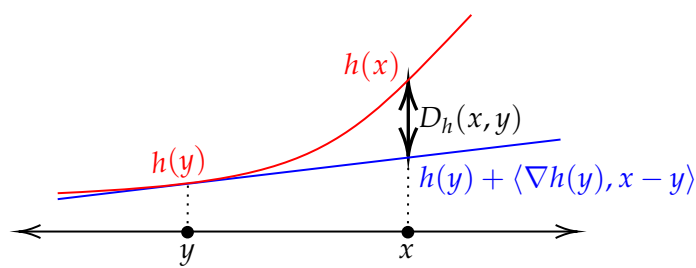}
    \caption{Bregman Divergence $D_h(x,y)$ for a function $h : \mathbb{R} \to \mathbb{R}$ \cite{Chen2018}. The Bregman divergence measures how much a function differs at $x$ from its linear approximation at $y$.}
    \label{fig:bregman_divergence}
\end{figure}

The definition of the Bregman divergence of a univariate function $h$ is visualized in \cref{fig:bregman_divergence}.

The modified variant of OGD, which uses a Bregman divergence as a regularizer, is known as \emph{online mirror descent}\index{mirror descent} (OMD) or online proximal gradient descent. Note that OMD is parametrized by $h$, which is used to describe the underlying geometry of the problem~\cite{Chen2018}.

For the function $h(x) = \frac{1}{2} \norm{x}_2^2$ from $\mathbb{R}^d$ to $\mathbb{R}$ (the \emph{squared $\ell_2$ norm}\index{squared $\ell_2$ norm}), the Bregman divergence is the Euclidean distance, i.e. $D_h(x,y) = \frac{1}{2} \norm{x-y}_2^2$~\cite{Chen2018}. Hence, OMD reduces to OGD if it is parametrized with the provided definition of $h$, i.e., a Euclidean geometry is used.

\paragraph{Mirror Map View} The proximal point view yields just one perspective of mirror descent. Another perspective that is used frequently is the perspective of mirror maps.

Recall that in OGD, starting from some point $X_{\tau-1}$, we moved into the direction of the gradient of $f$. However, note that $\nabla f_{\tau-1}(X_{\tau-1})$ belongs to the dual space\footnote{The \emph{dual space}\index{dual space} of a vector space $V$ over some field $\mathbb{F}$ is the set of all linear maps from vectors in $V$ to scalars in $\mathbb{F}$ (which are called \emph{linear functionals}\index{linear functionals})~\cite{Wadsley2015}} of $\mathbb{R}^d$. In the Euclidean space, this is not a problem as the dual space of the Euclidean space is the Euclidean space itself~\cite{Gupta2020}. However, when working with normed spaces that are not self-dual, this is problematic.

Instead of adding elements from the dual space to elements from the primal space, mirror descent maps points from the primal space to the dual space, performs the gradient step in the dual space, and then maps the resulting point back to the primal space.

\begin{figure}
    \centering
    \input{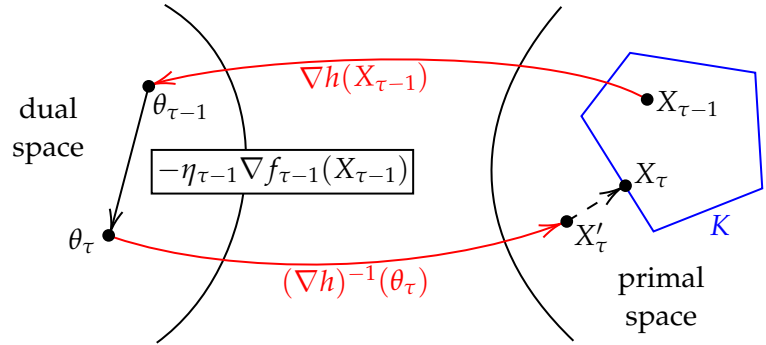}
    \caption{Visualization of a step of Mirror Descent. The previous point $X_{\tau-1}$ is first mapped to the dual space, $\theta_{\tau-1}$. Then, a step is taken into the direction of the gradient, $\theta_{\tau}$, and the resulting point is mapped back to the primal space. Finally, the resulting point $X'_{\tau}$ is projected back onto the feasible region $K$, resulting in the next point $X_{\tau}$ \cite{Gupta2020}.}
    \label{fig:mirror_descent}
\end{figure}

\begin{definition}\index{mirror map}
\cite{Gupta2020} Given a norm $\norm{\cdot}$ and a differentiable and $\alpha$-strongly convex function $h : \mathbb{R}^d \to \mathbb{R}$, the associated mirror map is $\nabla h : \mathbb{R}^d \to \mathbb{R}^d$ and the inverse mirror map is $(\nabla h)^{-1} : \mathbb{R}^d \to \mathbb{R}^d$.
\end{definition}

For $h(x) = \frac{1}{2} \norm{x}_2^2$ the mirror map and its inverse are the identity map~\cite{Gupta2020}. A complete description of the mirror descent framework is given in \cref{alg:md:omd}. Note that the choice of the mirror map is central as it describes the dual space (also called mirror image) where the gradient step is taken. \Cref{fig:mirror_descent} visualizes a step of mirror descent.

\begin{algorithm}
    \caption{Online Mirror Descent~\cite{Gupta2020}}\label{alg:md:omd}
    \SetKwInOut{Input}{Input}
    \Input{$\tau \in \mathbb{N}, K \subset \mathbb{R}^d, \norm{\cdot}, h \in \mathbb{R}^d \to \mathbb{R}$}
    map to the dual space $\theta_{\tau-1} \gets \nabla h (X_{\tau-1})$\;
    take a gradient step in the dual space $\theta_{\tau} \gets \theta_{\tau-1} - \eta_{\tau-1} \nabla f_{\tau-1}(X_{\tau-1})$\;
    map back to the primal space $X'_{\tau} \gets (\nabla h)^{-1}(\theta_{\tau})$\;
    project $X'_{\tau}$ onto a point $X_{\tau} \in K$ using the Bregman projection $X_{\tau} \gets \Pi_K^h(X'_{\tau})$\;
    \Return $X_{\tau}$\;
\end{algorithm}

\begin{definition}\index{Bregman projection}
\cite{Gupta2020} The Bregman projection of a point $x$ onto a convex set $K \subseteq \mathbb{R}^d$ given the distance-generating function $h$ is \begin{align*}
    \Pi_K^h(x) = \argmin_{y \in K} D_h(y,x).
\end{align*}
\end{definition}

Note that when $h$ is the squared $\ell_2$ norm, the Bregman projection is equivalent to the Euclidean projection. Using finite difference methods, the Bregman projection can be computed similarly to the Euclidean projection, $\epsilon$-optimally in $\mathcal{O}(O_{\epsilon}^d)$ time, assuming the runtime of $h$ is constant.

Next, we describe how the ideas from mirror descent are adapted for the smoothed online convex optimization setting, where agents operate with lookahead $1$ and movement costs need to be considered.

\subsubsection{Online Balanced Descent}\label{section:online_algorithms:md:descent_methods:obd}

The online balanced descent (OBD) algorithms that were developed by \citeauthor{Chen2018}~\cite{Chen2018} are a special case of online mirror descent (OMD) with lookahead $1$. In practice, to achieve the one-step lookahead, OBD moves to a point $X_{\tau}$ on a level set of $f_{\tau}(\cdot)$ such that the step is normal to the contour line of $f_{\tau-1}$. In contrast, OMD steps into a direction that is normal to the contour line of $f_{\tau-1}(X_{\tau-1})$. In other words, OMD takes a step with respect to its starting point on some level set, whereas OBD takes a step with respect to its destination on some level set. \Cref{fig:comparison_of_an_update_of_omd_and_obd} shows how these updates compare.

\begin{figure}
    \begin{subfigure}[b]{\textwidth}
    \centering
    \input{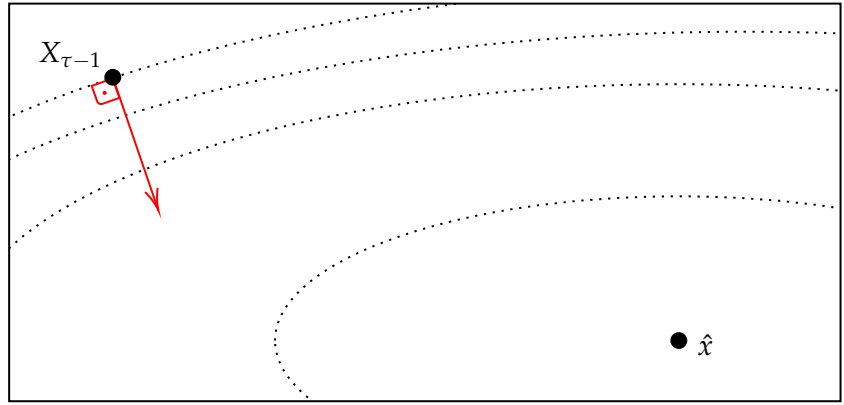}
    \caption{Update of OMD. The contour lines represent level sets of $f_{\tau-1}$.}
    \end{subfigure}
    \par\bigskip
    \begin{subfigure}[b]{\textwidth}
    \centering
    \input{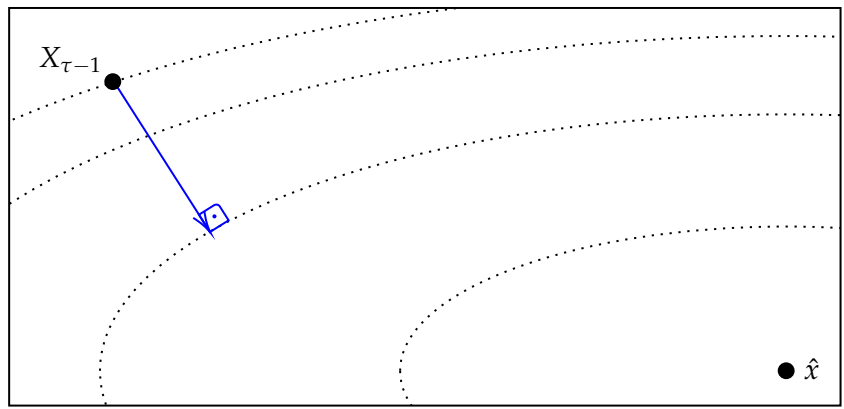}
    \caption{Update of OBD. The contour lines represent level sets of $f_{\tau}$.}
    \end{subfigure}
    \caption{Comparison of an update of OMD and OBD in two dimensions assuming the distance-generating function $h(x) = \frac{1}{2} \norm{x}_2^2$. OMD (red) takes a step in a direction normal to the contour line of $f_{\tau-1}$ at $X_{\tau-1}$. OBD (blue) takes a step in a direction normal to the contour line of $f_{\tau}$ at $X_{\tau}$ \cite{Chen2018}. The step in the direction of the minimizer of $f_{\tau}$ is shown in black. Note that it is not optimal to move in the direction of the minimizer to a point on some level set as there likely exists a closer point on the same level set.}
    \label{fig:comparison_of_an_update_of_omd_and_obd}
\end{figure}

The algorithmic framework for OBD can roughly be divided into two parts. First, the projection of the previous point onto some level set of the cost function.  Second, the strategies to choose the specific level set and geometry to balance hitting costs and movement costs. The algorithm for (1) is also called the \emph{meta} algorithm as it is parametrized with a concrete level set and geometry (i.e., mirror map).

\subsubsection{Meta Algorithm}

Similar to OMD, the meta algorithm of OBD chooses the next point $X_{\tau}$ in the dual space. However, whereas OMD takes an arbitrary step into the direction of the gradient, OBD takes the shortest step onto some sub-level set $K_l = \{x \in \mathcal{X} \mid f_{\tau}(x) \leq l\}$ of the revealed hitting cost $f_{\tau}$. In other words, we seek to find the Bregman projection of $X_{\tau-1}$ onto the sub-level set $K_l$. The first-order condition of the corresponding optimization in the dual space implies that \begin{align}\label{eq:pbd:first_order_condition}
    \nabla h(X_{\tau}) = \nabla h(X_{\tau-1}) - \eta_{\tau} \nabla f_{\tau}(X_{\tau})
\end{align} must be satisfied by $X_{\tau}$ where $\eta_{\tau}$ is the optimal slack of the inequality constraint $f_{\tau}(x) \leq l$~\cite{Chen2018}. Note that this corresponds to a variant of OMD with lookahead $1$.

OBD requires $h$ to be $\alpha$-strongly convex and $\beta$-Lipschitz smooth in the norm $\norm{\cdot}$ that is used to obtain the movement costs. The meta algorithm is described in \cref{alg:md:obd}.

\begin{algorithm}
    \caption{Online Balanced Descent (meta algorithm)~\cite{Chen2018}}\label{alg:md:obd}
    \SetKwInOut{Input}{Input}
    \Input{$\mathcal{I}_{\text{SCO}} = (\tau \in \mathbb{N}, \mathcal{X} \subset \mathbb{R}^d, \norm{\cdot}, (f_1, \dots, f_{\tau}) \in (\mathcal{X} \to \mathbb{R}_{\geq 0})^{\tau}), l \geq 0, \text{distance-generating function } h$}
    $X_{\tau} \gets \Pi_{K_l}^h(X_{\tau-1})$\;
    \Return $X_{\tau}$\;
\end{algorithm}

As an example, we consider the Euclidean space with the $\ell_2$ norm. In this setting $h(x) = \frac{1}{2} \norm{x}_2^2$ is 1-strongly convex and 1-Lipschitz smooth~\cite{Chen2018}. As the corresponding mirror map, $\nabla h$ is the identity map, the first-order condition \cref{eq:pbd:first_order_condition} reduces to \begin{align*}
    X_{\tau} = X_{\tau-1} - \eta_{\tau} \nabla f_{\tau}(X_{\tau})
\end{align*} corresponding to OGD with lookahead $1$.

We can generally choose $h$ to either perform well for the competitive ratio or regret. In their initial paper, \citeauthor{Chen2018}~\cite{Chen2018} propose two algorithms that balance hitting and movement costs in the primal and dual space and perform well concerning the competitive ratio and regret, respectively.

\subsubsection{Primal Algorithm}

\emph{Primal online balanced descent} (P-OBD) balances hitting and movement costs in the primal space. Let $\hat{x} = \argmin_{x \in \mathcal{X}} f_{\tau}(x)$ and $x(l) = \text{Meta-OBD}(\mathcal{I}, l, h) = \Pi_{K_l}^h(X_{\tau-1})$. Given some $\beta > 0$, the balance parameter $l$ is chosen such that a balance condition $g(l) = \norm{x(l) - X_{\tau-1}} \leq \beta l$ is satisfied. More formally, $l$ is chosen such that either $x(l) = \hat{x}$ and $g(l) < \beta l$ or $g(l) = \beta l$ hold~\cite{Chen2018}.

\citeauthor{Chen2018}~\cite{Chen2018} show that the balance function $g(l)$ is continuous in $l$. We observe that $l$ is lower bounded by $f_{\tau}(\hat{x})$. As we assume that $\hat{x}$ is unique, $x(l) = \hat{x}$ iff $l = f_{\tau}(\hat{x})$ and the first condition is fulfilled if and only if $g(f_{\tau}(\hat{x})) < \beta f_{\tau}(\hat{x})$. If the first condition is not satisfied, we can efficiently determine an $l$ fulfilling the second condition using a bracketed root finding method on $g(l) - \beta l$ within the interval $[f_{\tau}(\hat{x}), \gamma]$. Here, it suffices to choose $\gamma$ ``large enough'' to ensure that $g(\gamma) \leq \beta \gamma$. Observe that for $\gamma = f_{\tau}(X_{\tau-1})$, this is trivially satisfied, as $x(\gamma) = X_{\tau-1}$ and as such $g(\gamma) = 0$. The algorithm is described in \cref{alg:md:pobd}.

\begin{algorithm}
    \caption{Primal Online Balanced Descent~\cite{Chen2018}}\label{alg:md:pobd}
    \SetKwInOut{Input}{Input}
    \Input{$\mathcal{I}_{\text{SCO}} = (\tau \in \mathbb{N}, \mathcal{X} \subset \mathbb{R}^d, \norm{\cdot}, (f_1, \dots, f_{\tau}) \in (\mathcal{X} \to \mathbb{R}_{\geq 0})^{\tau}), \beta > 0, \text{distance-generating function } h$}
    $\hat{x} = \argmin_{x \in \mathcal{X}} f_{\tau}(x)$\;
    \If{$g(\hat{x}) \leq \beta f_{\tau}(\hat{x})$}{
        \Return $\hat{x}$\;
    }
    $l \gets $ root of $g(l') - \beta l'$ for $l' \in [f_{\tau}(\hat{x}), f_{\tau}(X_{\tau-1})]$\;
    \Return $\text{Meta-OBD}(\mathcal{I}, l, h)$\;
\end{algorithm}

The balancing is chosen such that the movement cost is upper bounded by the constant $\beta$ times the hitting cost~\cite{Chen2018}.

\citeauthor{Chen2018}~\cite{Chen2018} show that P-OBD attains a competitive ratio of at most $3 + \mathcal{O}(1 / \alpha)$ for some $\beta > 0$, $\alpha$-locally polyhedral cost functions, and the $\ell_2$ norm as movement cost. They also show that local polyhedrality is useful when other norms like the $l_{\infty}$ norm are used. Note that the algorithm is memoryless and, therefore, nearly optimal, as \citeauthor{Bansal2015}~\cite{Bansal2015} showed in the uni-dimensional setting that no memoryless algorithm can attain a better competitive ratio than $3$ (which also holds for locally polyhedral cost functions)~\cite{Chen2018}. When used with the $\ell_1$ norm, which is relevant in the application of right-sizing data centers, P-OBD attains a competitive ratio of $\mathcal{O}(\sqrt{d})$ if $\alpha$ is fixed~\cite{Chen2018}. The memoryless algorithm (\cref{alg:ud:memoryless}) of \citeauthor{Bansal2015}~\cite{Bansal2015} can be seen as a special case of P-OBD for $d = 1$ and $\beta = \frac{1}{2}$~\cite{Chen2018}.

The minimizer can be found $\epsilon$-optimally in $\mathcal{O}(C O_{\epsilon}^d)$ time. Assuming the runtime of $h$ is constant, the balance parameter $l$ can be determined $\epsilon$-optimally in $\mathcal{O}(O_{\epsilon}^d R_{\epsilon})$ time. Overall, P-OBD runs in $\mathcal{O}(C O_{\epsilon}^d + O_{\epsilon}^d R_{\epsilon})$ time.

\subsubsection{Dual Algorithm}

\citeauthor{Chen2018}~\cite{Chen2018} also developed an algorithm called \emph{dual online balanced descent} (D-OBD), which balances the movement cost in the dual space with the gradient of the hitting cost (which is also in the dual space). Before describing the algorithm, we must first describe how the movement cost can be represented in the dual space. We thus introduce the notion of a \emph{dual norm}.

\begin{definition}\index{dual norm}
\cite{Gupta2020} Given some norm $\norm{\cdot}$ on $\mathbb{R}^d$, its dual norm $\norm{\cdot}_*$ is defined as \begin{align*}
    \norm{y}_* = \sup_{x \in \mathbb{R}^d} \{\langle x, y \rangle \mid \norm{x} \leq 1\}.
\end{align*}
\end{definition}

The $\ell_2$ norm is self-dual~\cite{Gupta2020}. In general, the dual norm for a concrete point $y \in \mathbb{R}^d$ can be computed with the following convex optimization: \begin{align*}
    &\max_{x \in \mathbb{R}^d} &&\langle x, y \rangle \\
    &\text{subject to}         &&\norm{x} \leq 1.
\end{align*}

Returning to the description of D-OBD, for some fixed learning rate $\eta$, $l$ is now chosen such that \begin{align*}
    \norm{\nabla h(x(l)) - \nabla h(X_{\tau-1})}_* = \eta \norm{\nabla f_{\tau}(x(l))}_*
\end{align*} holds. Let $g_1(l) = \norm{\nabla h(x(l)) - \nabla h(X_{\tau-1})}_*$ and $g_2(l) = \norm{\nabla f_{\tau}(x(l))}_*$. Again, \citeauthor{Chen2018}~\cite{Chen2018} show that the balance function $\frac{g_1(l)}{g_2(l)}$ is continuous in $l$ under the assumption that $h$ and $f_{\tau}$ are continuously differentiable on $\mathcal{X}$. Similar to our analysis of P-OBD, we observe  that $l$ is lower bounded by $f_{\tau}(\hat{x})$. We can determine $l$ using a bracketed root finding method on $g_1(l) - \eta g_2(l)$ within the interval $[f_{\tau}(\hat{x}), \gamma]$. We observe that for $l = f_{\tau}(\hat{x})$, $g_1(l) \geq \eta g_2(l) = 0$. Therefore, we need to choose $\gamma$ such that $g_1(l) \leq \eta g_2(l)$ is satisfied. Similar to our argument for P-OBD, it suffices to choose $\gamma = f_{\tau}(X_{\tau-1})$, resulting in $x(\gamma) = X_{\tau-1}$, and implying $g_1(\gamma) = 0$. The resulting algorithm is described in \cref{alg:md:dobd}.

\begin{algorithm}
    \caption{Dual Online Balanced Descent~\cite{Chen2018}}\label{alg:md:dobd}
    \SetKwInOut{Input}{Input}
    \Input{$\mathcal{I}_{\text{SCO}} = (\tau \in \mathbb{N}, \mathcal{X} \subset \mathbb{R}^d, \norm{\cdot}, (f_1, \dots, f_{\tau}) \in (\mathcal{X} \to \mathbb{R}_{\geq 0})^{\tau}), \eta > 0, \text{distance-generating function } h$}
    $\hat{x} = \argmin_{x \in \mathcal{X}} f_{\tau}(x)$\;
    $l \gets $ root of $g_1(l') - \eta g_2(l')$ for $l' \in [f_{\tau}(\hat{x}), f_{\tau}(X_{\tau-1})]$\;
    \Return $\text{Meta-OBD}(\mathcal{I}, l, h)$\;
\end{algorithm}

\citeauthor{Chen2018}~\cite{Chen2018} show that the $L$-constrained dynamic regret of D-OBD is upper bounded by $\frac{G L}{\eta} + \frac{T \eta}{2 \alpha}$ where $h$ is $\alpha$-strongly convex in $\norm{\cdot}$, $\norm{\nabla h(x)}_*$ is upper bounded by $G$, and $\nabla h(0) = 0$. When $G$, $L$, and $T$ are known, $\eta$ can be chosen optimally as $\eta = \sqrt{\frac{2 G L \alpha}{T}}$, resulting in an $L$-constrained dynamic regret that is upper bounded by $\sqrt{\frac{2 G L T}{\alpha}}$~\cite{Chen2018}. Further, in this setting, D-OBD achieves static regret $\mathcal{O}(\sqrt{T})$~\cite{Chen2018}.

An evaluation of the dual norm can be computed $\epsilon$-optimally in $\mathcal{O}(O_{\epsilon}^d)$ time. Thus, assuming the runtime of $h$ is constant, $l$ can be found in $\mathcal{O}((O_{\epsilon}^d)^2 R_{\epsilon})$ time. Overall, the asymptotic time complexity of D-OBD is given as $\mathcal{O}(C O_{\epsilon}^d + (O_{\epsilon}^d)^2 R_{\epsilon})$.

\subsubsection{Greedy and Regularized Algorithms}

Later, \citeauthor{Goel2019}~\cite{Goel2019} proposed two additional algorithms using the OBD framework, \emph{greedy online balanced descent} (G-OBD) and \emph{regularized online balanced descent} (R-OBD). Both algorithms yield strong guarantees for the competitive ratio in the setting of squared $\ell_2$ norm movement costs and $\alpha$-strongly convex hitting costs where the optimal competitive ratio is $\mathcal{O}(1 / \sqrt{\alpha})$ as $\alpha$ approaches zero. G-OBD achieves this competitive ratio for quasiconvex\footnote{A function is quasiconvex iff it has a unique global minimum.} hitting costs that are $\alpha$-strongly convex around their minimizer and squared $\ell_2$ norm movement costs. R-OBD achieves this competitive ratio for $\alpha$-strongly convex hitting costs and arbitrary Bregman divergences as movement costs.

Both algorithms take an additional step of size $\mathcal{O}(\sqrt{\alpha})$ towards the minimizer of the hitting cost. G-OBD works by first taking a regular P-OBD step to some level set of the hitting cost. Then, it takes an additional step towards the minimizer of the hitting cost with a step size based on the convexity parameter $\alpha$. In contrast, R-OBD picks the next point by minimizing a weighted sum of hitting and movement costs. It uses an additional regularization term that encourages the algorithm to pick a point closer to the minimizer of the hitting cost~\cite{Goel2019}. We do not discuss G-OBD and R-OBD in more detail, as their theoretical guarantees do not cover the application of right-sizing data centers, but we provide implementations of them.

\section{Predicting}\label{section:online_algorithms:md:predictions}

In practice, we can attempt to use predicted hitting costs to improve the performance of online algorithms. Predicting future incoming loads to a high degree of accuracy in the data-center setting is often possible. Using predicted hitting costs and their uncertainty distributions, online algorithms can make more informed decisions in practice. In this section, we begin by discussing prediction windows. Then, we describe approaches for time-series predictions and end with discussing algorithms that use such predictions.

\subsection{Prediction Window}

A natural model to allow incorporating predictions is the use of a finite prediction window $w$. A prediction window bridges the gap between offline and online algorithms. Whereas an online algorithm only knows the hitting costs $f_t$ for $t \in [\tau]$ and an offline algorithm knows the hitting costs $f_t$ for all $t \in [T]$, an online algorithm with \emph{prediction window}\index{prediction window} of length $w$ knows all hitting costs $f_t$ up to $\tau + w$, i.e. $t \in [\tau + w]$. In other words, the prediction window $w$ represents the number of upcoming time slots at which the algorithm is assumed to have perfect knowledge of the future.

\subsubsection{Lazy Capacity Provisioning with Prediction Window}

\citeauthor{Lin2011}~\cite{Lin2011} extend their algorithm lazy capacity provisioning, which we discussed in \cref{section:online_algorithms:ud:lazy_capacity_provisioning} to support the prediction window by changing the update rule to \begin{align*}
    X_{\tau} = \begin{cases} 
        0 & \tau \leq 0 \\
        (X_{\tau-1})_{X_{\tau+w,\tau}^L}^{X_{\tau+w,\tau}^U} & \tau \geq 1
    \end{cases}
\end{align*}

The optimal schedules now need to obtained for $\tau + w$ rather than $\tau$ time slots. Thus, the time complexity changes to $\mathcal{O}((\tau + w) C O_{\epsilon}^{\tau + w})$ and $\mathcal{O}((\tau + w)^2 C \log_2 m)$ in the fractional and integral case, respectively.

The assumption of perfect knowledge of the future is sure to be violated when an online algorithm is used in practice. Still, \citeauthor{Lin2011}~\cite{Lin2011} show that lazy capacity provisioning with a prediction window is robust to this assumption in practice. \citeauthor{Lin2011}~\cite{Lin2011} and \citeauthor{Albers2018}~\cite{Albers2018} showed that using a finite prediction window does not improve the worst-case performance of the online algorithm for the fractional and integral case, respectively. In other words, the competitive ratio of lazy capacity provisioning is $3$ regardless of whether it uses a finite prediction window. In practice, however, \citeauthor{Lin2011}~\cite{Lin2011} show that a prediction window significantly improves the algorithm's performance.

There are two main drawbacks to using a finite prediction window. First, predictions windows are finite and typically constrained to a short period as they are assumed to be perfect. In contrast, predictions can be made for much longer time horizons, albeit with decreasing accuracy. Second, it completely disregards any knowledge or assumptions of the certainty and noise of the predictions by assuming the predictions to be perfect.

\subsection{Making Predictions}\label{section:online_algorithms:md:predictions:making_predictions}

There exist multiple paradigms for making time-series predictions. Due to much recent engagement in the field of deep learning generally and time-series predictions specifically, multiple approaches perform well in practical settings. Most algorithms separately tune parameters of individual models for short-term and long-term trends as well as seasonality~\cite{Taylor2017, Hosseini2021}.

A fundamental difference between models is Bayesianness, i.e., whether they use an underlying uncertainty distribution within the model. Facebook's Prophet algorithm is Bayesian, whereas LinkedIn's Greykite algorithm is not~\cite{Taylor2017, Hosseini2021}.

For Bayesian models, online algorithms can use the uncertainty distribution to consider outliers appropriately. For non-Bayesian models, additive white Gaussian noise can be added to the prediction to achieve a similar effect. In general, many strategies can be used to obtain a single representative prediction of the underlying distribution. In our experiments, we use the mean prediction to ensure appropriate consideration of outliers. The median or 90th percentile predictions are alternatives that are more robust to outliers.

Note that, in principle, predictions can be made arbitrarily far into the future. However, at some point, they become too uncertain to be valuable. For example, infeasible load profiles may be assigned a positive probability, which would result an infinite cost if we use the mean to obtain a representative prediction, even if all servers are active. Thus, we also use a prediction window, which needs to be set appropriately to account for the uncertainty distribution of the predicted loads.

\subsection{Receding Horizon Control}

\emph{Receding horizon control}\index{receding horizon control} (RHC) (or \emph{model predictive control}) is a methodology for making decisions based on predictions of the future that is commonly used to control data centers~\cite{Lin2012}. In RHC, an agent predicts their action up to some fixed point in time, referred to as the prediction window. Based on this prediction, the agent adjusts their action for the current time slot. In the next time slot, this process repeats~\cite{Zak2017}.

\citeauthor{Lin2012}~\cite{Lin2012} previously investigated the performance of RHC in the context of right-sizing data centers. RHC works by solving a convex optimization from time $\tau$ to time $\tau + w$ starting from the initial configuration $X_{\tau-1}$. We set $X_0 = \mathbf{0}$. Similar to our analysis of capacity provisioning in \cref{section:offline_algorithms:ud:capacity_provisioning}, we describe by $X^{\tau}(X_{\tau-1}) \in \mathcal{X}^{w+1}$ the optimal schedule for times $\tau$ through $\tau+w$. This schedule is obtained by minimizing \begin{align}\label{eq:rhc}
    \sum_{t=\tau}^{\tau+w} f_t(X_t) + \norm{X_t - X_{t-1}}
\end{align} over configurations $X_{\tau}, \dots, X_{\tau+w} \in \mathcal{X}$. This optimization has $\mathcal{O}(d w)$ dimensions and thus can be computed $\epsilon$-optimally in $\mathcal{O}(C O_{\epsilon}^{dw})$ time. Now, RHC simply picks the first predicted action. RHC is described in \cref{alg:predictions:rhc}.

\begin{algorithm}
    \caption{Receding Horizon Control~\cite{Lin2012}}\label{alg:predictions:rhc}
    \SetKwInOut{Input}{Input}
    \Input{$\mathcal{I}_{\text{SSCO}} = (\tau \in \mathbb{N}, m \in \mathbb{N}, \beta \in \mathbb{R}_{>0}, (f_1, \dots, f_{\tau}) \in (\mathbb{R}_{\geq 0} \to \mathbb{R}_{\geq 0})^{\tau})$}
    $X_{\tau} = X_{\tau}^{\tau}(X_{\tau-1})$\;
    \Return $X_{\tau}$\;
\end{algorithm}

\citeauthor{Lin2012}~\cite{Lin2012} prove the competitive ratio of RHC in the application of right-sizing data centers. They show that in the uni-dimensional setting, RHC attains a competitive ratio of $1 + \mathcal{O}(1/w)$ which is strictly better than the optimal competitive ratio (for deterministic algorithms without predictions) of $2$ and $3$ for memoryless algorithms for $w > 1$ and $w > \frac{1}{2}$, respectively. However, in a multi-dimensional setting, RHC is $(1 + \max_{k \in [d]} \beta_k / e_k(0))$-competitive where we defined $\beta_k$ as the switching cost of a server of type $k$ and $e_k(0)$ as the average energy cost of an idling server of type $k$. Importantly, this competitive ratio does not depend on the size of the prediction window $w$.

\subsection{Averaging Fixed Horizon Control}

In their paper, \citeauthor{Lin2012}~\cite{Lin2012} present another algorithm, \emph{averaging fixed horizon control} (AFHC), which attains a competitive ratio of $1 + \max_{k \in [d]} \frac{\beta_k}{(w+1) e_k(0)}$. In particular, AFHC is $(1 + \mathcal{O}(1/w))$-competitive. However, \citeauthor{Lin2012}~\cite{Lin2012} find that in many realistic settings, RHC performs better than AFHC.

At time $\tau$, AFHC works by performing $w + 1$ individual RHC steps starting from $t_0 = \tau-w$ up to $t_0 = \tau$ and averaging the results. Each individual step is also referred to as an iteration of \emph{fixed horizon control} (FHC).

We describe the sub-iterations of AFHC using $k \in [w+1]$. We set $t_0 = \tau+k-(w+1)$, ensuring that $t_0 \in [\tau-w,\tau]$. We denote by $X^{t_0}(X_{t_0-1}^{(k)})$ the optimal schedule for times $t_0$ through $t_0+w$ which is obtained analogously to \cref{eq:rhc}. We also set $X_t = \mathbf{0}$ and $X_t^{(k)} = \mathbf{0}$ for all $t \leq 0$ and $k \in [w+1]$. AFHC is described in \cref{alg:predictions:afhc}.

\begin{algorithm}
    \caption{Averaging Fixed Horizon Control~\cite{Lin2012}}\label{alg:predictions:afhc}
    \SetKwInOut{Input}{Input}
    \Input{$\mathcal{I}_{\text{SSCO}} = (\tau \in \mathbb{N}, m \in \mathbb{N}, \beta \in \mathbb{R}_{>0}, (f_1, \dots, f_{\tau}) \in (\mathbb{R}_{\geq 0} \to \mathbb{R}_{\geq 0})^{\tau})$}
    \ForEach{$k \in [w+1]$}{
        $t_0 \gets \tau+k-(w+1)$\;
        $X^{(k)} \gets X^{t_0}(X_{t_0-1}^{(k)})$\;
    }
    $X_{\tau} = \frac{1}{w+1} \sum_{k=1}^{w+1} X_{\tau}^{(k)}$\;
    \Return $X_{\tau}$\;
\end{algorithm}

Intuitively, AFHC can be interpreted as performing $w+1$ FHC-steps in parallel, where each FHC-step starts from a different $t_0 \in [\tau-w,\tau]$, and then averaging all configurations for time $\tau$. Note that RHC is equivalent to the last FHC-step with initial time $t_0 = \tau$, i.e., $k = w+1$. The asymptotic time complexity of AFHC is given as $\mathcal{O}(w C O_{\epsilon}^{dw})$.

\citeauthor{Chen2015}~\cite{Chen2015} show that AFHC achieves sublinear regret and a constant competitive ratio using a prediction window of constant length. \citeauthor{Badiei2015}~\cite{Badiei2015} introduce a class of ``forward-looking'' algorithms that can consider cost functions within some prediction window but are only allowed to use a constant limited number of past cost functions. They show that among these algorithms, AFHC achieves optimal regret.

In~\cite{Chen2016}, \citeauthor{Chen2016} generalize RHC and AFHC to a class of algorithms called \emph{committed horizon control} (CHC), which consist of $v \in [w+1]$ sub-iterations of FHC. Note that RHC corresponds to CHC with parameter $v = 1$, whereas AFHC corresponds to CHC with parameter $v = w+1$. They investigate how $v$ can be chosen optimally based on the noise distribution of predictions.

\citeauthor{Lin2019}~\cite{Lin2019} extend AFHC to a new algorithm called \emph{synchronized fixed horizon control} (SFHC), which is $(1 + \mathcal{O}(1/w))$-competitive for both convex and non-convex cost functions. \citeauthor{Li2018}~\cite{Li2018} propose two new gradient-based online algorithms, \emph{receding horizon gradient descent} (RHGD) and \emph{receding horizon accelerated gradient} (RHAG), and show that the dynamic regret of RHAG is near-optimal when compared to a class of online algorithms that includes CHC.


\chapter{Implementation}\label{chapter:implementation}

As part of our work, we implemented the algorithms described in \cref{chapter:offline_algorithms} and \cref{chapter:online_algorithms}. Our implementation is written in Rust and has Python bindings to interact with some components~\cite{Huebotter2021_2, Huebotter2021_3}. Detailed documentation of the implementation is available~\cite{Huebotter2021_4}. In this chapter, we discuss the general architecture and the points of focus of our implementation.

\section{Architecture}\label{section:implementation:architecture}

Our implementation is separated into four main components. First, a component that encompasses the implementation of the general problems described in \cref{chapter:theory}. Second, a related component containing the implementations of the offline and online algorithms from \cref{chapter:offline_algorithms} and \cref{chapter:online_algorithms}, respectively. Third, a component including abstractions to generate and update problem instances over time (referred to as \emph{models}\index{models}). This includes the models of data centers described in \cref{chapter:application}. Fourth, a component that implements the algorithms' practical use. Among other things, this component is used to execute online algorithms in real-time by sequentially updating the underlying problem instance as new information becomes available (referred to as \emph{streaming}\index{streaming} an online algorithm).

\subsection{Problems and Reductions}

This component includes data structures representing instances of the problems introduced in \cref{chapter:theory}. Note-worthy aspects are the definition of norms for SCO and hitting costs for SCO, SSCO, and SBLO. We also implemented the reductions from SLO to SBLO, SBLO to SSCO, and SSCO to SCO.

\paragraph{Norms} Our implementation includes the Manhattan norm, Euclidean norm, Mahalanobis norm, and a dimension-dependently scaled variant of the Manhattan norm, which is used in the reduction from SSCO to SCO. We also include implementations of the square of a norm as well as the dual of a norm.

\paragraph{Hitting Costs} The primary data structure used across SCO, SSCO, and SBLO is the definition of hitting costs. Here, our implementation has to allow for sequentially arriving subsets of the hitting cost. Consider the example of an online algorithm with a prediction window $w$. Then, at each time slot $\tau$, the algorithm expects to receive a hitting cost with domain $[\tau : \tau + w]$ which conflicts with previous and subsequent hitting costs. To complicate the matter further, offline algorithms only require a single hitting cost covering the entire domain $[T]$.

In our implementation, we refer to these hitting costs, which may only cover definitions of the hitting cost for a subset of all time slots as \emph{single hitting costs}\index{single hitting cost}. In the online setting, the domain of a single hitting cost may cover future time slots, in which case the hitting cost is uncertain. From our discussion of predictions in \cref{section:online_algorithms:md:predictions:making_predictions}, it naturally follows to model the distribution of hitting costs for a future time slot as a vector of samples that can then be used to estimate the density. Thus, a single cost function arriving at time slot $\tau$ is formally described as a function $\chi_{\tau} : [T] \times \mathcal{S} \to \bigcup_{n=1}^{\infty} \mathbb{R}_{\geq 0}^n$ mapping a time slot $t \in [T]$ and some point $x$ in the space $\mathcal{S}$ to a prediction with a varying sample size $n$. In the case of a certain prediction, the sample size is $1$. For SCO and SSCO we set $\mathcal{S} = \mathcal{X}$ while for SBLO we set $\mathcal{S} = \mathbb{R}$.

We implement the hitting cost $f_t(x)$ using a B-Tree-Map\footnote{A B-Tree-Map is a map based on a B-Tree balancing cache-efficiency and search performance~\cite{BTreeMap}} mapping the time slot of their arrival to single hitting costs. Crucially, there need not be a single hitting cost for every time slot, as the example of an offline algorithm illustrates. The concrete single hitting cost is determined by choosing the last single hitting cost arriving during a time slot in $[t]$. In an online setting, this ensures that if $t$ is in the future, the current single hitting cost is used, and if $t$ is in the past, the single hitting cost from time slot $t$ (or the closest previous single hitting cost) is used.

\subsection{Algorithms}

The algorithms component encompasses the implementation of all offline and online algorithms described in \cref{chapter:offline_algorithms} and \cref{chapter:online_algorithms}. For both classes of algorithms, we define a common interface.

\paragraph{Offline} An offline algorithm receives as input a problem instance and some algorithm-specific options. It returns an arbitrary data structure that can be used to obtain the determined schedule. In our implementation, we generalized the algorithms slightly so as to support the following uses:

First, we adapted the graph-based algorithms for SSCO to support inverted movement costs, i.e., paying the switching cost for shutting down a server rather than powering up a server. Second, we adapted all algorithms to allow for computing the $\alpha$-unfair optimal offline solution where movement costs are scaled by a factor $\alpha$. Third, we adapted the general algorithm for the multi-dimensional fractional case to support computing the $L$-constrained optimal offline solution where movement costs are upper-bounded by the constant $L$. The $L$-constrained optimal offline solution cannot easily be computed using the graph-based approaches for the integral case as this additional constraint prevents the application of Bellman's optimality principle to find a shortest path using dynamic programming. Fourth, we extended the implementations to allow starting from an arbitrary initial time slot. Fifth, we adapted the uni-dimensional optimal graph search (\cref{alg:ud:optimal_graph_search}) to support arbitrary initial configurations in the first time slot.

Some of these extensions are required to implement some online algorithms, and others were merely added for analysis purposes. We also provide functions to compute the static fractional and integral optima.

\paragraph{Online} An online algorithm receives as input a problem instance, the current time slot $\tau$, the schedule for all previous time slots, some algorithm-specific memory, and some algorithm-specific options. The algorithm returns the configuration for time slot $\tau$ as well as the updated memory. The practical use of online algorithms is described in greater detail in \cref{section:implementation:architecture:streaming}.

\subsection{Models}

In practice, it is not very useful to generate and update the problem instances directly. In the application of right-sizing data centers, we have discussed models of operating costs and switching costs in detail in \cref{chapter:application}. These models can be used to generate instances of the problems discussed in \cref{chapter:theory}. This approach is not limited to the application of right-sizing data centers. A model is a data structure with associated functions to produce an associated problem instance and update an existing problem instance online. As seen in the example of right-sizing data centers, these generators may require additional inputs, which are provided as the load profiles for each time slot. We refer to the inputs required for the initial generation of a problem instance as \emph{offline inputs}\index{offline inputs} and the inputs required to update a problem instance online as \emph{online inputs}\index{online inputs}. Offline inputs encapsulate information for all time slots with respect to some time horizon $[T]$ whereas online inputs for some time slot $\tau$ encapsulate information for the current time slot $\tau$ as well as all time slots in the prediction window $[\tau + 1 : \tau + w]$.

In the application of right-sizing data centers, an offline input is a vector of load profiles $\mathcal{I} = (\lambda_1, \dots, \lambda_T)$ with $\lambda_t \in \mathbb{N}_0^e$ for all $t \in [T]$. In contrast, an online input is a vector of predicted load profiles $\mathcal{I} = (\lambda_{\tau}, \mathcal{P}_{\tau + 1}, \dots, \mathcal{P}_{\tau + w})$ with $\lambda_{\tau} \in \mathbb{N}_0^e$ and where $\mathcal{P}_t \in \left(\bigcup_{n=1}^{\infty} \mathbb{N}_0^n\right)^e$ is a \emph{predicted load profile}\index{predicted load profile}, i.e., a vector of predicted loads for each job type. Similar to our implementation of single hitting costs, we use a varying number of $n$ samples to describe the distribution of predicted loads. The load profile for the current time slot $\tau$ is certain.

We observe that given a predicted load profile $\mathcal{P}_t$ with $n_i$ samples of loads with type $i \in [e]$, the sampled load profiles are all combinations of the samples of individual loads, resulting in $\prod_{i=1}^e n_i$ sampled load profiles. To reduce the number of considered samples, we randomly select $n$ sampled loads of each type to produce $n$ load profile samples. This approach is reasonable as the samples of the individual loads were only an approximate prediction, to begin with.

Another core element of a model are \emph{model outputs}\index{model outputs}. These are results that are returned from the model, which reach beyond simply some cost. For our data center model, we implemented outputs to return the energy cost instead of the revenue loss and the determined optimal assignment of jobs to server types.

The implementation of our data center model consists of separate models for energy consumption, energy cost, delay, revenue loss, and switching cost, which can be combined in various ways to obtain a concrete model of a data center or a network of data centers. For the generation of instances of SBLO, the model must be limited to a single location, source, and job type. For the generation of instances of SLO, the model further assumes full utilization of every active server during each time slot and averages the energy cost over the time horizon.

Crucially, we assume in our model that job arrivals fall precisely onto the beginning of new time slots. Without this approximation, we cannot ensure that all arriving jobs can be completed by the end of the time slot. In practice, this requires matching jobs to time slots either by delaying them to the next time slots (in which case they are available from the beginning) or simply assuming they are available from the beginning of the current time slot.

\subsection{Streaming and Practical Use}\label{section:implementation:architecture:streaming}

Using an offline algorithm is relatively straightforward. Given a model and offline inputs, we can generate a problem instance that the offline algorithm can then solve. With online algorithms, this becomes more involved. For example, some uses require the immediate streaming of an online algorithm from an initial to a final time slot. This is required by the Lazy Budgeting algorithm (\cref{alg:md:lazy_budgeting:sblo_c}), which streams an intermediate online algorithm for the sub time slots of the current time slot.

However, in general, the challenge in streaming an online algorithm is to remember the algorithm's state with infrequent incoming iterations (once per time slot). To this end, we implemented a simple client-server architecture. Here, the server runs in the background and remembers the problem instance and state of the algorithm. Whenever new information arrives, the client requests the next iteration from the server. On initialization of the server, the online algorithm can be streamed up to the current time slot given previous offline inputs. This architecture has two main benefits: First, the responsibility for memorizing the problem instance and state of the algorithm is offloaded entirely to the server. Second, the initialization of the server only requires a model and (optionally) an offline input. Requesting the next iteration only requires the relevant online input. In particular, the problem instance remains opaque, drastically simplifying the interaction with the interface and reducing the amount of data that needs to be remembered, and serialized, sent, and deserialized from the client to the server. This simplification also makes it feasible to stream an online algorithm relying entirely on Python bindings.

\section{Focuses}

We now turn to describe more general aspects of our implementation.

\subsection{Adaptability}

A central point of focus was to maximize the adaptability and general applicability of our implementation. More concretely, we focused on minimizing the number of cross-dependencies between the components described in \cref{section:implementation:architecture} to maximize the reusability across all components of the architecture. For example, an entirely new application can be supported solely by implementing a new model. Once a new algorithm is implemented, it can immediately be substituted for any other algorithm using any previously implemented model while benefiting from the infrastructure built around streaming online algorithms. This allows for faster prototyping of new algorithms while assessing their practical performance, particularly when using predictions. Also, this allows determining which algorithm performs best in a concrete application instance.

Moreover, using the client-server architecture, algorithms for smoothed online convex optimization can be utilized without much effort in any practical application. As the implementation is in Rust, it can be interfaced with any other programming language. For our case studies in \cref{chapter:case_studies}, we already implemented Python bindings~\cite{Huebotter2021_3}.

\subsection{Efficiency}

Next to the strong safety guarantees of Rust, the main benefit of an implementation in Rust is its memory efficiency and speed. Rust achieves similar performance to C or C++~\cite{Benchmarksgame, Rust, Perkel2020}. At the same time, Rust is readable and allows for high-level abstractions and polymorphism, which we use heavily to maximize adaptability.

In our implementation, we heavily parallelize tasks to achieve optimal performance~\cite{Matsakis2015}. For example, we parallelize calculating hitting costs for multiple predicted samples or determining the optimal predecessor of a vertice in iterations of the dynamic program finding a shortest path in a graph.

\subsection{Numeric Computations and Rounding}

In implementations involving numeric computations to some precision $\epsilon$ and frequent flooring or ceiling operations, it is essential to round numeric results to precision; otherwise, the results are not numerically stable. In our experimens, we use a precision of $\epsilon = 10^{-2}$, which performs well when results are rounded to integral solutions. For example, if the result of some numeric computation with precision epsilon is $10^{-3}$, which is then ceiled, we would falsely obtain $1$ as the result if we did not apply the precision to the result before proceeding with the algorithm.

For convex programs, we additionally interpret $\epsilon$ as relative to the absolute value of the optimization to maintain good performance.

\chapter{Case Studies}\label{chapter:case_studies}

This chapter examines the performance of the previously described models and algorithms using real server traces. Thereby we focus on two aspects: First, we are interested in how well the discussed algorithms compare in absolute terms and relative to each other. Second, we are interested in the general promise of dynamically right-sizing data centers, which we study by conservatively estimating cost savings and relating them to previous research.

\section{Method}

First, we describe our experimental setup. We begin with a detailed discussion of the characteristics of the server traces, which we use as a basis for our analysis. Then, we examine the underlying assumptions of our analysis. This is followed by a discussion of alternative approaches to right-sizing data centers, which we use as a foundation for estimating the cost savings resulting from dynamic right-sizing of data centers. Next, we describe the general model parameters we use in our analysis and relate them to previous research. Lastly, we introduce the precise performance metrics used in the subsequent sections.

Throughout our experiments, we seek to determine conservative approximations for the resulting performance and cost savings. Our experimental results were obtained on a machine with 16 GB memory and an Intel Core i7-8550U CPU with a base clock rate of 1.80GHz.

\subsection{Traces}\label{section:case_studies:method:traces}

We use several traces with varying characteristics for our experiments. Some traces are from clusters rather than individual data centers. However, to simplify our analysis, we assume traces apply to a single data center without restricting the considered server architectures.

\citeauthor{Amvrosiadis2018}~\cite{Amvrosiadis2018} showed that the characteristics of traces vary drastically even within a single trace when different subsets are considered individually. Their observation shows that it is crucial to examine long server traces and various server traces from different sources to gain a proper perspective of real-world performance.

\subsubsection{Characteristics}

To understand the varying effectiveness of dynamic right-sizing for the considered traces, we first analyze the properties of the given traces.

The most immediate and fundamental properties of a trace are its duration, the number of appearing jobs, the number of job types, and the underlying server infrastructure -- especially whether this infrastructure is homogeneous or heterogeneous.

Then, we also consider several more specific characteristics. The \emph{interarrival time}\index{interarrival time} (or submission rate) of jobs is the distribution of times between job arrivals. This distribution indicates the average system load as well as load uniformity. The \emph{peak-to-mean ratio (PMR)}\index{peak-to-mean ratio} is defined as the ratio of the maximum load and the mean load. It is a good indicator of the uniformity of loads. We refer to time slots as \emph{peaks}\index{peak load} when their load is greater than the 0.9-quantile of loads. We call the ratio of the 0.9-quantile of loads and the mean load \emph{true peak-to-mean-ratio (TPMR)}\index{true peak-to-mean ratio} as it is less sensitive to outliers than the PMR. We refer to periods between peaks as \emph{valleys}\index{valley}. More concretely, we refer to the time between two consecutive peaks as \emph{peak distance}\index{peak distance} and the number of consecutive time slots up to a time slot with a smaller load as \emph{valley length}\index{valley length}. Further, we say that a trace follows a \emph{diurnal pattern}\index{diurnal pattern} if during every 24 hours, excluding the final day, there is at least one valley spanning 12 hours or more. We exclude the final day as the final valley might be shortened by the end of the trace.

We also consider some additional information included in some traces, such as the measured scheduling rate (or queuing delay), an indicator for utilization.

\subsubsection{Overview}

We now give an overview of all used traces. For our initial analysis, we use a time slot length of 10 minutes.

\paragraph{MapReduce\footnote{MapReduce is a programming model for processing and generating large data sets in a functional style~\cite{Dean2004}} Workload from a Hadoop\footnote{Apache Hadoop is an open-source software for managing clusters} Cluster at Facebook~\cite{SWIM2013}} This trace encompasses three day-long traces from 2009 and 2010, extracted from a 6-month and a 1.5-month-long trace containing 1 million homogeneous jobs each. The traces are visualized in \cref{fig:facebook:histogram} and summarized in \cref{tab:facebook}. The cluster consists of 600 machines which we assume to be homogeneous. For the trace from 2010, we adjust the maximum number of servers to 1000 as otherwise the trace is infeasible under our models. \Cref{fig:facebook:schedule} visualizes the corresponding dynamic and static offline optimal schedules under our second model (which is described in \cref{section:case_studies:traces:model-parameters}). The trace was published by \citeauthor{SWIM2013}~\cite{SWIM2013} as part of the SWIM project at UC Berkeley.

\begin{figure}
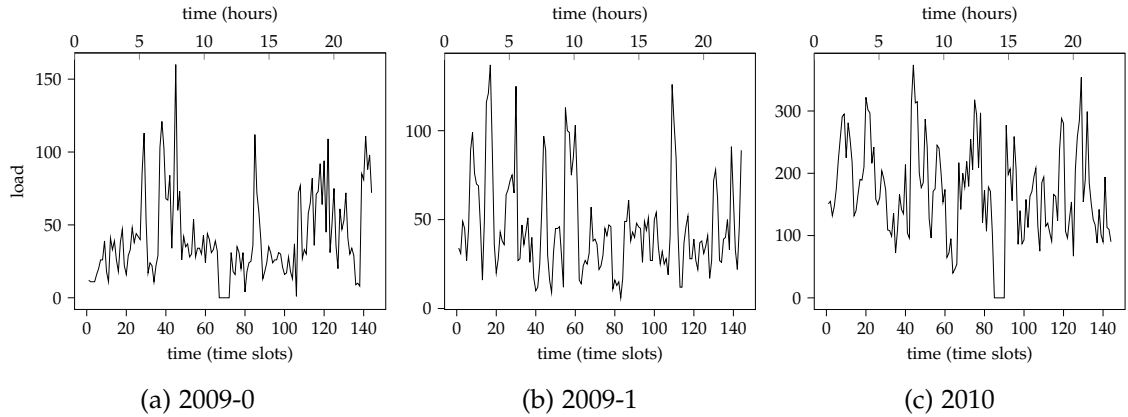

    \begin{subfigure}[b]{.3425\linewidth}
    \resizebox{\textwidth}{!}{\input{thesis/figures/facebook_2009_0_histogram.tex}}
    \caption{2009-0}
    \end{subfigure}
    \begin{subfigure}[b]{.32\linewidth}
    \resizebox{\textwidth}{!}{\input{thesis/figures/facebook_2009_1_histogram.tex}}
    \caption{2009-1}
    \end{subfigure}
    \begin{subfigure}[b]{.32\linewidth}
    \resizebox{\textwidth}{!}{\input{thesis/figures/facebook_2010_histogram.tex}}
    \caption{2010}
    \end{subfigure}
    \caption{Facebook MapReduce workloads.}
    \label{fig:facebook:histogram}
\end{figure}

\begin{figure}
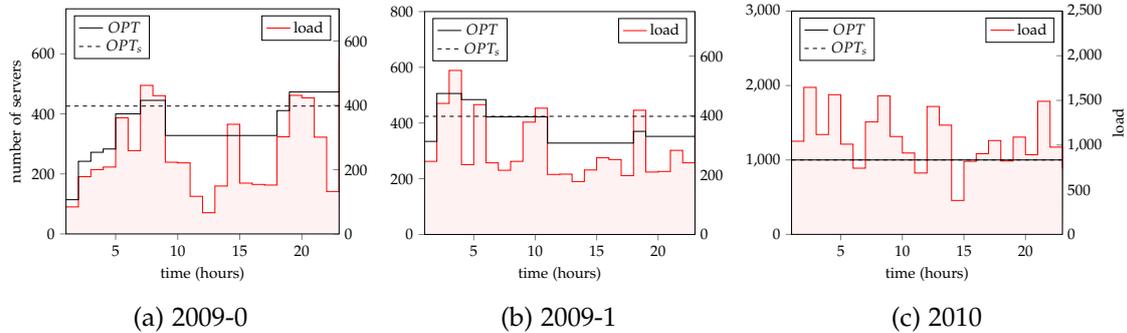

    \begin{subfigure}[b]{.33\linewidth}
    \resizebox{\textwidth}{!}{\input{thesis/figures/facebook_2009_0_schedule.tex}}
    \caption{2009-0}
    \end{subfigure}
    \begin{subfigure}[b]{.3075\linewidth}
    \resizebox{\textwidth}{!}{\input{thesis/figures/facebook_2009_1_schedule}}
    \caption{2009-1}
    \end{subfigure}
    \begin{subfigure}[b]{.3475\linewidth}
    \resizebox{\textwidth}{!}{\input{thesis/figures/facebook_2010_schedule}}
    \caption{2010}
    \end{subfigure}
    \caption{Optimal dynamic and static offline schedules for the last day of the Facebook workloads. The left y axis shows the number of servers of the static and dynamic offline optima at a given time (black). The right y axis shows the number of jobs (i.e., the load) at a given time (red).}
    \label{fig:facebook:schedule}
\end{figure}

\begin{table}
    \centering
    \begin{tabularx}{\textwidth}{>{\bfseries}l|X|X|X}
        characteristic & 2009-0 & 2009-1 & 2010 \\\hline
        duration & 1 day & 1 day & 1 day \\
        number of jobs & 6 thousand & 7 thousand & 24 thousand \\
        median interarrival time & 7 seconds & 7 seconds & 2 seconds \\
        PMR & 3.91 & 2.97 & 2.2 \\
        TPMR & 2.04 & 1.93 & 1.69 \\
        mean peak distance & 95 minutes & 106 minutes & 87 minutes \\
        mean valley length & 44 minutes & 36 minutes & 35 minutes \\
    \caption{Characteristics of Facebook's MapReduce workloads.}
    \end{tabularx}
    \label{tab:facebook}
\end{table}

\paragraph{Los Alamos National Lab HPC Traces~\cite{Amvrosiadis2018_3, Amvrosiadis2018, Amvrosiadis2018_2}} This trace comprises two separate traces from high-performance computing clusters from Los Alamos National Lab (LANL). The traces were published by \citeauthor{Amvrosiadis2018}~\cite{Amvrosiadis2018} as part of the Atlas project at Carnegie Mellon University.

The first trace is from the Mustang cluster, a general-purpose cluster consisting of 1600 homogeneous servers. Jobs were assigned to entire servers. The dataset covers 61 months from October 2011 to November 2016 and is shown in \cref{fig:los_alamos:histogram}. Note that the PMR is large at $622$ due to some outliers in the data. The median job duration is roughly 7 minutes, although the trace includes some extremely long-running outliers, resulting in a mean job duration of over 2.5 hours. In the trace, jobs were assigned to one or multiple servers. To normalize the trace, we consider each job once for each server it was processed on. \Cref{fig:los_alamos:schedule} shows the dynamic and static offline optimal schedules under our second model.

The second trace is from the Trinity supercomputer. This trace is very similar to the Mustang trace but includes an even more significant number of long-running jobs. We, therefore, do not consider this trace in our analysis.

\paragraph{Microsoft Fiddle Trace~\cite{Jeon2019}} This trace consists of deep neural network training workloads on internal servers from Microsoft. The trace was published as part of the Fiddle project from Microsoft Research. The jobs are run on a heterogeneous set of servers which we group based on the number of GPUs of each server. There are 321 servers with two GPUs and 231 servers with eight GPUs. The median job duration is just below 15 minutes. The load profiles are visualized in \cref{fig:microsoft:histogram}.

The CPU utilization of the trace is extremely low, with more than 80\% of servers running with utilization 30\% or less~\cite{Santhanam2019}. However, memory utilization is high, with an average of more than 80\% indicating that overall server utilization is already very high~\cite{Santhanam2019}. Again, the PMR is rather large at 89.43 due to outliers.

In our model, we adjust the runtime of jobs relative to the number of available GPUs in the respective server, i.e., the average runtime of jobs on a 2-GPU-server is four times as long as the average runtime of jobs on an 8-GPU-server. We adjust for the increased energy consumption of a server with eight GPUs by increasing the energy consumption of servers with two GPUs by a factor of 4.2. We also associate a fifteen times higher switching cost with servers with eight GPUs.

The dynamic and static offline optimal schedules under our second model are shown in \cref{fig:microsoft:schedule}. Note that under the given load servers with two GPUs are preferred to servers with eight GPUs when they are only needed for a short period due to their lower switching costs. This might seem counterintuitive at first, as 2-GPU-servers seem to be strictly better than 8-GPU-servers as the operating and switching cost of 8-GPU-servers is worse by a factor greater than four than the respective cost of 2-GPU-servers. However, we assume an average job runtime of 7.5 minutes on 8-GPU-servers as opposed to an average job runtime of 30 minutes on 2-GPU-servers (a factor of four), implying that 8-GPU-servers can process more than four jobs in an hour without a significant increase in delay, whereas 2-GPU-servers are limited to one job per time slot.

\paragraph{Alibaba Trace~\cite{Alibaba2018}} This trace consists of a mixture of long-running applications and batch jobs. We are using their trace from 2018, covering eight days. The trace is visualized in \cref{fig:alibaba:histogram}, the dynamic and static offline optimal schedules under our second model are shown in \cref{fig:alibaba:schedule}. The jobs are processed on 4000 homogeneous servers. In our models, we assume a total of 10,000 servers to ensure that the number of servers is not a bottleneck. Jobs themselves are grouped into 11 types which we further simplify to 4 types based on their average runtime. We consider \emph{short}, \emph{medium}, \emph{long}, and \emph{very long} jobs. Their average runtime in the trace is shown in \cref{tab:alibaba:job_types}. The mean job duration is just below 15 minutes. The median job duration is 8 seconds, and the mean job duration is just over 1.5 minutes.

\begin{table}
    \centering
    \begin{tabularx}{\textwidth}{>{\bfseries}l|c}
        job type & mean runtime \\\hline
        short & 68 seconds \\
        medium & 196 seconds \\
        long & 534 seconds \\
        very long & 1180 seconds \\
    \caption{Characterization of the job types of the Alibaba trace.}
    \end{tabularx}
    \label{tab:alibaba:job_types}
\end{table}

Data from a previous trace indicates that mean CPU utilization varies between 10\% and 40\% while mean memory utilization varies between 40\% and 65\%~\cite{Lu2017}. This indicates that the overall server utilization is not optimal.

In our model, we scale job runtimes by a factor of 2.5 from short to very long jobs, roughly matching the runtimes of jobs from the trace.

\begin{figure}
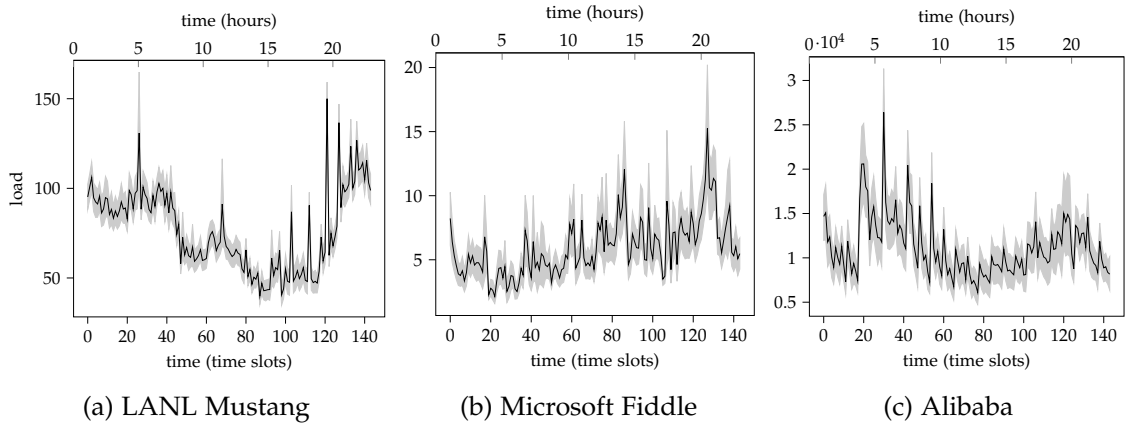

    \begin{subfigure}[b]{.3425\linewidth}
    \resizebox{\textwidth}{!}{\input{thesis/figures/los_alamos_mustang_histogram.tex}}
    \caption{LANL Mustang}\label{fig:los_alamos:histogram}
    \end{subfigure}
    \begin{subfigure}[b]{.32\linewidth}
    \resizebox{\textwidth}{!}{\input{thesis/figures/microsoft_histogram}}
    \caption{Microsoft Fiddle}\label{fig:microsoft:histogram}
    \end{subfigure}
    \begin{subfigure}[b]{.32\linewidth}
    \resizebox{\textwidth}{!}{\input{thesis/figures/alibaba_histogram}}
    \caption{Alibaba}\label{fig:alibaba:histogram}
    \end{subfigure}
    \caption{LANL Mustang, Microsoft Fiddle, and Alibaba traces. The figures display the average number of job arrivals throughout a day. The interquartile range is shown as the shaded region.}
\end{figure}

\begin{figure}
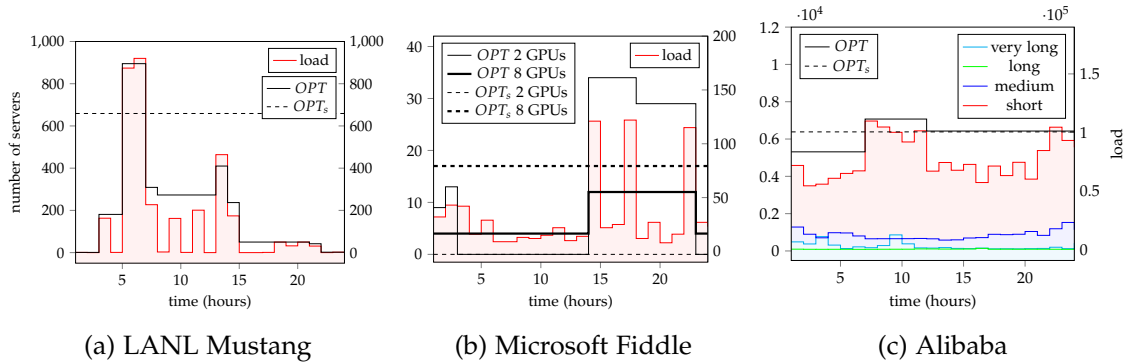

    \begin{subfigure}[b]{.345\linewidth}
    \resizebox{\textwidth}{!}{\input{thesis/figures/los_alamos_mustang_schedule.tex}}
    \caption{LANL Mustang}\label{fig:los_alamos:schedule}
    \end{subfigure}
    \begin{subfigure}[b]{.305\linewidth}
    \resizebox{\textwidth}{!}{\input{thesis/figures/microsoft_schedule}}
    \caption{Microsoft Fiddle}\label{fig:microsoft:schedule}
    \end{subfigure}
    \begin{subfigure}[b]{.335\linewidth}
    \resizebox{\textwidth}{!}{\input{thesis/figures/alibaba_schedule}}
    \caption{Alibaba}\label{fig:alibaba:schedule}
    \end{subfigure}
    \caption{Optimal dynamic and static offline schedules for the last day of the LANL Mustang, Microsoft Fiddle, and the second to last day of the Alibaba trace. The left y axis shows the number of servers of the static and dynamic offline optima at a given time (black). The right y axis shows the number of jobs (i.e., the load) at a given time (red).}
\end{figure}

\paragraph{} We have seen traces from very different real-world use cases. The Microsoft Fiddle trace is based on a heterogeneous server architecture, and the Alibaba trace receives heterogeneous loads. The PMR and valley lengths of the days used in our analysis are shown in \cref{tab:pmr_vl}. Interestingly, as shown in \cref{tab:traces}, their TPMR, peak distances, and valley lengths are mostly similar.

\subsection{Assumptions}

We impose a couple of assumptions to simplify our analysis. First, and already mentioned, our analysis is inherently limited by the traces we used as a basis for our experiments. While we examine a wide variety of traces, the high variability in traces indicates they are a fundamental limitation to any estimation of real-world performance.

Another common limitation of models is that the interarrival times of jobs are on the order of seconds or smaller~\cite{Amvrosiadis2018}. However, this is not a limitation of our analysis as we are using a general Poisson process with an appropriate mean arrival rate in our delay model.

In the context of high-performance computing, jobs typically have a \emph{gang scheduling}\index{gang scheduling} requirement, i.e., a requirement that related jobs are processed simultaneously even though they are run on different hardware~\cite{Amvrosiadis2018}. For simplification, we assume this requirement always to be satisfied. However, this is not a substantial limitation as the scheduling of jobs within a time slot is not determined by the discussed algorithms and instead left to the server operator. Nevertheless, in principle, the gang scheduling requirement may render some schedules infeasible if the processing time on servers exceeds the length of a time slot when gang scheduling constraints are considered.

There are also some limitations resulting from the design of our model. As was mentioned previously, we assume that the jobs arrive at the beginning of a new time slot rather than at random times throughout the time slot. Moreover, we assumed that for every job, a server type exists that can process this job within one time slot. In other words, there exists no job running longer than $\delta$. We have seen in \cref{section:case_studies:method:traces} that this assumption is violated in most practical scenarios. In \cref{section:application:dynamic_duration}, we described how this assumption can be removed. The same approach can also be used to remove the assumption that jobs must arrive at the beginning of a time slot.

\subsection{Alternatives to Right-Sizing Data Centers}\label{section:case_studies:method:alternatives}

To determine the benefit of dynamically right-sizing data centers, we must first describe the alternative strategies to managing a data center. We will then use these approaches as a point of reference in our analysis.

Most data centers are statically provisioned; that is, the configuration of active servers is only changed rarely (often manually) and remains constant during most periods~\cite{Whitney2014}. To support the highest loads, the data centers are peak-provisioned, i.e., the number of servers is chosen such that they suffice to process all jobs even during times where most jobs arrive. Moreover, as a safety measure, data centers are typically provisioned to handle much higher loads than the loads encountered in practice~\cite{Whitney2014}.

\begin{table}
    \centering
    \begin{tabularx}{\textwidth}{>{\bfseries}l|X|X|X}
        characteristic & LANL Mustang & Microsoft Fiddle & Alibaba \\\hline
        duration & 5 years & 30 days & 8 days \\
        number of jobs & 20 million & 120 thousand & 14 million \\
        median interarrival time & 0 seconds & 8 seconds & 0 seconds \\
        PMR & 621.94 & 89.43 & 3.93 \\
        TPMR & 2.5 & 1.68 & 1.77 \\
        mean peak distance & 100 minutes & 105 minutes & 89 minutes \\
        mean valley length & 120 minutes & 115 minutes & 74 minutes \\
        diurnal pattern & yes & - & yes \\
    \caption{Characteristics of the LANL Mustang, Microsoft Fiddle, and Alibaba traces.}
    \end{tabularx}
    \label{tab:traces}
\end{table}

Naturally, traces with a high PMR or long valleys are more likely to benefit from alternatives to static provisioning. Therefore another widely used alternative is \emph{valley filling}\index{valley filling}, which aims to schedule lower priority jobs (i.e., some batch jobs) during valleys. In an ideal scenario, this approach can achieve $\text{PMR} \approx 1$, which would allow for efficient static provisioning. Crucially, this approach requires a large number of low-priority jobs which may be processed with a significant delay (requiring a considerable minimum perceptible delay $\delta_i$ for a large number of jobs of type $i$), and thus in most cases, valleys cannot be eliminated entirely. \citeauthor{Lin2011}~\cite{Lin2011} showed that dynamic-right sizing can be combined with valley filling to achieve a significant cost reduction. The optimal balancing of dynamic right-sizing and valley filling is mainly determined by the change to the PMR. \citeauthor{Lin2011}~\cite{Lin2011} showed that cost savings of 20\% are possible with a PMR of 2 and a PMR of approximately 1.3 can still achieve cost savings of more than 5\%. Generally, the cost reduction vanishes once the PMR approaches $1$, which may happen between 30\% to 70\% mean background load~\cite{Lin2011}. The results when dynamic right-sizing is used together with valley filling can be estimated from previous results.

\subsection{Performance Metrics}

Let $OPT$ denote the dynamic offline optimum and $OPT_s$ denote the static offline optimum. In our analysis, the \emph{normalized cost}\index{normalized cost} of an online algorithm is the ratio of the obtained cost and the dynamic optimal offline cost, i.e. $NC(ALG) = c(ALG) / c(OPT)$. Further, we base our estimated \emph{cost reduction}\index{cost reduction} on an optimal offline static provisioning: \begin{align*}
    CR(ALG) = \frac{c(OPT_s) - c(ALG)}{c(OPT_s)}.
\end{align*} Note that this definition is similar to the definition of regret, but expressed relative to the overall cost. We refer to $SDR = c(OPT_s) / c(OPT)$ as the \emph{static/dynamic ratio}\index{static/dynamic ratio}, which is closely related to the \emph{potential cost reduction}\index{potential cost reduction} $PCR = CR(OPT)$.

\subsection{Previous Results}

\citeauthor{Lin2011}~\cite{Lin2011} showed that the cost reduction is directly proportional to the PMR and inversely proportional to the normalized switching cost. Additionally, \citeauthor{Lin2011}~\cite{Lin2011} showed that, as one would expect, the possible cost reduction decreases as the delay cost assumes a more significant fraction of the overall hitting costs. In practice, this can be understood as the effect of making the model more conservative.

\begin{table}
    \centering
    \begin{tabularx}{\textwidth}{>{\bfseries}l|c|c}
        trace & PMR & mean valley length (hours) \\\hline
        Facebook 2009-0 & 2.115 & 2.565 \\
        Facebook 2009-1 & 1.913 & 1.522 \\
        Facebook 2010 & 1.549 & 1.435 \\
        LANL Mustang & 6.575 & 1.167 \\
        Microsoft Fiddle & 3.822 & 2.125 \\
        Alibaba & 1.339 & 2.792 \\
    \caption{PMR and mean valley length of the traces used in our analysis. Note that the valley lengths are typically shorter than the normalized switching cost of our model.}
    \end{tabularx}
    \label{tab:pmr_vl}
\end{table}

\subsection{Model Parameters}\label{section:case_studies:traces:model-parameters}

We now describe how we parametrized our model in our case studies. In our models, we strive to choose conservative estimates to under-estimate the cost savings from dynamically right-sizing data centers. This approach is similar to the study by \citeauthor{Lin2011}~\cite{Lin2011}. \Cref{tab:model} gives an overview of the used parameters producing the results of subsequent sections.

\paragraph{Energy} We use the linear energy consumption model from \autoref{eq:energy_model:1} in our experiments. In their analysis, \citeauthor{Lin2011}~\cite{Lin2011} choose energy cost and energy consumption such that the fixed energy cost (i.e., the energy cost of a server when idling) is $1$ and the dynamic energy cost is $0$ as, on most servers, the fixed costs dominate the dynamic costs~\cite{Clark2005}. We investigate this model and an alternative model. In the alternative model, we estimate the power consumption of a server with 1 kW during peak loads and with 500 W when idling to yield a conservative estimate (as cooling costs are included). According to the U.S. Energy Information Administration (EIA), the average cost of energy in the industrial sector in the United States during April 2021 was 6.77 cents per kilowatt-hour~\cite{EIA2021}. We use this as a conservative estimate as data centers typically use a more expensive portfolio of energy sources. If the actual carbon cost of the used energy were to be considered, which is the case in some data centers as discussed in \cref{section:application:operating_cost:energy}, energy costs are likely to be substantially higher.

\paragraph{Revenue Loss} According to measurements, a 500 ms increase in delay results in a revenue loss of 20\% or 0.04\%/ms~\cite{Lin2012, Hamilton2009}. Thus, scaling the delay measured in ms by 0.1 can be used as a slight over-approximation of revenue loss. \citeauthor{Lin2011}~\cite{Lin2011} choose the minimal perceptible delay as 1.5 times the time to run a job, which is a very conservative estimate if valley filling is assumed a viable alternative. In our model, we choose the minimal perceptible delay as 2.5 times the time to run a job which is equivalent as we also added the processing time of a job to the delay. In the case of valley filling, jobs are typically processed with a much more significant delay. Similar to \citeauthor{Lin2012}~\cite{Lin2012}, we also estimate a constant network delay of 10 ms.

\paragraph{Switching Cost} We mentioned in \cref{section:application:switching_cost} that in practice, the switching cost should be on the order of operating a server between an hour to several hours. To obtain a conservative estimate, we choose $\beta$ such that the normalized switching cost times the length of a time slot equals 4 hours.

\paragraph{Time Slot Length} We choose a time slot length of 1 hour. We further assume that the average processing time of jobs is $\delta / 2$ unless noted otherwise.

\begin{table}
    \centering
    \begin{tabularx}{\textwidth}{>{\bfseries}l|X|X}
        parameter & model 1 & model 2 \\\hline
        time slot length & 1 hour & 1 hour \\
        energy cost & $c=1$ & $c=0.0677$ \\
        energy consumption & $\Phi_{\text{min}}=1, \Phi_{\text{max}}=1$ & $\Phi_{\text{min}}=0.5, \Phi_{\text{max}}=1$ \\
        revenue loss & $\gamma = 0.1, \delta_i = 2.5 \eta_i$ & $\gamma = 0.1, \delta_i = 2.5 \eta_i$ \\
        normalized switching cost & 4 hours & 4 hours \\
    \caption{Models used in our case studies. $\eta_i$ is the processing time of jobs of type $i$.}
    \end{tabularx}
    \label{tab:model}
\end{table}

\section{Uni-Dimensional Algorithms}

The results of this section are based on the final day of the LANL Mustang, Facebook, and the second to last day of the Alibaba trace. We begin by discussing the general features of the traces. Then, we compare the uni-dimensional online algorithms with respect to their achieved normalized cost, cost reduction, and runtime.

\paragraph{Fractional vs. Integral Cost} For all traces, the ratio of the fractional and the integral costs is 1 for a precision of at least $10^{-3}$. This is not surprising due to the large number of servers used in each model.

\begin{figure}
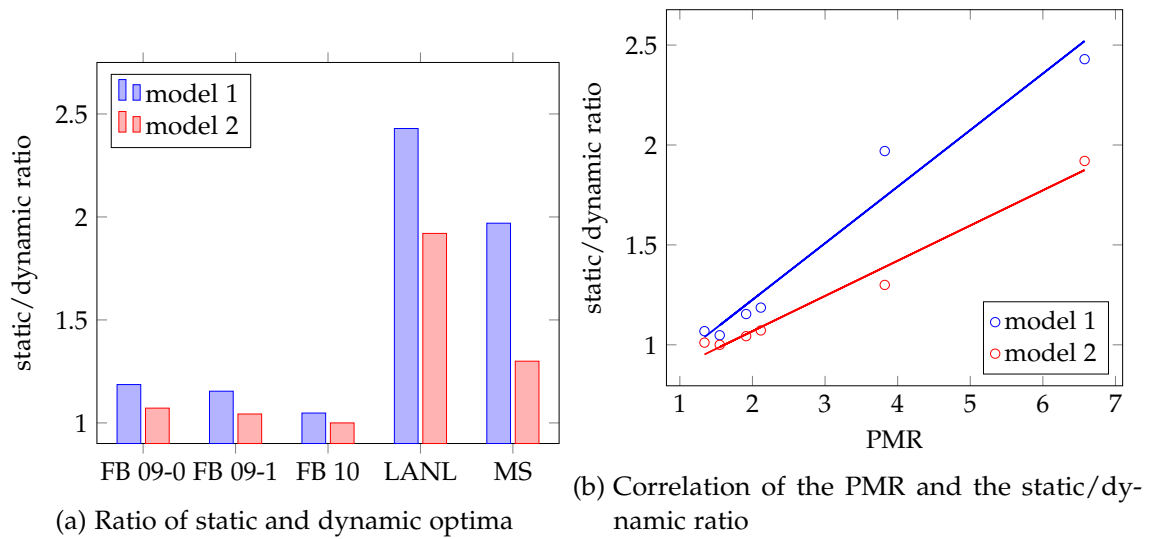

    \begin{subfigure}[b]{.5\linewidth}
    \resizebox{\textwidth}{!}{\input{thesis/figures/opt_vs_opts}}
    \caption{Ratio of static and dynamic optima}
    \end{subfigure}
    \begin{subfigure}[b]{.5\linewidth}
    \resizebox{\textwidth}{!}{\input{thesis/figures/opts_opt_vs_pmr}}
    \caption{Correlation of the PMR and the static/dynamic ratio}\label{fig:case_studies:ud:opt_vs_opts:pmr}
    \end{subfigure}
    \caption{Ratio of static and dynamic offline optima for each trace. The LANL Mustang and Microsoft Fiddle traces have a significantly higher PMR than the remaining traces. Generally, we observe a strong correlation of PMR and the static/dynamic ratio.}\label{fig:case_studies:ud:opt_vs_opts}
\end{figure}

\begin{figure}
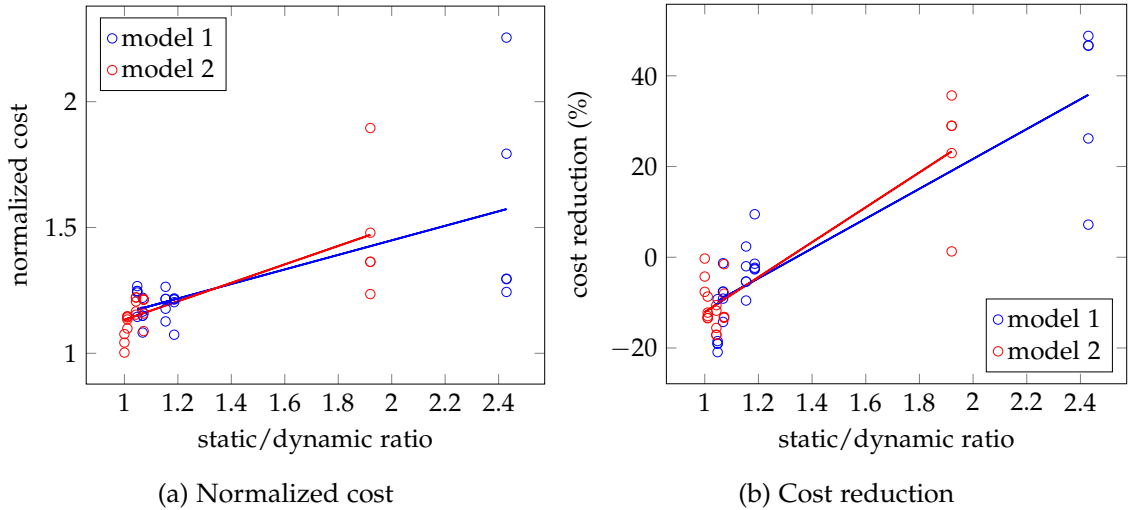

    \begin{subfigure}[b]{.49\linewidth}
    \resizebox{\textwidth}{!}{\input{thesis/figures/opt_vs_opts_against_normalized_cost}}
    \caption{Normalized cost}\label{fig:case_studies:ud:opt_vs_opts_against_normalized_cost}
    \end{subfigure}
    \begin{subfigure}[b]{.51\linewidth}
    \resizebox{\textwidth}{!}{\input{thesis/figures/opt_vs_opts_against_mean_cost_reduction}}
    \caption{Cost reduction}\label{fig:case_studies:ud:opt_vs_opts_against_mean_cost_reduction}
    \end{subfigure}
    \caption{Effect of the ratio of static and dynamic optima on the cost reduction and normalized cost achieved by the memoryless algorithm.}
\end{figure}

\paragraph{Dynamic vs. Static Cost} The dynamic and static costs differ significantly depending on the trace. The ratio of dynamic and static optimal costs for each trace is shown in \cref{fig:case_studies:ud:opt_vs_opts}.

\Cref{fig:case_studies:ud:opt_vs_opts_against_mean_cost_reduction} shows a strong positive correlation between the average cost reduction achieved by the memoryless algorithm and the ratio of the static and dynamic optima. As $OPT_s / OPT$ is directly linked to the PMR, this also indicates a strong correlation between cost reduction and the PMR. Even under our very conservative estimates of parameters, we achieve a significant cost reduction when the ratio of the static and dynamic offline optimum exceeds 1.5. Similar to \citeauthor{Lin2011}~\cite{Lin2011}, we observe that cost savings increase rapidly as the PMR increases.

We also observe in \cref{fig:case_studies:ud:opt_vs_opts_against_normalized_cost} that as the static/dynamic ratio increases, the normalized costs achieved by the memoryless algorithm increases too but not as much as the potential energy savings, resulting in the observed significant cost reduction.

\begin{figure}
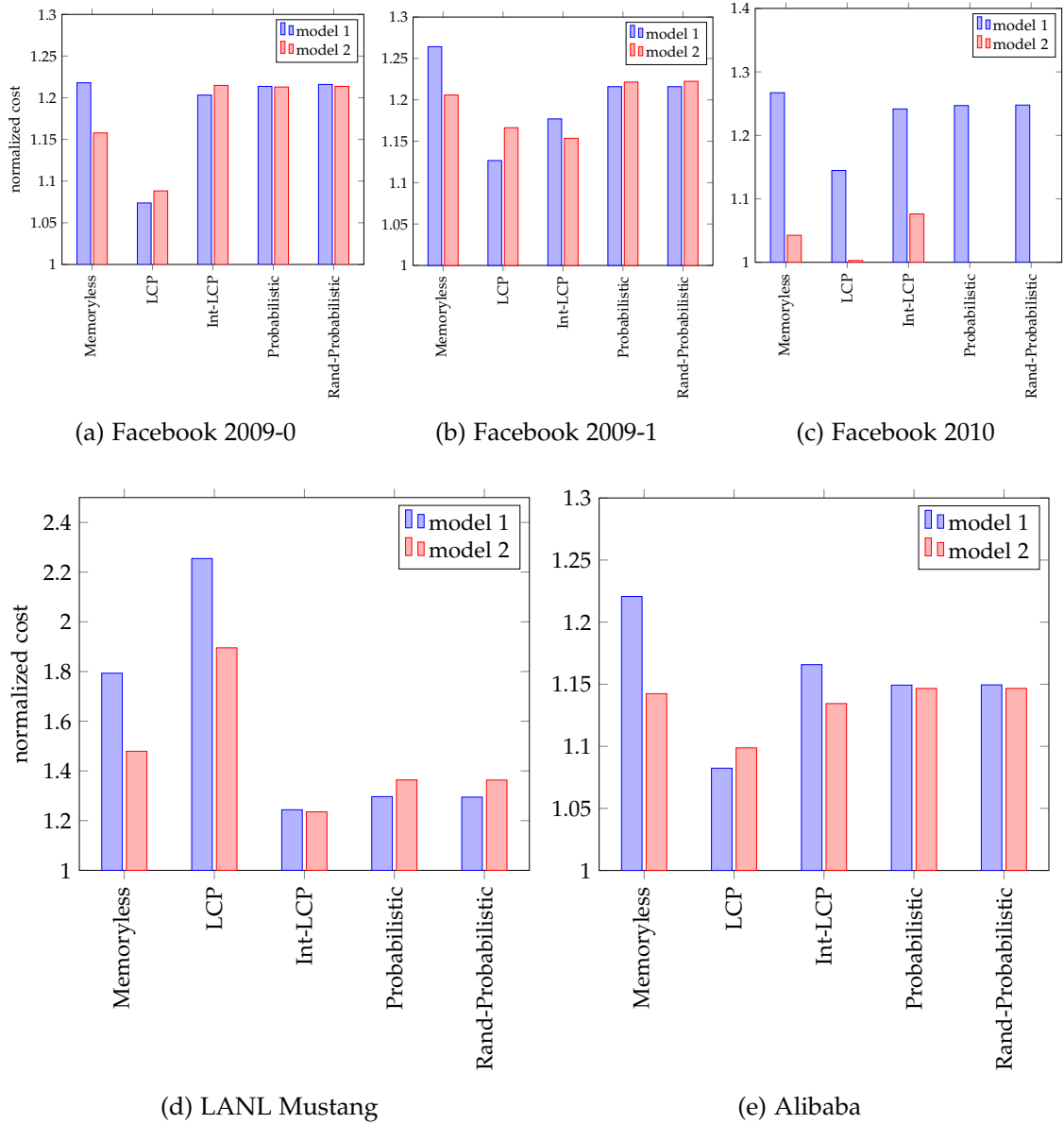

    \begin{subfigure}[b]{.3425\linewidth}
    \resizebox{\textwidth}{!}{\input{thesis/figures/fb90_normalized_cost}}
    \caption{Facebook 2009-0}
    \end{subfigure}
    \begin{subfigure}[b]{.32\linewidth}
    \resizebox{\textwidth}{!}{\input{thesis/figures/fb91_normalized_cost}}
    \caption{Facebook 2009-1}
    \end{subfigure}
    \begin{subfigure}[b]{.32\linewidth}
    \resizebox{\textwidth}{!}{\input{thesis/figures/fb10_normalized_cost}}
    \caption{Facebook 2010}
    \end{subfigure}
    \par\bigskip
    \begin{subfigure}[b]{.50\linewidth}
    \resizebox{\textwidth}{!}{\input{thesis/figures/lanl_normalized_cost}}
    \caption{LANL Mustang}
    \end{subfigure}
    \begin{subfigure}[b]{.48\linewidth}
    \resizebox{\textwidth}{!}{\input{thesis/figures/alibaba_normalized_cost}}
    \caption{Alibaba}
    \end{subfigure}
    \caption{Normalized costs of uni-dimensional online algorithms. For the Facebook 2010 trace, the second model results in an optimal schedule constantly using all servers, explaining the disparate performance compared to the first model. Further, Probabilistic and Rand-Probabilistic perform very poorly in this setting and are therefore not shown for model 2. Generally, Probabilistic and Rand-Probabilistic achieve similar results. Interestingly, we observe that LCP outperforms Int-LCP when the potential cost reduction is small. In contrast, Int-LCP and the probabilistic algorithms outperform Memoryless and LCP significantly when the potential cost reduction is large. \Cref{fig:case_studies:ud:lcp_vs_int_lcp} compares the schedules obtained by LCP and Int-LCP in greater detail.}\label{fig:case_studies:ud:normalized_cost}
\end{figure}

\begin{figure}
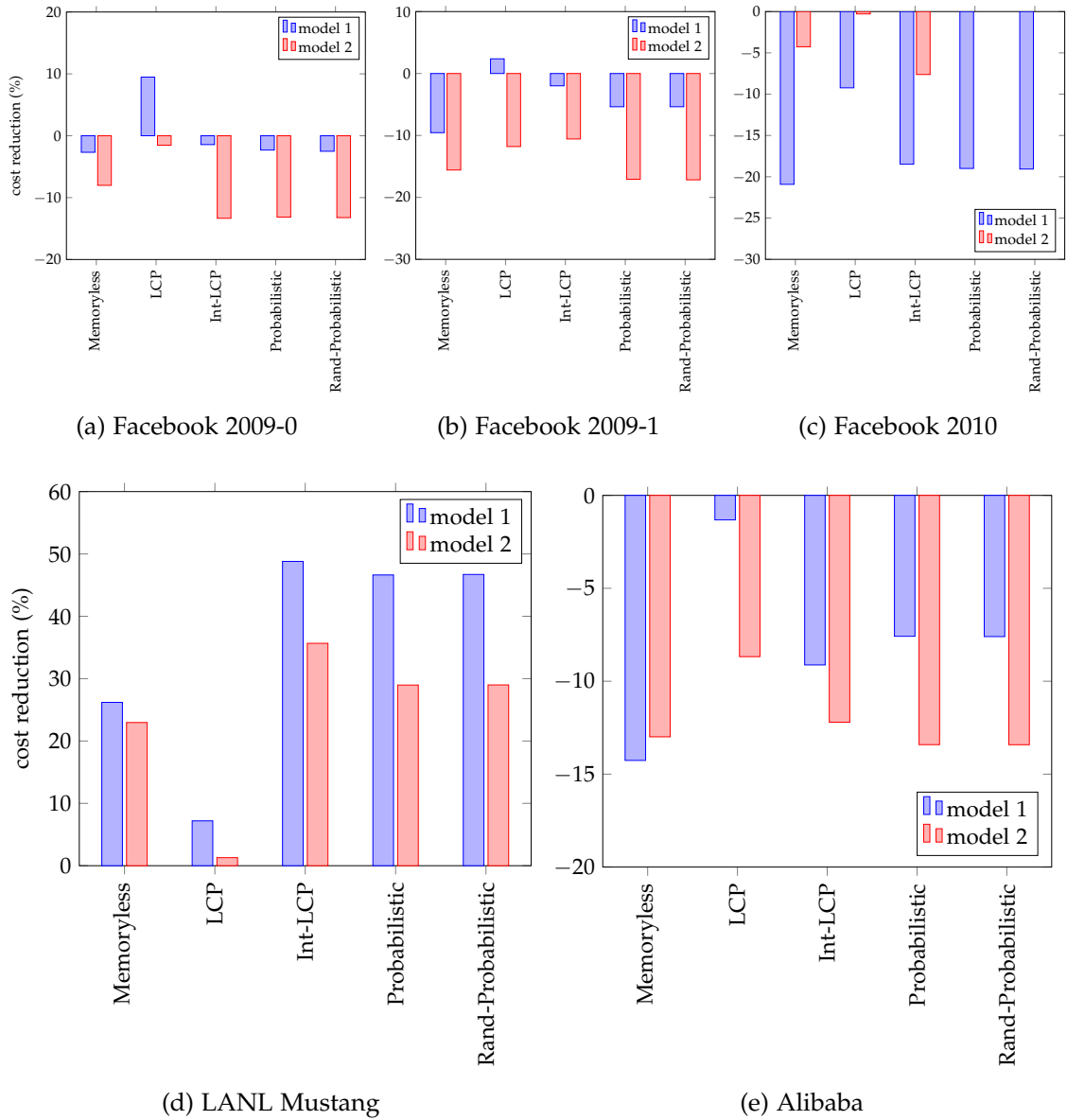

    \begin{subfigure}[b]{.3425\linewidth}
    \resizebox{\textwidth}{!}{\input{thesis/figures/fb90_cost_reduction}}
    \caption{Facebook 2009-0}
    \end{subfigure}
    \begin{subfigure}[b]{.32\linewidth}
    \resizebox{\textwidth}{!}{\input{thesis/figures/fb91_cost_reduction}}
    \caption{Facebook 2009-1}
    \end{subfigure}
    \begin{subfigure}[b]{.32\linewidth}
    \resizebox{\textwidth}{!}{\input{thesis/figures/fb10_cost_reduction}}
    \caption{Facebook 2010}
    \end{subfigure}
    \par\bigskip
    \begin{subfigure}[b]{.50\linewidth}
    \resizebox{\textwidth}{!}{\input{thesis/figures/lanl_cost_reduction}}
    \caption{LANL Mustang}
    \end{subfigure}
    \begin{subfigure}[b]{.48\linewidth}
    \resizebox{\textwidth}{!}{\input{thesis/figures/alibaba_cost_reduction}}
    \caption{Alibaba}
    \end{subfigure}
    \caption{Cost reduction of uni-dimensional online algorithms. For the Facebook 2010 trace, the second model results in an optimal schedule constantly using all servers, explaining the disparate performance compared to the first model. Further, Probabilistic and Rand-Probabilistic perform very poorly in this setting and are therefore not shown for model 2. Results are mainly determined by the normalized cost and the potential cost reduction (or static/dynamic ratio). We observe that the achieved cost reduction is dominated by the potential cost reduction (see \cref{fig:case_studies:ud:cost_reduction_vs_normalized_cost}).}\label{fig:case_studies:ud:cost_reduction}
\end{figure}

\begin{figure}
    \centering
    \input{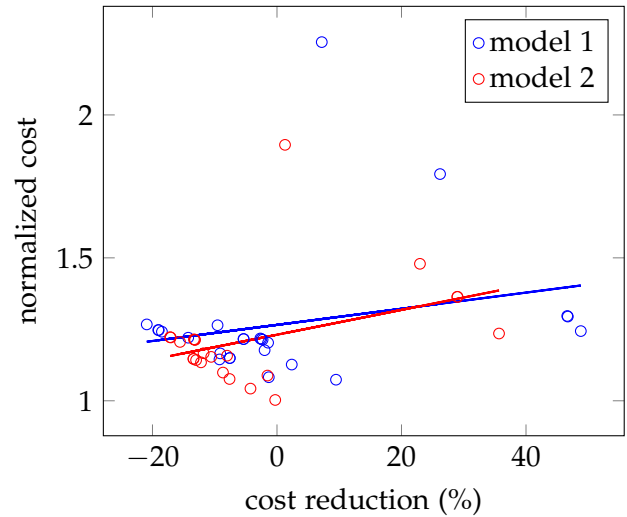}
    \caption{Weak positive correlation of achieved normalized cost and cost reduction. Intuitively, one would expect a strong negative correlation, i.e., the achieved cost reduction increases as the normalized cost approaches 1. Here, we find a positive correlation as the achieved cost reduction is dominated by the potential cost reduction (see \cref{fig:case_studies:ud:opt_vs_opts_against_mean_cost_reduction}).}\label{fig:case_studies:ud:cost_reduction_vs_normalized_cost}
\end{figure}

\paragraph{Normalized Cost} In the application of right-sizing data centers, we are interested in the cost associated with integral schedules. \Cref{fig:case_studies:ud:normalized_cost} shows the normalized cost of each algorithm. For fractional algorithms, we consider the cost of the associated integral schedule obtained by ceiling each configuration. Notably, LCP and Int-LCP perform differently depending on the trace and used model. We explore this behavior in \cref{fig:case_studies:ud:lcp_vs_int_lcp}.

\paragraph{Cost Reduction} \Cref{fig:case_studies:ud:cost_reduction} shows the achieved cost reduction. In general, we observe in \cref{fig:case_studies:ud:opt_vs_opts_against_normalized_cost}, \cref{fig:case_studies:ud:opt_vs_opts_against_mean_cost_reduction}, and \cref{fig:case_studies:ud:cost_reduction_vs_normalized_cost} that the achieved cost reduction is dominated by the potential cost reduction (which is mainly influenced by the PMR, see \cref{fig:case_studies:ud:opt_vs_opts:pmr}). When the potential cost reduction is small, algorithms with a smaller normalized cost in a particular setting, achieve a significantly higher cost reduction.

\begin{figure}
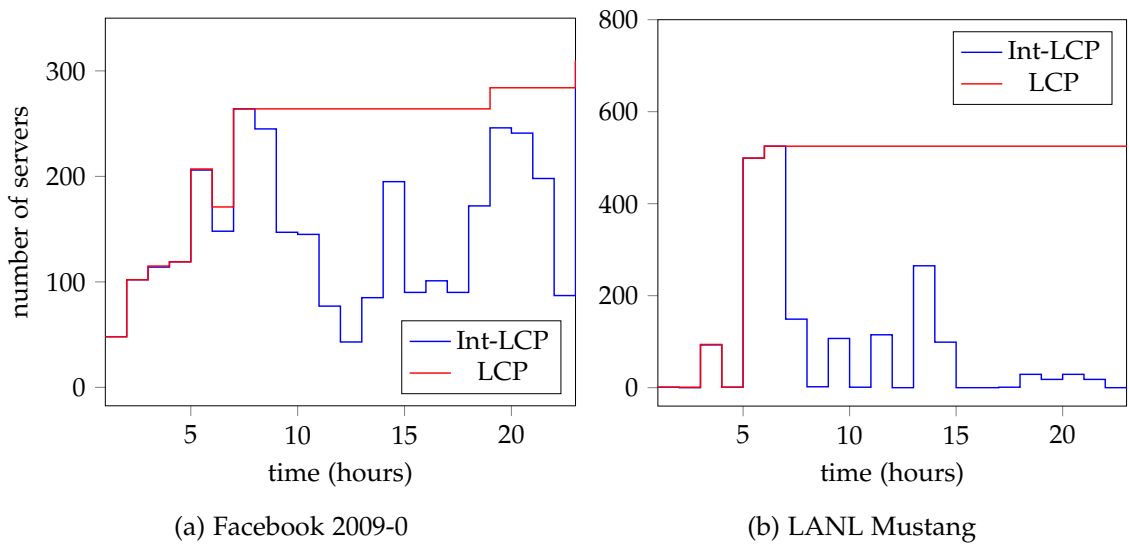

    \begin{subfigure}[b]{.5175\linewidth}
    \resizebox{\textwidth}{!}{\input{thesis/figures/lcp_vs_ilcp_1}}
    \caption{Facebook 2009-0}
    \end{subfigure}
    \begin{subfigure}[b]{.4825\linewidth}
    \resizebox{\textwidth}{!}{\input{thesis/figures/lcp_vs_ilcp_2}}
    \caption{LANL Mustang}
    \end{subfigure}
    \caption{Comparison of the schedules obtained by LCP and Int-LCP under our first model. We observe that LCP is much ``stickier'' than Int-LCP, which is beneficial when the potential cost reduction is small (i.e., for the Facebook 2009-0 trace) but detrimental when the potential cost reduction is large (i.e., for the LANL Mustang trace). In our experiments, we observe that the memoryless algorithm tends to behave similarly to LCP (i.e., is more ``sticky''), whereas the probabilistic algorithms tend to behave similarly to Int-LCP (i.e., are less ``sticky'').}\label{fig:case_studies:ud:lcp_vs_int_lcp}
\end{figure}

\begin{figure}
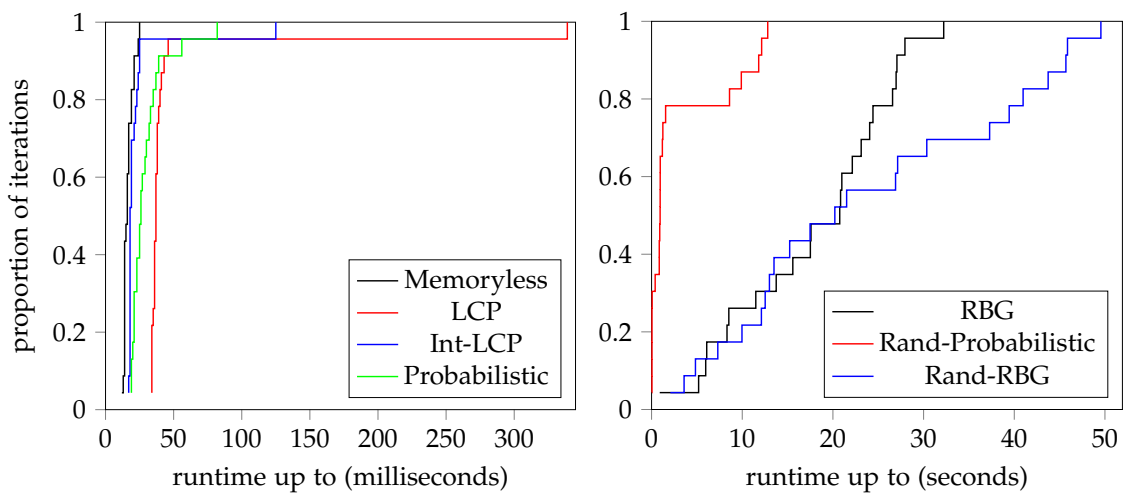

    \begin{subfigure}[b]{.5175\linewidth}
    \resizebox{\textwidth}{!}{\input{thesis/figures/ud_runtimes}}
    \end{subfigure}
    \begin{subfigure}[b]{.4825\linewidth}
    \resizebox{\textwidth}{!}{\input{thesis/figures/ud_runtimes_slow}}
    \end{subfigure}
    \caption{Runtimes of uni-dimensional online algorithms.}\label{fig:case_studies:ud:runtimes}
\end{figure}

\paragraph{Runtime} \cref{fig:case_studies:ud:runtimes} shows the distribution of runtimes (per iteration) of the online algorithms using the Facebook 2009-1 trace. The memoryless algorithm, LCP, and Int-LCP are very fast, even as the number of time slots increases. The runtime of Probabilistic and Rand-Probabilistic is slightly dependent on the used trace and model but generally good. However, when resulting schedules are thight around the upper bound of the decision space, as is the case for the Facebook 2010 trace under our second model, the probabilistic algorithms perform take multiple minutes per iteration. Rand-Probabilistic is significantly slower than Probabilistic as due to the relaxation the integrals need to be computed in a piecewise fashion. The runtime of RBG grows linearly with time and is shown in \cref{fig:case_studies:ud:runtimes} for the first four time slots.

\begin{figure}
    \centering
    \input{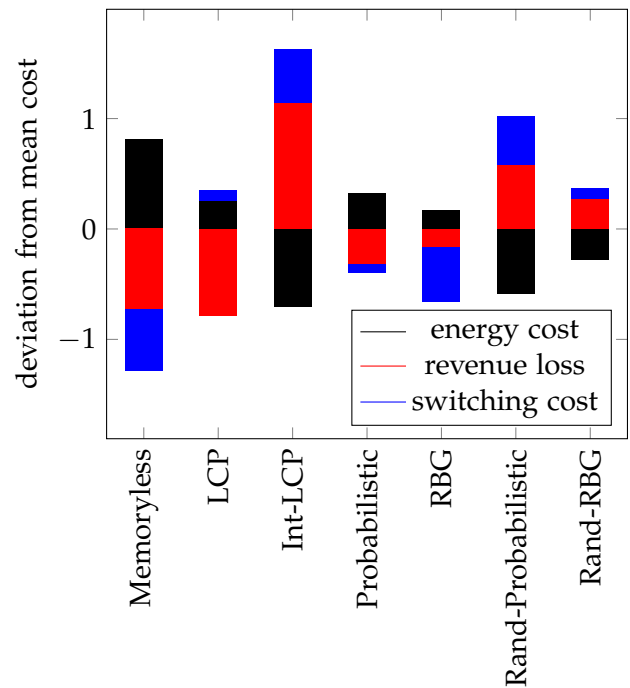}
    \caption{Cost profiles of uni-dimensional online algorithms. Note that integral algorithms seem to prefer revenue loss and switching costs over energy costs, whereas fractional algorithms prefer energy costs. However, this is likely because integral algorithms balance energy costs and revenue loss more accurately than the ceiled schedules of fractional algorithms.}\label{fig:case_studies:ud:costs}
\end{figure}

\paragraph{Cost Makeup} An interesting aspect of the (integral) schedules obtained by the online algorithms is the makeup of their associated costs to understand whether an algorithm systematically prefers some cost over another. We measure this preference of an algorithm as the normalized deviation, i.e., the cost of the algorithm minus the mean cost among all algorithms divided by the standard deviation of costs among all algorithms. We then average the results between all traces. \Cref{fig:case_studies:ud:costs} shows the cost profiles for each algorithm. We observe that fractional algorithms prefer energy cost over revenue loss and switching cost, while integral algorithms prefer revenue loss and switching cost over energy cost. This is likely because fractional algorithms cannot balance energy cost and revenue loss optimally. When fractional schedules are ceiled, this results in an additional energy cost while reducing revenue loss. In absolute terms, the deviations make up less than 1\% of the overall costs of each type.

\paragraph{} We have seen that even in our conservative model, significant cost savings with respect to the optimal static provisioning in hindsight can be achieved in practical settings when the PMR is large enough or the normalized switching cost is less than the typical valley length. Due to the conservative estimates of our model, it is likely that in practice, much more drastic cost savings are possible. For example, when energy costs are higher, or the switching costs are on the order of operating a server in an idle state for one hour rather than four hours.

\section{Multi-Dimensional Algorithms}

Now, we turn to the discussed multi-dimensional algorithms. We begin by analyzing the simplified settings, SLO and SBLO, from \cref{section:online_algorithms:md:lazy_budgeting}. Then, we analyze the gradient-based methods from \cref{section:online_algorithms:md:descent_methods}.

\subsection{Smoothed Load Optimization}

Recall that for SLO, during a single time slot a server can process at most one job. Hence, we cannot use dynamic job durations to model the different runtimes of jobs on servers with two GPUs and servers with eight GPUs. Instead, we use a simplified model based on our second model, which is described in \cref{tab:simp_model}. Note that we disregard revenue loss and that we assume, servers operate at full utilization (if they are active). Overall, we obtain the operating costs $c = (243.720, 219.348)$ and switching costs $\beta = (487.440, 663.672)$.

\begin{table}
    \centering
    \begin{tabularx}{\textwidth}{>{\bfseries}l|X}
        cost & simplified model \\\hline
        operating cost & servers with eight GPUs have $0.9$ times the energy consumption (per processed job) as servers with two GPUs due to improved cooling efficiency \\
        switching cost & servers with eight GPUs have $1.3$ times the switching cost as servers with two GPUs due to an increased associated risk \\
    \caption{Simplified model used in our case studies of SLO and SBLO. The model parameters are based on our second model described in \cref{tab:model}.}
    \end{tabularx}
    \label{tab:simp_model}
\end{table}

The achieved normalized cost and cost reduction of lazy budgeting are shown in \cref{fig:case_studies:md:slo:normalized_cost} and \cref{fig:case_studies:md:slo:cost_reduction}, respectively. The dynamic offline optimal schedule primarily uses 8-GPU-servers and only uses 2-GPU-servers for short periods. The static offline optimal schedule uses 122 8-GPU-servers and no 2-GPU-servers as they have a larger operating cost, which would have to be paid throughout the entire day. The lazy budgeting algorithms primarily use 2-GPU-servers due to their lower switching cost and stick with 8-GPU-servers once they were powered up. \Cref{fig:case_studies:md:slo:det:schedule} and \cref{fig:case_studies:md:slo:rand:schedule} show the schedules obtained by the online algorithms in comparison to the offline optimal schedule. The runtime of the deterministic and randomized variants is shown in \cref{fig:case_studies:md:slo:runtimes}.

\begin{figure}
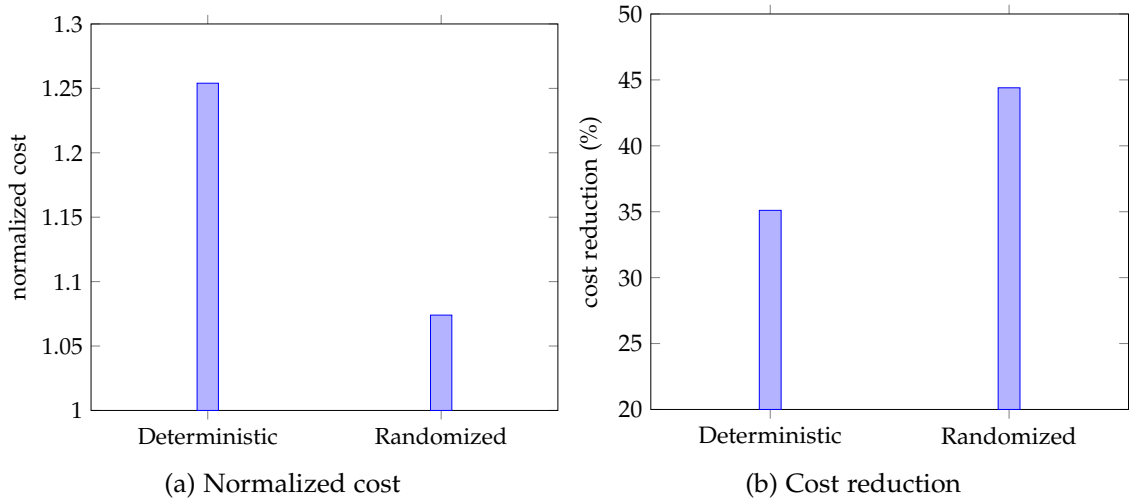

    \begin{subfigure}[b]{.5\linewidth}
    \resizebox{\textwidth}{!}{\input{thesis/figures/slo_normalized_cost}}
    \caption{Normalized cost}\label{fig:case_studies:md:slo:normalized_cost}
    \end{subfigure}
    \begin{subfigure}[b]{.5\linewidth}
    \resizebox{\textwidth}{!}{\input{thesis/figures/slo_cost_reduction}}
    \caption{Cost reduction}\label{fig:case_studies:md:slo:cost_reduction}
    \end{subfigure}
    \caption{Performance of lazy budgeting (SLO) for the Microsoft Fiddle trace when compared against the offline optimum. The results of the randomized algorithm are based on five individual runs.}
\end{figure}

\begin{figure}
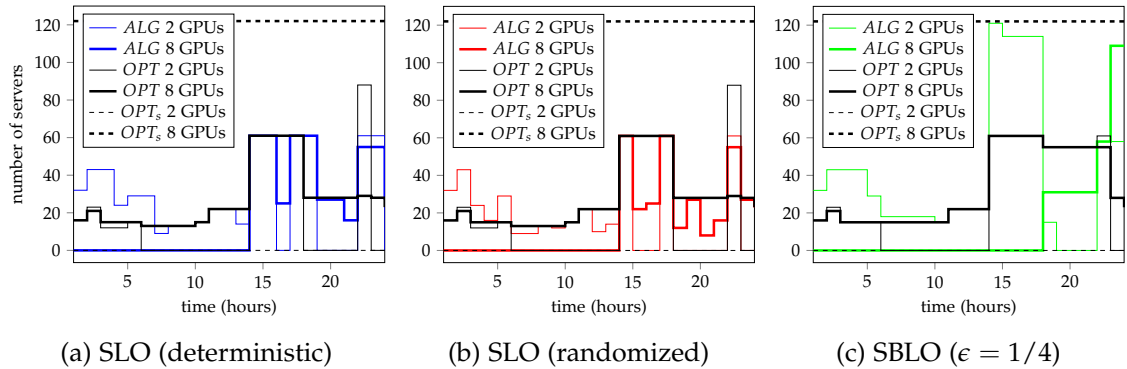

    \begin{subfigure}[b]{.3425\linewidth}
    \resizebox{\textwidth}{!}{\input{thesis/figures/slo_det_schedule}}
    \caption{SLO (deterministic)}\label{fig:case_studies:md:slo:det:schedule}
    \end{subfigure}
    \begin{subfigure}[b]{.32\linewidth}
    \resizebox{\textwidth}{!}{\input{thesis/figures/slo_rand_schedule}}
    \caption{SLO (randomized)}\label{fig:case_studies:md:slo:rand:schedule}
    \end{subfigure}
    \begin{subfigure}[b]{.32\linewidth}
    \resizebox{\textwidth}{!}{\input{thesis/figures/sblo_schedule}}
    \caption{SBLO ($\epsilon = 1/4$)}\label{fig:case_studies:md:sblo:schedule}
    \end{subfigure}
    \caption{Comparison of the schedules obtained by lazy budgeting for the Microsoft Fiddle trace and the offline optimum. For SLO, the deterministic algorithm is shown in blue and one result of the randomized algorithm is shown in red. The lazy budgeting algorithm for SBLO is shown in green. Note that the lazy budgeting algorithms prefer the 2-GPU-servers initially due to their low switching costs. For SLO, the randomized algorithm appears to be less ``sticky'' than the deterministic algorithm, resulting in a better normalized cost.}
\end{figure}

\begin{figure}
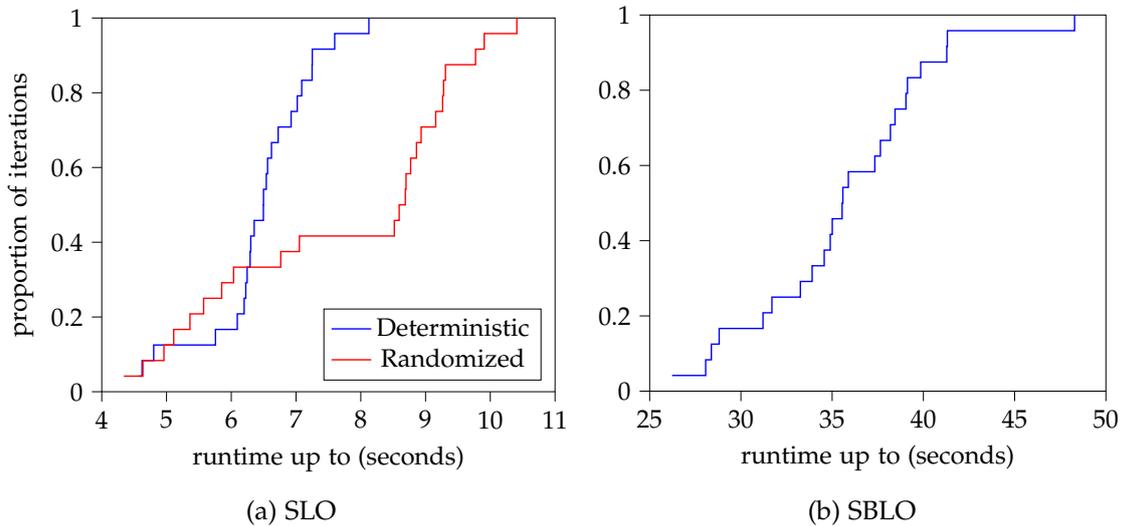

    \begin{subfigure}[b]{.5175\linewidth}
    \resizebox{\textwidth}{!}{\input{thesis/figures/slo_runtimes}}
    \caption{SLO}\label{fig:case_studies:md:slo:runtimes}
    \end{subfigure}
    \begin{subfigure}[b]{.4825\linewidth}
    \resizebox{\textwidth}{!}{\input{thesis/figures/sblo_runtimes}}
    \caption{SBLO}\label{fig:case_studies:md:sblo:runtimes}
    \end{subfigure}
    \caption{Runtime of lazy budgeting algorithms.}
\end{figure}

\subsection{Smoothed Balanced-Load Optimization}

For our analysis of SBLO, we use the same simplified model that we used in our analysis of SLO (see \cref{tab:simp_model}). In particular, we still assume that a server can at most process a single job during a time slot. The dynamic and static offline optimum are similar to those in our analysis of SLO. In particular, the static offline optimum still only uses 8-GPU-servers. \Cref{fig:case_studies:md:sblo:schedule} shows the schedule obtained by lazy budgeting ($\epsilon = 1/4$) in comparison with the offline optimal. Lazy budgeting achieves normalized costs 1.284 and a cost reduction of 11\%. The runtime of the algorithm is shown in \cref{fig:case_studies:md:sblo:runtimes}.

\subsection{Descent Methods}

We also evaluated the performance of P-OBD and D-OBD on the Microsoft Fiddle trace under our original models (see \cref{tab:model}). In our analysis, we use the squared $\ell_2$ norm as the distance-generating function, i.e. $h(x) = \frac{1}{2} \norm{x}_2^2$, which is strongly convex and Lipschitz smooth in the $\ell_2$ norm. In our data center model, we use the $\ell_1$ norm to calculate switching costs, however, we observe that this approximation still achieves a good performance when compared against the dynamic offline optimum. The negative entropy $h(x) = \sum_{k=1}^d x_k \log_2 x_k$, which is commonly used as a distance-generating function for the $\ell_1$ norm cannot be used in the right-sizing data center setting as $\mathbf{0} \in \mathcal{X}$. \Cref{fig:case_studies:md:obd:normalized_cost} and \cref{fig:case_studies:md:obd:cost_reduction} show the achieved normalized cost and cost reduction. The resulting schedules under the first model are compared with the offline optimal in \cref{fig:case_studies:md:obd:schedule}. Remarkably, P-OBD and D-OBD obtain the exact same schedule under our second model. \Cref{fig:case_studies:md:obd:runtimes} visualizes the runtime of P-OBD and D-OBD under our first model.

\begin{figure}
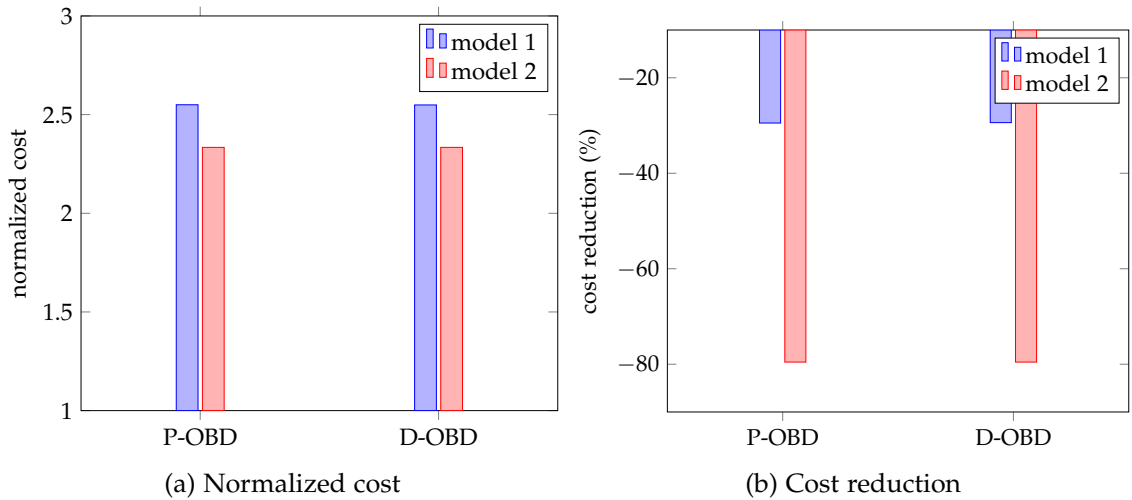

    \begin{subfigure}[b]{.5\linewidth}
    \resizebox{\textwidth}{!}{\input{thesis/figures/obd_normalized_cost}}
    \caption{Normalized cost}\label{fig:case_studies:md:obd:normalized_cost}
    \end{subfigure}
    \begin{subfigure}[b]{.5\linewidth}
    \resizebox{\textwidth}{!}{\input{thesis/figures/obd_cost_reduction}}
    \caption{Cost reduction}\label{fig:case_studies:md:obd:cost_reduction}
    \end{subfigure}
    \caption{Performance of P-OBD  ($\beta = 1/2$) and D-OBD  ($\eta = 1$) for the Microsoft Fiddle trace when compared against the offline optimum. $h(x) = \frac{1}{2} \norm{x}_2^2$ is used as the distance-generating function.}
\end{figure}

\begin{figure}
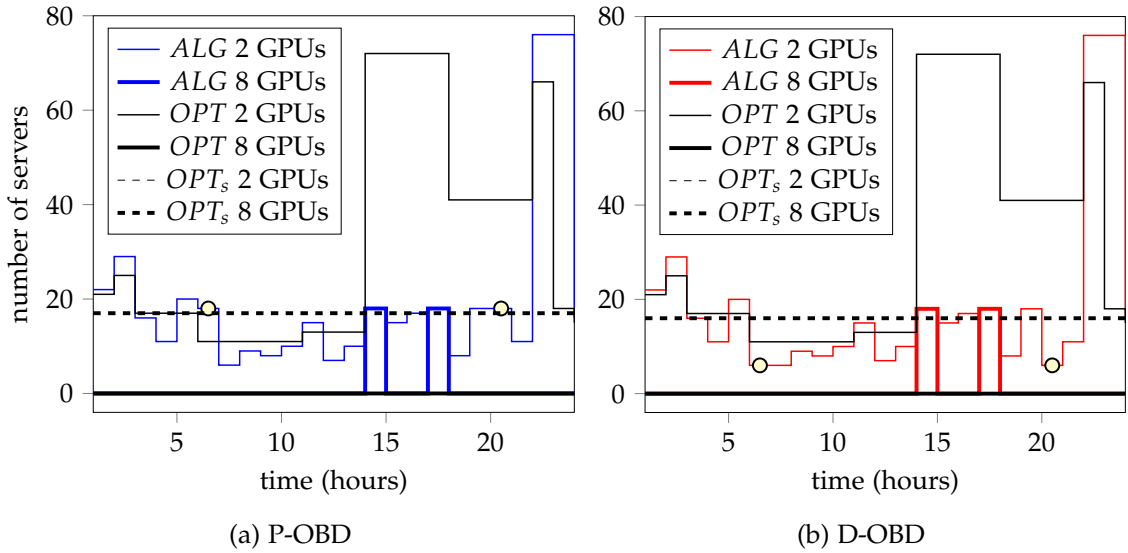

    \begin{subfigure}[b]{.5175\linewidth}
    \resizebox{\textwidth}{!}{\input{thesis/figures/pobd_schedule}}
    \caption{P-OBD}
    \end{subfigure}
    \begin{subfigure}[b]{.4825\linewidth}
    \resizebox{\textwidth}{!}{\input{thesis/figures/dobd_schedule}}
    \caption{D-OBD}
    \end{subfigure}
    \caption{Comparison of the schedules obtained by OBD for the Microsoft Fiddle trace under our first model and the offline optimum. P-OBD ($\beta = 1/2$) is shown in blue and D-OBD ($\eta = 1$) is shown in red. The two time slots during which P-OBD and D-OBD differ are marked in yellow ($t \in \{6, 20\}$). In both time slots, D-OBD is slightly less ``sticky'', resulting in a slightly better normalized cost. Also observe that under our first model, the dynamic offline optimum strictly prefers 2-GPU-servers over 8-GPU-servers In contrast, under our second model, 8-GPU-servers are slightly preferred by the dynamic offline optimum (see \cref{fig:microsoft:schedule}). $h(x) = \frac{1}{2} \norm{x}_2^2$ is used as the distance-generating function.}\label{fig:case_studies:md:obd:schedule}
\end{figure}

Although, OBD incurs an increased cost compared to the static offline optimal (the optimal static choice in hindsight) albeit by a factor less than 1, our very conservative model indicates that in practice, significant cost savings are possible.

\begin{figure}
    \centering
    \input{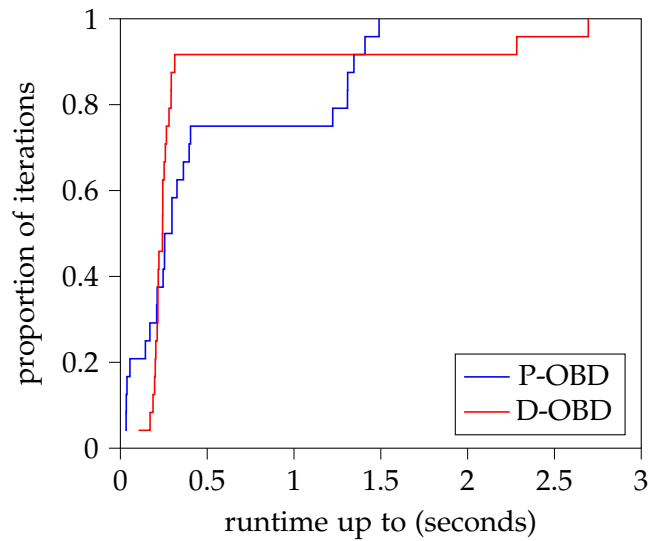}
    \caption{Runtime of OBD algorithms.}\label{fig:case_studies:md:obd:runtimes}
\end{figure}

\section{Predictions}

In our analysis, we use the Alibaba trace under our second model to evaluate the effects of using predictions. We consider two different types of predictions: perfect predictions, and actual predictions that are obtained using Prophet as described in \cref{section:online_algorithms:md:predictions:making_predictions}. The obtained predictions are based on the four preceding days of the Alibaba trace up until the second to last day. \Cref{fig:case_studies:predictions:prediction} visualizes the used prediction for the most common job type (i.e., short jobs). Note that in our analysis, we use the mean predicted load.

\begin{figure}
    \centering
    \input{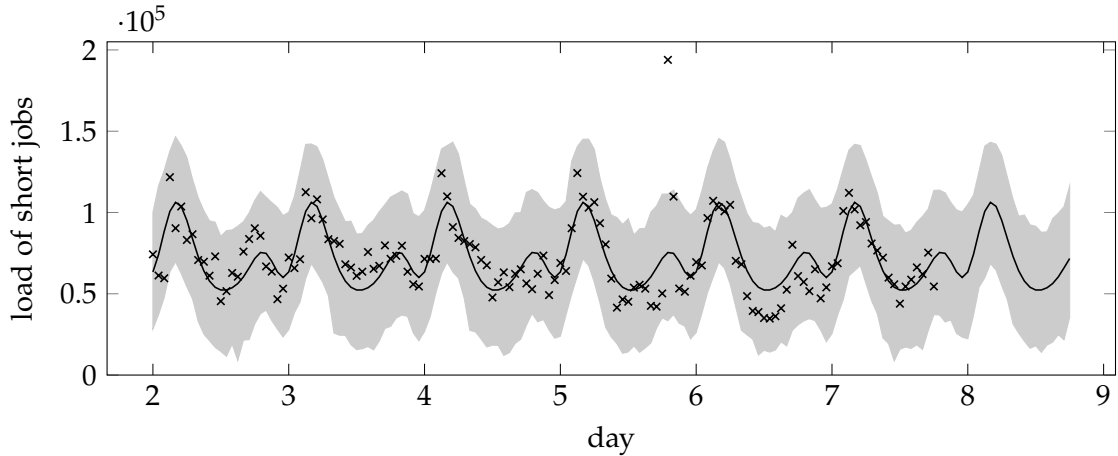}
    \caption{Prediction of the load of short jobs for the second to last day of the Alibaba trace. The mean prediction is shown as the black line. The interquartile range of the predicted distribution is shown as the shaded region. The marks represent actual loads.}\label{fig:case_studies:predictions:prediction}
\end{figure}

\Cref{fig:case_studies:predictions:lcp} shows the effect of the prediction window $w$ when used with LCP for perfect and actual predictions. We observe that in practice, the prediction window can significantly improve the algorithm performance. Additionally, we find that this effect is also achieved with imperfect (i.e., actual) predictions. Previously, \citeauthor{Lin2011}~\cite{Lin2011} only showed this effect for perfect predictions with additive white Gaussian noise.

Interestingly, RHC and AFHC achieve equivalent results for perfect and imperfect predictions. \Cref{fig:case_studies:predictions:mpc} shows how the achieved normalized cost changes with the prediction window. Crucially, note that the MPC-style algorithms do not necessarily perform better for a growing prediction window. \citeauthor{Lin2012}~\cite{Lin2012} showed this effect previously for an adversarially chosen example, however, we observe this behavior with AFHC in a practical setting. In fact, for the Alibaba trace, RHC and AFHC achieve their best result when used without a prediction window, i.e. $w = 0$.

\begin{figure}
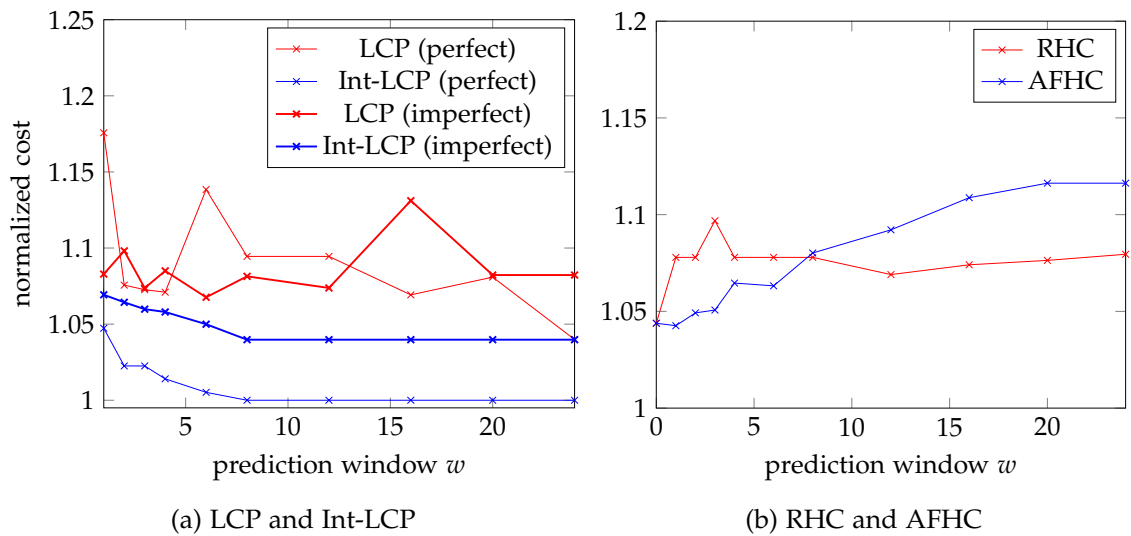

    \begin{subfigure}[b]{.5175\linewidth}
    \resizebox{\textwidth}{!}{\input{thesis/figures/pred_lcp_cost}}
    \caption{LCP and Int-LCP}\label{fig:case_studies:predictions:lcp}
    \end{subfigure}
    \begin{subfigure}[b]{.4825\linewidth}
    \resizebox{\textwidth}{!}{\input{thesis/figures/pred_mpc_cost}}
    \caption{RHC and AFHC}\label{fig:case_studies:predictions:mpc}
    \end{subfigure}
    \caption{Performance of online algorithms with a prediction window for the Alibaba trace. The left figure shows the performance of LCP and Int-LCP. The right figure shows the performance of RHC and AFHC. Note that LCP does not continuously approach the normalized cost $1$ as $w \to 24$ because of numerical inaccuracies solving the convex optimizations and as the obtained is compared to the integral offline optimum rather than the fractional offline optimum. For RHC and AFHC, the achieved normalized cost is independent of whether perfect predictions or the predictions from \cref{fig:case_studies:predictions:prediction} are used.}
\end{figure}

\chapter{Future Work}\label{chapter:future_work}

Our implementation separates the model layer from the problem and algorithm layer. Therefore, testing the empirical performance of the discussed online algorithms in other application areas is a natural and exciting direction for future research. In \cref{section:introduction:related_work}, we have given an overview of some promising applications.

\paragraph{Algorithms for Convex Body Chasing} We have mentioned in \autoref{chapter:introduction} that smoothed convex optimization and convex body chasing are equivalent. Therefore, an interesting research project would be to extend our library of implemented algorithms by the known algorithms for convex body chasing to compare their empirical performance. In particular, the $\mathcal{O}(d)$-competitive algorithm obtained by \citeauthor{Argue2019}~\cite{Argue2019} is highly relevant for the application of right-sizing data centers as it does not impose a restriction on cost functions beyond their convexity.

\paragraph{Dynamic Bounds and Dimensions} In practice, the number of available servers (and even server types) in a data center is likely to change over time. An unexplored area of research is how the discussed approaches for online algorithms can be extended to a setting where the bounds on the decision space $\mathcal{X}$ and the dimension of $\mathcal{X}$ are allowed to change over time. \citeauthor{Albers2021_2}~\cite{Albers2021_2} discuss how their offline algorithm solving the multi-dimensional integral case can be extended to a setting with time-dependent bounds.

\paragraph{Optimal Valley Filling} In \cref{section:case_studies:method:alternatives}, we have discussed the advantages and disadvantages of valley filling, i.e., scheduling low-priority tasks during periods of low loads to reduce the peak-to-mean ratio. Therefore, an interesting problem is finding optimal server configurations and job scheduling times such that the operating costs and switching costs of the data center are balanced with the revenue loss of delaying specific jobs. This problem extends smoothed convex optimization in the data center setting by allowing for incoming loads to be postponed to a later time slot.

\paragraph{Optimal Assignments of Jobs to Servers} Smoothed convex optimization determines the optimal assignment of jobs to a collection of servers of the same type. In \cref{section:application:dispatching}, we have seen that an optimal dispatching rule is to distribute jobs across all servers evenly. However, this approach does not have to be optimal in practice as job arrival times may vary and jobs are discrete, i.e., it may not be possible to distribute them evenly. Therefore, an interesting problem is determining an optimal scheduling of jobs that accounts for the approximations of smoothed convex optimization.

\paragraph{Long-Running Jobs} An vital research problem is to extend the data center model we discussed in \cref{chapter:application} to support jobs with a longer runtime than the length of a time slot. In our case studies from \cref{chapter:case_studies}, we have seen that for a time slot length on the order of tens of minutes to an hour, a significant fraction of jobs (in practical scenarios more than 50\%) fulfill this criterion. This can be achieved by ``memorizing'' which jobs have not been completed yet and ensuring that enough servers of each type are active in the following time slot so that jobs do not have to be rerouted to a server of a different type.

\paragraph{Gang Scheduling Requirement} A similar problem is to extend the model to allow for mutual job requirements. For example, requirements ensuring jobs are processed simultaneously or on servers of the same type.

\paragraph{Without Lookahead} The classical smoothed online convex optimization problem assumes that the convex cost function is known before a move has to be made. This is reasonable to separate the problem of anticipating movement costs with current hitting costs from estimating the current hitting costs. However, in practice, many applications require action before the hitting costs are known. For example, in the setting of right-sizing data centers, enough servers need to be available before jobs arrive as the powering up of servers requires some time. We have seen in \cref{section:online_algorithms:md:predictions} that predictions can be used to get around this restriction. Nonetheless, it would be interesting to see which algorithms perform well in a smoothed setting without lookahead.

\paragraph{Using Predictions} More research is needed to find online algorithms that use predictions robustly and consistently. Here, robustness refers to a bounded competitive ratio when predictions are adversarial, and consistency refers to an improved competitive ratio when predictions are accurate \cite{Li2021}. Note that clearly model predictive control-style algorithms can perform arbitrarily poorly when predictions are adversarial.

\chapter{Conclusion}\label{chapter:conclusion}

This work surveyed numerous offline and online algorithms for fractional and integral smoothed convex optimization in single and high dimensions. We evaluated the performance of online algorithms for the application of dynamically right-sizing data centers using traces from various different sources. We found that the discussed online algorithms perform nearly optimally with respect to the observed normalized cost when compared to the dynamic offline optimum. Moreover, we observed that the normalized cost increases as potential cost savings increase. However, the achieved cost reduction depends primarily on the difference between the dynamic offline optimum and the strategy compared to (in our experiments the static offline optimum). Previous results show that changes in the data center model, such as increasing the energy cost or decreasing the normalized switching cost lead to significant increases in cost savings. We also observed that the difference between fractional and integral solutions is negligible in the application of right-sizing data centers, or generally when total costs are large. Further, we have seen that online algorithms using predictions of future loads can achieve a better normalized cost, even when predictions are non-perfect but obtained using standard methods. However, we have also seen that the achieved normalized cost does not necessarily improve for larger prediction windows.

Finally, we made our implementation of the discussed algorithms available. This implementation can be used to test the performance of algorithms for smoothed online convex optimization in other applications and compare their performance. Our data center model balances energy cost and revenue loss. Using past data, this model can be tuned to be safe (i.e., provide enough computing resources for incoming loads), fulfill service level agreements, and provide significant cost and energy savings in real-world settings.

\appendix{}

\microtypesetup{protrusion=false}

\listoftheorems[ignoreall,show=problem,title={List of Problems}]{}

\listofalgorithms{}
\addcontentsline{toc}{chapter}{List of Algorithms}

\listoftables{}

\printindex

\addchap[Notation]{Notation}\label{chapter:notation}

\section*{Problems}

\nopagebreak\begin{tabularx}{\textwidth}{p{100pt}X}
    $T \in \mathbb{N}$ & number of time slots \\
    $\delta \in \mathbb{R}_{>0}$ & length of a time slot \\
    $d \in \mathbb{N}$ & number of dimensions \\
    $\mathcal{X} \subset \mathbb{R}^d$ & decision space \\
    $m_k \in \mathbb{N}$ & upper bound of dimension $k$ \\
    $M_k \subset \mathbb{N}_0$ & allowed values of dimension $k$ (in the discrete case) \\
    $\mathcal{M} = \mathcal{X}$ & set of all configurations (in the discrete case) \\
    $f_t(x) \in \mathbb{R}_{\geq 0}$ & hitting cost of action $x \in \mathcal{X}$ during time slot $t$ \\
    $\beta_k \in \mathbb{R}_{>0}$ & switching cost of dimension $k$ \\
    $\lambda_t \in \mathbb{N}_0^e$ & load profile during time slot $t$ \\
    $c_k \in \mathbb{R}_{\geq 0}$ & time-independent hitting cost of dimension $k$ \\
\end{tabularx}

\section*{Data Center Model}

\subsection*{Dispatching}

\nopagebreak\begin{tabularx}{\textwidth}{p{100pt}X}
    $m_k \in \mathbb{N}$ & maximum number of servers of type $k$ \\
    $l_k^{max} \in \mathbb{N}$ & maximum number of jobs a server of type $k$ can process in a single time slot \\
    $e \in \mathbb{N}$ & number of load types \\
    $\lambda_{t,i} \in \mathbb{N}_0$ & number of jobs of type $i$ during time slot $t$ \\
    $\lambda_t \in \mathbb{N}_0$ & total load during time slot $t$ \\
    $\lambda^{max} \in \mathbb{N}$ & maximum total load of feasible load profiles \\
    $\mathcal{Z}$ & set of job assignments of jobs of individual load types to server types \\
    $z_{t,k,i} \in [0,1]$ & fraction of jobs of type $i$ assigned to servers of type $k$ during time slot $t$ \\
    $l_{t,k,i} \in [0,\lambda^{max}]$ & (fractional) number of jobs of type $i$ assigned to servers of type $k$ during time slot $t$ \\
    $l_{t,k} \in [0,\lambda^{max}]$ & (fractional) number of jobs assigned to servers of type $k$ during time slot $t$ \\
    $s_k(l) \in [0,1]$ & utilization (or speed) of a server of type $k$ with load $l$ \\
    $\theta_k \in [0,1]$ & maximum allowed utilization of a server of type $k$ \\
    $g_{t,k}(l) \in \mathbb{R}_{\geq 0}$ & operating cost of a server of type $k$ during time slot $t$ when $l$ jobs are processed on the server \\
    $q_{t,k,i}(l) \in \mathbb{R}_{\geq 0}$ & cost of processing a job of type $i$ on a server of type $k$ during time slot $t$ when $l$ jobs are processed on the server \\
\end{tabularx}

\subsection*{Energy Cost}

\nopagebreak\begin{tabularx}{\textwidth}{p{100pt}X}
    $e_{t,k}(s) \in \mathbb{R}_{\geq 0}$ & energy cost of a server of type $k$ with utilization $s$ during time slot $t$ \\
    $\nu_{t,k}(p) \in \mathbb{R}_{\geq 0}$ & energy cost of a server of type $k$ with energy consumption $p$ during time slot $t$ \\
    $\phi_k(s) \in \mathbb{R}_{\geq 0}$ & energy consumption of a server of type $k$ with utilization $s$ \\
    $\Phi_k(s) \in \mathbb{R}_{\geq 0}$ & power consumption of a server of type $k$ with utilization $s$ \\
    $\Phi_k^{max} \in \mathbb{R}_{\geq 0}$ & power consumption of a server of type $k$ on full load \\
    $\Phi_k^{min} \in \mathbb{R}_{\geq 0}$ & power consumption of an idling server of type $k$ \\
    $c_{t,i} \in \mathbb{R}_{\geq 0}$ & energy cost of energy source $i$ per unit of energy during time slot $t$ \\
\end{tabularx}

\subsection*{Revenue Loss}

\nopagebreak\begin{tabularx}{\textwidth}{p{100pt}X}
    $\gamma \in \mathbb{R}_{\geq 0}$ & revenue loss factor, i.e., lost revenue per unit of delay \\
    $r_{t,i}(d) \in \mathbb{R}_{\geq 0}$ & revenue loss of jobs of type $i$ with an average delay $d$ during time slot $t$ \\
    $d_{t,k,i}(l) \in \mathbb{R}_{\geq 0}$ & average delay of a job of type $i$ processed on a server of type $k$ during time slot $t$ where the total load on the server is $l$ \\
    $\delta_{t,k,i} \in \mathbb{R}_{\geq 0}$ & constant delay when processing a job of type $i$ on a server of type $k$ during time slot $t$ \\
    $\mu_k \in \mathbb{R}_{\geq 0}$ & service rate of a server of type $k$ \\
    $\delta_i \in \mathbb{R}_{\geq 0}$ & minimal detectable delay of jobs of type $i$ \\
    $\eta_{k,i} \in \mathbb{R}_{\geq 0}$ & processing time of a job of type $i$ on a server of type $k$ \\
\end{tabularx}

\subsection*{Switching Cost}

\nopagebreak\begin{tabularx}{\textwidth}{p{100pt}X}
    $\epsilon_k \in \mathbb{R}_{\geq 0}$ & additional energy consumed by toggling a server of type $k$ on and off \\
    $\delta_k \in \mathbb{R}_{\geq 0}$ & delay in migrating connections or data of a server of type $k$ before it can be powered down \\
    $\tau_k \in \mathbb{R}_{\geq 0}$ & wear-and-tear costs of toggling a server \\
    $\rho_k \in \mathbb{R}_{\geq 0}$ & perceived risk associated with toggling a server of type $k$ \\
    $\xi_k \in \mathbb{R}_{>0}$ & normalized switching cost measuring the minimum duration a server of type $k$ must be asleep to outweigh the switching cost \\
\end{tabularx}

\subsection*{Networks}

\nopagebreak\begin{tabularx}{\textwidth}{p{100pt}X}
    $\iota \in \mathbb{N}$ & number of data centers \\
    $\zeta \in \mathbb{N}$ & number of geographically centered request sources \\
    $\delta_{t,j,s} \in \mathbb{R}_{\geq 0}$ & network delay incurred by routing a request from source $s$ to data center $j$ during time slot $t$ \\
    $\xi$ & number of energy sources \\
    $p_{t,i,j} \in \mathbb{R}_{\geq 0} \cup \{\infty\}$ & energy from source $i$ available at data center $j$ during time slot $t$ \\
    $u_{t,i} \in \mathbb{R}_{\geq 0}$ & average profit per unit of energy from source $i$ during time slot $t$ \\
    $q_{t,i} \in [0,1]$ & minimum fraction for energy from source $i$ during time slot $t$ \\
    $\delta_{t,i,j} \in \mathbb{R}_{\geq 0}$ & remaining energy requirement of data center $j$ during time slot $t$ after all energy sources up to source $i$ were used \\
\end{tabularx}

\section*{Complexity}

\nopagebreak\begin{tabularx}{\textwidth}{p{100pt}X}
    $\mathcal{O}(C)$ & complexity of computing the hitting costs $f_t$ \\
    $\mathcal{O}(O_{\epsilon}^d)$ & convergence rate of a convex optimization finding an $\epsilon$-optimal solution in $d$ dimensions \\
    $\mathcal{O}(R_{\epsilon})$ & convergence rate of a root finding method finding an $\epsilon$-optimal root \\
    $\mathcal{O}(I_{\epsilon})$ & convergence rate of a quadrature method finding an $\epsilon$-optimal integral \\
\end{tabularx}

\section*{Algorithms}

\nopagebreak\begin{tabularx}{\textwidth}{p{100pt}X}
    $\tau$ & current time slot \\
    $\hat{x}$ & optimal value of $x$ with respect to some optimization \\
    $X \in \mathcal{X}^t$ & schedule, i.e., sequence of configurations over the time horizon $t$ \\
    $x, y \in \mathcal{X}$ & configuration \\
    $i, j \in [m_k]_0$ & value of dimension $k$ \\
    $x_{k \gets j}$ & configuration $x$ after updating dimension $k$ to $j$ \\
\end{tabularx}

\section*{Implementation}

\nopagebreak\begin{tabularx}{\textwidth}{p{100pt}X}
    $\chi_{\tau}(t, x) \in \bigcup_{n=1}^{\infty} \mathbb{R}_{\geq 0}^n$ & hitting cost of action $x \in \mathcal{S}$ at time slot $t \in [T]$ given the information from time slot $\tau$ \\
    $\mathcal{P}_t \in \left(\bigcup_{n=1}^{\infty} \mathbb{N}_0^n\right)^e$ & predicted load profile for time slot $t$, i.e., vector of predicted loads for each job type\\
\end{tabularx}

\section*{Miscellaneous}

\nopagebreak\begin{tabularx}{\textwidth}{p{100pt}X}
    $[n] := \{1, \dots, n\}$ & range of natural numbers from $1$ to $n$ \\
    $[n]_0 := \{0\} \cup [n]$ & range of integers from $0$ to $n$ \\
    $[a : b] :=\newline \{a, a+1, \dots, b\}$ & range of integers from $a$ to $b$ \\
    $(x)_a^b :=\newline \max\{a, \min\{b, x\}\}$ & uni-dimensional projection of $x \in \mathbb{R}$ onto $[a,b]$ \\
    $(x)^+ := \max\{0, x\}$ & uni-dimensional projection of $x \in \mathbb{R}$ onto $[0, \infty)$ \\
    $x_{\text{min}}$ & smallest entry of the vector $x$ \\
    $x_{\text{max}}$ & largest entry of the vector $x$ \\
\end{tabularx}

\begin{landscape}
\chapter{Taxonomy of Online Algorithms}\label{appendix:taxonomy}

Here, we give an overview of the algorithms that were discussed in this chapter. The first table includes online algorithms for the one dimensional setting. The latter table includes online algorithms for the general setting.

\begin{table}[!ht]
    \centering
    \begin{tabularx}{\textheight}{l|l|X|l}
        name & problem & performance & time complexity \\\hline
        \hyperref[alg:ud:lcp]{LCP} & \hyperref[problem:simplified_smoothed_convex_optimization]{SSCO} & 3-competitive & $\mathcal{O}((\tau + w) C O_{\epsilon}^{\tau + w})$ \\
        \hyperref[alg:ud:lcp]{Int-LCP} & \hyperref[problem:simplified_smoothed_convex_optimization]{Int-SSCO} & 3-competitive & $\mathcal{O}((\tau + w)^2 C \log_2 m)$ \\
        \hyperref[alg:ud:memoryless]{Memoryless} & \hyperref[problem:simplified_smoothed_convex_optimization]{SSCO} & 3-competitive & $\mathcal{O}(C O_{\epsilon}^1)$ \\
        \hyperref[alg:ud:probabilistic]{Probabilistic} & \hyperref[problem:simplified_smoothed_convex_optimization]{SSCO} & 2-competitive & $\mathcal{O}(\tau^2 C I_{\epsilon} |B_{f_0}| R_{\epsilon} O_{\epsilon}^1)$ \\
        \hyperref[alg:ud:rbg]{RBG} & \hyperref[problem:smoothed_convex_optimization]{SCO} & $\alpha$-unfair competitive ratio $(1+\theta) / \min \{\theta, \alpha\}$ and regret $\mathcal{O}(\max \{T / \theta, \theta\})$; 2-competitive for $\theta = 1$ & $\mathcal{O}(C (O_{\epsilon}^1)^{\tau+1})$ \\
        \hyperref[alg:ud:randomized]{Rand-Probabilistic} & \hyperref[problem:simplified_smoothed_convex_optimization]{Int-SSCO} & 2-competitive & $\mathcal{O}(\tau^2 m C I_{\epsilon} R_{\epsilon} O_{\epsilon}^1)$ \\
        \hyperref[alg:ud:randomized]{Rand-RBG} & \hyperref[problem:smoothed_convex_optimization]{Int-SCO} & 2-competitive & $\mathcal{O}(C (O_{\epsilon}^1)^{\tau+1})$ \\
    \caption{Taxonomy of uni-dimensional online algorithms.}
    \end{tabularx}
\end{table}

\begin{table}[!ht]
    \centering
    \begin{tabularx}{\textheight}{l|l|X|l}
        name & problem & performance & time complexity \\\hline
        \hyperref[alg:md:lazy_budgeting:det_slo]{LB} & \hyperref[problem:slo]{SLO} & $2d$-competitive & $\mathcal{O}(m d^2 + C d \prod_{k=1}^d m_k)$ \\
        \hyperref[alg:md:lazy_budgeting:det_slo]{Randomized LB} & \hyperref[problem:slo]{SLO} & $\frac{e}{e-1}d$-competitive & $\mathcal{O}(m d^2 + C d \prod_{k=1}^d m_k)$ \\
        \hyperref[alg:md:lazy_budgeting:sblo_c]{LB} & \hyperref[problem:sblo]{SBLO} & $(2d + 1 + \epsilon)$-competitive & $\mathcal{O}(\widetilde{n}_{\tau} \widetilde{\tau}^2 |\mathcal{M}| C d)$ \\
        \hyperref[alg:md:ogd]{OGD} & \hyperref[problem:smoothed_convex_optimization]{SCO} & $\mathcal{O}(\sqrt{T})$-regret & $\mathcal{O}(d C O_{\epsilon}^d)$ \\
        \hyperref[alg:md:obd]{OBD} & \hyperref[problem:smoothed_convex_optimization]{SCO} & - & $\mathcal{O}(O_{\epsilon}^d)$ \\
        \hyperref[alg:md:pobd]{P-OBD} & \hyperref[problem:smoothed_convex_optimization]{SCO} & $3 + \mathcal{O}(1 / \alpha)$-competitive & $\mathcal{O}(C O_{\epsilon}^d + O_{\epsilon}^d R_{\epsilon})$ \\
        \hyperref[alg:md:dobd]{D-OBD} & \hyperref[problem:smoothed_convex_optimization]{SCO} & $\mathcal{O}(\sqrt{T})$-regret & $\mathcal{O}(C O_{\epsilon}^d + (O_{\epsilon}^d)^2 R_{\epsilon})$ \\
        \hyperref[alg:predictions:rhc]{RHC} & \hyperref[problem:smoothed_convex_optimization]{SCO} & $(1 + \mathcal{O}(1/w))$-competitive in one dimension; otherwise $(1 + \max_{k \in [d]} \beta_k / e_k(0))$-competitive & $\mathcal{O}(C O_{\epsilon}^{dw})$ \\
        \hyperref[alg:predictions:afhc]{AFHC} & \hyperref[problem:smoothed_convex_optimization]{SCO} & $(1 + \max_{k \in [d]} \frac{\beta_k}{(w+1) e_k(0)})$-competitive, i.e., $(1 + \mathcal{O}(1/w))$-competitive & $\mathcal{O}(w C O_{\epsilon}^{dw})$ \\
    \caption{Taxonomy of multi-dimensional online algorithms.}
    \end{tabularx}
\end{table}
\end{landscape}

\microtypesetup{protrusion=true}

\printbibliography{}

\end{document}